\documentclass[11pt]{amsart}%
\usepackage{comment}
\usepackage{color}
\usepackage{amsmath}
\usepackage{amsfonts}
\usepackage{amssymb}
\usepackage{graphicx}
\usepackage{hyperref}%
\providecommand{\U}[1]{\protect\rule{.1in}{.1in}}
\providecommand{\U}[1]{\protect\rule{.1in}{.1in}}
\def \pt{/\!/}

\def\enddoc{\end{document}}

\def\Xint#1{\mathchoice
	{\XXint\displaystyle\textstyle{#1}}%
	{\XXint\textstyle\scriptstyle{#1}}%
	{\XXint\scriptstyle\scriptscriptstyle{#1}}%
	{\XXint\scriptscriptstyle\scriptscriptstyle{#1}}%
	\!\int}
\def\XXint#1#2#3{{\setbox0=\hbox{$#1{#2#3}{\int}$}
	\vcenter{\hbox{$#2#3$}}\kern-.5\wd0}}
\def\dashint{\Xint-}  	

\newtheorem{theorem}{Theorem}[section]
\newtheorem{metatheorem}[theorem]{Meta-Theorem}
\newtheorem{metacorollary}[theorem]{Meta-Corollary}
\newtheorem{metapropositions}[theorem]{Meta-Proposition}
\newtheorem{metalemma}[theorem]{Meta-Lemma}

\newtheorem{corollary}[theorem]{Corollary}

\newtheorem{definition}[theorem]{Definition}
\newtheorem{example}[theorem]{Example}

\newtheorem{lemma}[theorem]{Lemma}
\newtheorem{notation}[theorem]{Notation}

\newtheorem{proposition}[theorem]{Proposition}
\newtheorem{remark}[theorem]{Remark}

\newtheorem{assumption}{Assumption}

\newcommand{\GraphicsDirectory}{./graphics/}
\graphicspath{{\GraphicsDirectory}}

\def\svgwidth{2in}
\newcommand{\psize}[1]{\def\svgwidth{#1}}

\numberwithin{equation}{section}
\excludecomment{Aside}
\begin{document}
\title{A Functional Integral Approaches to the Makeenko-Migdal Equations}
\author{Bruce K. Driver}
\address{Department of Mathematics \\
University of California, San Diego \\
La Jolla, CA 92093-0112 \\
\texttt{bdriver@math.ucsd.edu}}
\maketitle

\begin{abstract}
Makeenko and Migdal (1979) gave heuristic identities involving the expectation
of the product of two Wilson loop functionals associated to splitting a single
loop at a self-intersection point. Kazakov and I. K. Kostov (1980)
reformulated the Makeenko--Migdal equations in the plane case into a form
which made rigorous sense. Nevertheless, the first rigorous proof of these
equations (and their generalizations) was not given until the fundamental
paper of T. L\'{e}vy (2011). Subsequently Driver, Kemp, and Hall (2017) gave a
simplified proof of L\'{e}vy's result and then with F. Gabriel (2017) we
showed that these simplified proofs extend to the Yang-Mills measure over
arbitrary compact surfaces. All of the proofs to date are elementary but
tricky exercises in finite dimensional integration by parts. The goal of this
article is to give a rigorous functional integral proof of the
Makeenko--Migdal equations guided by the original heuristic machinery invented
by Makeenko and Migdal. Although this stochastic proof is technically more
difficult, it is conceptually clearer and explains \textquotedblleft
why\textquotedblright\ the Makeenko--Migdal equations are true. It is hoped
that this paper will also serve as an introduction to some of the problems
involved in making sense of quantizing Yang-Mill's fields.

\end{abstract}
\tableofcontents

\today\ \emph{File:\jobname{.tex} }

\section{Introduction\label{sec.1}}

Let $K$ be any compact Lie group with $\mathfrak{k}=\operatorname*{Lie}\left(
K\right)  $ being its Lie algebra which is assumed to be equipped with an
$\mathrm{Ad}_{K}$ -- invariant inner product denoted by $\left\langle
\cdot,\cdot\right\rangle _{\mathfrak{k}}$ or often by $\left\langle
\cdot,\cdot\right\rangle .$ [For notational simplicity (and without loss of
generality) we will assume that $K$ is a closed matrix Lie-subgroup of
$\mathbb{C}^{D\times D}$ -- the space of $D\times D$ complex matrices for some
$D\in\mathbb{N}.$] Further suppose that $\left(  M,g,o\right)  $ is a pointed
$d$ -- dimensional Riemannian manifold and $\operatorname*{Vol}_{g}$ is the
Riemannian volume measure on $M.$ [We will soon specialize to the case where
$d=2,$ $M=\mathbb{R}^{2},$ $o=0,$ and $g$ is the usual Euclidean metric in
which case $\operatorname*{Vol}_{g}$ is Lebesgue measure $\left(  m\right)  $
on $\mathbb{R}^{2}.$] Throughout the paper we write $\dot{k}\left(  t\right)
$ for $dk\left(  t\right)  /dt$ and $h^{\prime}\left(  s\right)  $ for
$dh\left(  s\right)  /ds,$ i.e. upper-dot and prime stand for $t$ and
$s$-derivatives respectively.

\begin{notation}
\label{not.1.1}Let $\mathcal{A}:=\Omega^{1}\left(  M,\mathfrak{k}\right)  $ be
the space of $\mathfrak{k}$ -- valued \textbf{connection one-forms} on $M,$
$\mathcal{G}$ be the\textbf{ gauge group} consisting of functions
$g:M\rightarrow K$ and $\mathcal{G}_{o}$ be the \textbf{restricted gauge group
}defined by $\mathcal{G}_{o}=\left\{  g\in\mathcal{G}:g\left(  o\right)
=I\right\}  .$ The smooth gauge group, $\mathcal{G},$ acts (as a right action)
on $\mathcal{A}$ via,%
\begin{equation}
g\rightarrow A^{g}:=g^{-1}Ag+g^{-1}dg\text{ for all }g\in\mathcal{G}.
\label{e.1.1}%
\end{equation}
where (locally) $A=\sum_{i=1}^{d}A_{i}dx_{i}\ $with $A_{i}$ being locally
defined $\mathfrak{k}$-valued functions on $M.$
\end{notation}

\begin{definition}
[Covariant differentiation and parallel translation]\label{def.1.2}Let
$A\in\mathcal{A}$ be a $\mathfrak{k}$-valued connection one form on manifold
$M$ and for a curve $\ell:\left[  a,b\right]  \rightarrow M$ which is
absolutely continuous and a differentiable function $k:\left[  a,b\right]
\rightarrow K,$ let $\nabla k\left(  t\right)  $ denote the \textbf{covariant
differential of }$k$ defined by,%
\[
\nabla_{t}^{A}k\left(  t\right)  :=\frac{d}{dt}k\left(  t\right)  +A\left(
\dot{\ell}\left(  t\right)  \right)  k\left(  t\right)  .
\]
Also let $\pt_{t}^{A}\left(  \ell\right)  \in K$ be \textbf{parallel
translation} along $\ell$ defined as the solution to the ODE (or more
precisely its related integral equation),%
\[
\nabla_{t}^{A}\pt_{t}^{A}\left(  \ell\right)  =\frac{d}{dt}\pt_{t}^{A}\left(
\ell\right)  +A\left(  \dot{\ell}\left(  t\right)  \right)  \pt_{t}^{A}\left(
\ell\right)  =0\text{ with }\pt_{a}^{A}\left(  \ell\right)  =I\in K.
\]
We typically write $\pt^{A}\left(  \ell\right)  $ for $\pt_{b}^{A}\left(
\ell\right)  .$
\end{definition}

\begin{definition}
[Curvature]\label{def.1.3}The \textbf{curvature two form }of $A\in\mathcal{A}$
is $F^{A}=dA+A\wedge A\in\Omega^{2}\left(  M,\mathfrak{k}\right)  ,$ i.e. for
all $p\in M$ and $v,w\in T_{p}M,$%
\[
F^{A}\left(  v,w\right)  =dA\left(  v,w\right)  +\left[  A\left(  v\right)
,A\left(  w\right)  \right]  _{\mathfrak{k}}.
\]

\end{definition}

See Theorem \ref{thm.A.1} of Appendix \ref{sec.A} to see how the gauge group
acts on $\pt_{t}^{A}\left(  \ell\right)  $ and $F^{A}.$

\begin{definition}
[Yang-Mills Energy]\label{def.1.4}The \textbf{Yang-Mills energy }associated to
$A\in\mathcal{A}$ is defined by,%
\[
\left\Vert F^{A}\right\Vert ^{2}=\int_{M}\left\vert F^{A}\right\vert
^{2}\left(  x\right)  d\operatorname*{Vol}\nolimits_{g}\left(  x\right)
\]
where, for $p\in M$ and any orthonormal basis, $\left\{  e_{i}\right\}
_{i=1}^{d},$ of $T_{p}M,$%
\[
\left\vert F^{A}\right\vert ^{2}\left(  p\right)  =\sum_{1\leq i<j\leq
d}\left\vert F^{A}\left\langle e_{i},e_{j}\right\rangle \right\vert
_{\mathfrak{k}}^{2}.
\]

\end{definition}

Now suppose that $M=\mathbb{R}^{d}$ for some $d\in\mathbb{N}$ and let $m$
denote Lebesgue measure on $\mathbb{R}^{d}.$ For $u,v\in L^{2}\left(
\mathbb{R}^{d},m;\mathfrak{k}\right)  ,$ let
\begin{align}
\left\langle u,v\right\rangle  &  :=\int_{\mathbb{R}^{d}}\left\langle u\left(
\cdot\right)  ,v\left(  \cdot\right)  \right\rangle _{\mathfrak{k}}dm\text{
and }\label{e.1.2}\\
\left\Vert u\right\Vert ^{2}  &  :=\left\langle u,u\right\rangle
=\int_{\mathbb{R}^{d}}\left\vert u\right\vert _{\mathfrak{k}}^{2}dm.
\label{e.1.3}%
\end{align}
The informal Euclidean Yang-Mills' \textquotedblleft measure\textquotedblright%
\ on $\mathcal{A}$ is the expression
\begin{equation}
d\mu_{\text{YM}_{d}}\left(  A\right)  =\frac{1}{Z}\exp\left(  -\frac{1}%
{2}\left\Vert F^{A}\right\Vert ^{2}\right)  \mathcal{D}A, \label{e.1.4}%
\end{equation}
where $\mathcal{D}A$ is a (non-existent) Lebesgue measure on $\mathcal{A}$ and
$Z$ is a \textquotedblleft normalizing constant.\textquotedblright\ Since
$F^{A^{g}}=\mathrm{Ad}_{g^{-1}}F^{A}$ (see Theorem \ref{thm.A.1} below) and
the $\left\langle \cdot,\cdot\right\rangle _{\mathfrak{k}}$ is an
$\mathrm{Ad}_{K}$-invariant inner product, the functional, $A\rightarrow
\left\Vert F^{A}\right\Vert ^{2}$, is invariant under the right action of
$\mathcal{G}$ on $\mathcal{A}$ given in Eq. (\ref{e.1.1}). As $\mathcal{G}$ is
an infinite dimensional \textquotedblleft infinite volume\textquotedblright%
\ group, it is not possible, even at this informal level, to interpret the
expression on the right side of Eq. (\ref{e.1.4}) as a probability measure on
$\mathcal{A}.$ As is customary, one should try to interpret $\mu_{\text{YM}}$
as a measure on the quotient space $\mathcal{A}/\mathcal{G}_{0}.$ To be more
concrete one usually tries to \textquotedblleft define\textquotedblright%
\ $d\mu_{\text{YM}}$ by appropriately restricting $A$ in Eq. (\ref{e.1.4}) to
be in a slice, $\mathcal{A}_{0}\subset\mathcal{A},$ of the gauge group action.
We will discuss \textquotedblleft gauge fixing\textquotedblright\ in more
detail in Appendix \ref{sec.B} below.

For the purposes of this paper we now specialize to $d=2$ so that
$M=\mathbb{R}^{2}$ equipped with its standard coordinates, $\left(
x,y\right)  .$ We will further take $\mathcal{A}_{0}$ to be the connection one
forms in the so called \textquotedblleft complete axial\textquotedblright\ gauge.

\begin{notation}
[Complete axial gauge]\label{n.1.5}Let $\mathcal{A}_{0}$ denote the subspace
functions $A:\mathbb{R}^{2}\rightarrow\mathfrak{k}$ such that $A\left(
x,0\right)  =0$ for all $x\in\mathbb{R}.$ We will identify $A\in
\mathcal{A}_{0}$ with the connection one form $A\,dx$ and refer to
$A\in\mathcal{A}_{0}$ as a connection one form in the \textbf{complete axial
gauge.}
\end{notation}

If $A:\mathbb{R}^{2}\rightarrow\mathfrak{k}$ ($A\left(  x,0\right)  $ need not
be zero yet), then the curvature $2$-form $A\,dx$ is given simply by
$F^{A\,dx}=f^{A}dx\wedge dy$ where $f^{A}:=-\partial_{y}A$ is the
\textbf{curvature function} associated to $A.$ Conversely given a (continuous
say) function, $f:\mathbb{R}^{2}\rightarrow\mathfrak{k},$ we may define $A\in
A_{0}$ by
\begin{equation}
A\left(  x,y\right)  =-\left[  \int_{0}^{y}f\left(  x,y^{\prime}\right)
dy^{\prime}\right]  . \label{e.1.5}%
\end{equation}
Thus Eq. (\ref{e.1.5}) sets up a linear isomorphism between $\mathfrak{k}%
$-valued curvature functions $\left(  f\right)  $ on $\mathbb{R}^{2}$ and
elements of $\mathcal{A}_{0}$ with the inverse operation given by
$A\rightarrow f=-\partial_{y}A.$

Restricting the expression in Eq. (\ref{e.1.4}) to $\mathcal{A}_{0}$ leads us
to consider the informally defined \textquotedblleft
measure,\textquotedblright%
\begin{equation}
d\mu\left(  A\right)  =\text{\textquotedblleft}\frac{1}{Z}\exp\left(
-\frac{1}{2}\left\Vert -\partial_{y}A\right\Vert ^{2}\right)  \mathcal{D}%
A,\text{\textquotedblright} \label{e.1.6}%
\end{equation}
where again $\mathcal{D}A$ represents an ill defined Lebesgue measure on
$\mathcal{A}_{0}$ and $Z$ is chosen to make $\mu$ a probability measure.
[Justification for using Eq. (\ref{e.1.6}) is provided in Appendix \ref{sec.B}
below.] With this notation and formally using standard finite dimensional
Gaussian integral formulas in this infinite dimensional setting, we should
have
\begin{equation}
\mathbb{E}\left[  e^{i\left\langle f,u\right\rangle }\right]  =\frac{1}{Z}%
\int_{\mathcal{A}_{0}}e^{i\left\langle -\partial_{y}A,u\right\rangle }%
\exp\left(  -\frac{1}{2}\left\Vert \partial_{y}A\right\Vert ^{2}\right)
\mathcal{D}A=e^{-\frac{1}{2}\left\Vert u\right\Vert ^{2}}, \label{e.1.7}%
\end{equation}
where, as usual, we use \textquotedblleft$\mathbb{E}$\textquotedblright\ to
denote integration relative to a probability measure. In other words, for all
$u\in L^{2}\left(  \mathbb{R}^{2},m;\mathfrak{k}\right)  ,$ we should
interpret $\left\langle f,u\right\rangle $ to be a mean zero Gaussian random
variable with variance equal to $\left\Vert u\right\Vert ^{2}.$

Starting with either Eqs. (\ref{e.1.6}) or Eq. (\ref{e.1.7}) along with Eq.
(\ref{e.1.5}), we ultimately want to understand expectations of gauge
invariant functions, $\Psi:\mathcal{A}\rightarrow\mathbb{C}.$ The typical
example of such functions are the so called \textbf{Wilson functionals, }i.e.
functions on $\mathcal{A}$ of the form $\Psi\left(  A\right)  :=U\left(
\left\{  \pt^{A}\left(  \sigma_{i}\right)  \right\}  _{i=1}^{N}\right)  $
where $U:K^{N}\rightarrow\mathbb{C}$ and $\left\{  \sigma_{i}\right\}
_{i=1}^{N}$ is a collection of paths in $\mathbb{R}^{2}$ all chosen so that
$\Psi$ is gauge invariant. The prototypical examples of Wilson functionals is
to take $\Psi$ to be a function of \textbf{Wilson loop variables, }i.e.
$\Psi\left(  A\right)  :=U\left(  \left\{  \operatorname{tr}\left[
\pt^{A}\left(  \sigma_{i}\right)  \right]  \right\}  _{i=1}^{N}\right)
$\textbf{ }where \textquotedblleft$\operatorname{tr}$\textquotedblright\ is
the matrix trace and each $\sigma_{i}$ is now assumed be a loop in the plane.
The main goal of this paper is to give a stochastic proof
\textit{Makeenko--Migdal identities}, see Theorems \ref{thm.2.23} and Theorems
\ref{thm.2.21} below. At  the heart of these equations is an infinite
dimensional integration by parts argument with gauge-field theory
complications. [Without the complications of gauge invariance, the
Makeenko-Migdal identities would likely be referred to as Schwinger-Dyson
equations in the quantum filed theory litterature.] Although most interacting
field theories are still not on firm mathematical grounds, formal integration
in heuristic path integral expressions gives significant insight and
constraints on the underlying quantum field theory.

For the $M=\mathbb{R}^{2}$ -- setting with $K=U\left(  N\right)  ,$ in the
fundamental paper \cite{Levy2017} (which appeared on the archive in 2011), T.
L\'{e}vy was able to show in the $N\rightarrow\infty$ limit that that the
Wilson loop functionals can, with a little added information, be completely
recovered from the Makeenko--Migdal identities. Recent significant progress in
this direction when $M$ is a sphere or even general compact surface may be
found in \cite{Dahlqvist2017} and \cite{Hall2017} respectively. The reader is
also referred to these papers and to \cite{Levy2017} and \cite{Cebron2017a}
for more background on (generalized) gauge fields over two dimensional
manifolds.

The original Makeenko and Migdal (heuristic) identities, in any dimension,
were the subject of  \cite{MM}. V. A. Kazakov and I. K. Kostov (1980) in
\cite[Section 4]{KK} showed in the plane case that one side of the MM identity
may be interpreted as the alternating sum of derivatives of the Wilson loop
functional with respect to the areas of the faces surrounding a simple
crossing, see also \cite[Equation 9]{KazakovU(N)} and Gopakumar and Gross
(1995) \cite[Equation 6.4]{GG}. T. L\'{e}vy \cite{Levy2017} (appeared in 2011)
was the first to provide a rigorous proof of the planar Makeenko--Migdal
equations and also introduced a more general form of the equations. L\'{e}vy's
proof of the generalized Makeenko--Migdal equation (gMM for short) are finite
dimensional in nature and are based on Wilson functional expectation formulas
developed in \cite{Driver89b} and also \cite{GrossKingSengupta}.

A different proof of L\'{e}vy's gMM equations was subsequently given by A.
Dahlqvist (2014) in \cite{Dahl}. Recently, in \cite{DHK2016}, three new proofs
of the gMM identity were given and then later in \cite{DHKsurf} it was shown
two of these proofs also work in the context of the Yang--Mills measure over
an arbitrary compact surface. The latter extension replaces the planar formula
from \cite{Driver89b} by their extensions to compact surfaces developed by A.
Sengupta \cite{Sen1,Sen2,Sen3,Sen4}.

Up to now, all of the rigorous proofs of the gMM equations have been
elementary but tricky exercises in finite dimensional integration by parts.
The main aims of this paper are; 1) to explain the original heuristic infinite
dimensional integration by parts arguments of the \textit{Makeenko--Migdal}
equations in more detail and precision, and 2) to then make (using stochastic
calculus) these heuristic arguments rigorous.

One would certainly like to develop the theory here in the physically
interesting case of $d=4.$ However, it is still unknown how to make sense out
of Eq. (\ref{e.1.4}) when $d\geq3.$ (The $d=4$ case is one aspect of the
Clay-Mathematics Millennium problem dealing with quantized Yang-Mills fields.)
There is however some very promising recent progress in this direction.
Regarding the lattice approximations to $\mu_{\text{YM}_{d}},$ see Chatterjee
\cite{Chatterjee2015}\footnote{In Theorem 7.1 of \cite{Chatterjee2015},
Chatterjee introduces a \textquotedblleft generalized Schwinger--Dyson
equation for $SO\left(  N\right)  $ which in fact is a non-intrinsic writing
of the standard Green's identity for $SO\left(  N\right)  ,$ i.e. integration
by parts for the Laplacian on $SO\left(  N\right)  .$ He is able to make good
use of this non-intrinsic formula to find interesting convergent series expansions
for loop expectations in the lattice model. However, the integration by parts
used by Chatterjee seems to not be very closely related to what Makeenko and
Migdal did in \cite{MM}.}, \cite{Chatterjee2016a}, and especially
\cite{Chatterjee2016} where the lattice normalization constant is shown to be
under \textquotedblleft control\textquotedblright\ as the lattice spacing
tends to zero. On another front, the reader is directed to the work of
Charalambous and L. Gross
\cite{Charalambous2013,Charalambous2015,Charalambous2017} and L. Gross
\cite{Gross2016} where the general theme is to develop the Yang-Mills heat
equation as a method of regularizing the Wilson loop variables when $d\geq3.$
An outline of this strategy (which coupled with \cite{Chatterjee2016} may
finally lead to a construction of quantized Yang-Mills fields when $d=4)$ is
given in the introduction in \cite{Charalambous2015}.

We finish this introduction with a road map of this paper. The author would
also like to empahsize that the this paper is a rigorous interpretation of the
extremely illuminating heuristic discussion of the MM equations given in the
introduction of \cite{Levy2017}.

\subsection{Guide to the reader}

In section \ref{sec.2} we will formally introduce the YM$_{2}$-expectations
and give the statements of the main theorems listed here.

\begin{enumerate}
\item Theorem \ref{thm.2.14} reviews the structure of the YM$_{2}%
$-expectations as worked out in \cite{Driver89b}. [For a perturbative theory
perspective of these expectations the reader is directed to T. Nguyen's very
interesting papers \cite{Nguyen2015,Nguyen2016,Nguyen2016a,Nguyen2016b}.]

\item Theorem \ref{thm.2.21} is a statement of T. L\'{e}vy's generalized form
of the Makeenko--Migdal equations, see \cite{Levy2017}.

\item Theorem \ref{thm.2.23} is a white noise integration by parts formula
relating one side of the gMM equation to an expression directly involving the
curvature white noise function. This is the first main theorem of the paper.

\item Theorem \ref{thm.2.27} is a Wilson-loop expansion formula which relates
curvature white noise expression in Theorem \ref{thm.2.23} to the other side
of the gMM\textit{ }equations.

\item The proof of Theorem \ref{thm.2.21} is a simple consequence of Theorems
\ref{thm.2.23} and \ref{thm.2.27} as explained at the end of Section
\ref{sec.2}.
\end{enumerate}

In Section \ref{sec.3} we will give heuristic \textquotedblleft
proofs\textquotedblright\ of Theorems \ref{thm.2.23} and \ref{thm.2.27}. The
rigorous proofs of Theorems \ref{thm.2.23} and \ref{thm.2.27} will be given in
Sections \ref{sec.4} and Section \ref{sec.5} respectively. Lastly there are
two appendices to this paper. Appendix \ref{sec.A} reviews a few basic facts
involving curvature and parallel translation. Appendix \ref{sec.B} reviews
\textquotedblleft homotopy\textquotedblright\ gauges and their relevance for
interpreting the informal Yang-Mills' expectations which are suggestively
described by Eq. (\ref{e.1.4}).

\subsection{Acknowledgements}

It is a pleasure to acknowledge useful discussions with Brian Hall, Len Gross,
and Todd Kemp pertaining to this work.

\section{Stochastic formulation and statement of the theorems\label{sec.2}}

\subsection{The stochastic framework\label{sec.2.1}}

The heuristics Eq. (\ref{e.1.7}), suggests that the random curvature function,
\textquotedblleft$f=-\partial_{y}A$\textquotedblright, should be taken to be a
$\mathfrak{k}$-valued white noise as in the next definition.

\begin{definition}
[White noise]\label{def.2.1}A $\mathfrak{k}$-valued \textbf{white noise} is a
probability space, $\left(  \Omega,\mathcal{B},\mathbb{P}\right)  ,$ along
with a linear map, $f:L^{2}\left(  \mathbb{R}^{2},m;\mathfrak{k}\right)
\rightarrow L^{2}\left(  \Omega,\mathbb{P};\mathbb{R}\right)  ,$ such that
$f\left(  u\right)  $ is a mean zero Gaussian random variable with covariance
$\mathbb{E}\left[  \left[  f\left(  u\right)  \right]  ^{2}\right]
=\left\Vert u\right\Vert _{2}^{2}$ for each $u\in L^{2}\left(  \mathbb{R}%
^{2},m;\mathfrak{k}\right)  .$
\end{definition}

We will typically write $f\left(  u\right)  $ as $\left\langle
f,u\right\rangle $ and informally think that $\left\langle f,u\right\rangle $
is given as in Eq. (\ref{e.1.2}) even though $f\left(  x,y\right)  $ has no
meaning as a random variable for any $\left(  x,y\right)  \in\mathbb{R}^{2}$.

\begin{notation}
\label{not.2.2}If $\left(  \Omega,\mathcal{B},\mathbb{P},f\right)  $ is a
$\mathfrak{k}$-valued white noise and $B$ is a finite area Borel subset of
$\mathbb{R}^{2},$ we abuse notation and let $f\left(  B\right)  $ (informally
thought of as $\int_{B}f\left(  x,y\right)  dxdy))$ be the random variable
(well defined a $\mathbb{P}$-a.e.) which satisfies,
\[
\left\langle f\left(  B\right)  ,\xi\right\rangle _{\mathfrak{k}%
}:=\left\langle f,1_{B}\xi\right\rangle \text{ for all }\xi\in\mathfrak{k.}%
\]

\end{notation}

A version of $f\left(  B\right)  \in L^{2}\left(  \Omega,\mathbb{P}%
;\mathfrak{k}\right)  $ is given by $f\left(  B\right)  =\sum_{\xi\in\beta
}\left\langle f,1_{B}\xi\right\rangle \xi$ where $\beta\subset\mathfrak{k}$ is
any orthonormal basis for $\mathfrak{k}.$ From the white noise, $f,$ we may
formally use Eq. (\ref{e.1.5}) to recover the random (distribution valued)
connection one form, $A.$ Let us recall L. Gross's method (as explained in
\cite{Driver89b,GrossKingSengupta}) on how to rigorously do this and then how
to use stochastic calculus to interpret the resulting random parallel
translation operators associated to \textquotedblleft tame\textquotedblright%
\ curves in the plane.

\begin{definition}
[Horizontal curves]\label{def.2.3}A \textbf{horizontal curve} is a path,
$\ell:\left[  a,b\right]  \rightarrow\mathbb{R}^{2},$ of the form $\ell\left(
t\right)  =\left(  t,y\left(  t\right)  \right)  $ where $y:\left[
a,b\right]  \rightarrow\mathbb{R}$ is a continuous function on a compact
interval, $\left[  a,b\right]  .$ Further let $R^{\ell}\left(  t\right)  $ be
the region between the curve $y$ and the $x$-axis over interval $\left[
a,t\right]  ,$ see Figure \ref{fig.1}.
\end{definition}

\begin{figure}[ptbh]
\centering
\par
\psize{2.5in} %
\executeiffilenewer{\GraphicsDirectoryRL_regions.svg}{\GraphicsDirectoryRL_regions.pdf}%
{inkscape -z -D --file=\GraphicsDirectoryRL_regions.svg --export-pdf=\GraphicsDirectoryRL_regions.pdf --export-latex}%
\input{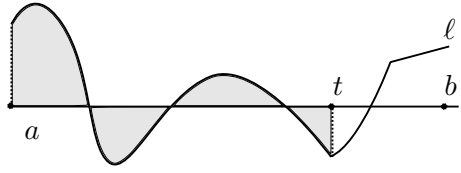}%
 \caption{A typical horizontal curve with
$R^{\ell}\left(  t\right)  $ being the shaded region.}%
\label{fig.1}%
\end{figure}

Again working formally, for $A~dx\in\mathcal{A}_{0}$ we integrate
$f=-\partial_{y}A$ to find
\[
A\left(  x,y\right)  =-\int_{0}^{y}f\left(  x,y^{\prime}\right)  dy^{\prime}%
\]
and so if $\ell\left(  t\right)  =\left(  t,y\left(  t\right)  \right)  $ is a
horizontal\ curve as in Definition \ref{def.2.3}, then
\[
A\,dx\left\langle \dot{\ell}\left(  \tau\right)  \right\rangle =A\left(
\tau,y\left(  \tau\right)  \right)  =-\int_{0}^{y\left(  \tau\right)
}f\left(  \tau,y^{\prime}\right)  dy^{\prime}.
\]
We now integrate this expression to define,%
\begin{align}
M_{t}^{f}\left(  \ell\right)   &  :=\int_{a}^{t}A\,dx\left\langle \dot{\ell
}\left(  \tau\right)  \right\rangle d\tau=-\int_{a}^{t}\left[  \int
_{0}^{y\left(  \tau\right)  }f\left(  \tau,y^{\prime}\right)  dy^{\prime
}\right]  d\tau\label{e.2.1}\\
&  =-\int_{R^{\ell}\left(  t\right)  }\hat{f}\left(  x,y\right)  dxdy,
\label{e.2.2}%
\end{align}
where, for a function $u:\mathbb{R}^{2}\rightarrow\mathfrak{k,}$ we let
\begin{equation}
\hat{u}\left(  x,y\right)  :=\mathrm{sgn}(y)u\left(  x,y\right)  =\left\{
\begin{array}
[c]{ccc}%
u\left(  x,y\right)  & \text{if} & y\geq0\\
-u\left(  x,y\right)  & \text{if} & y\leq0
\end{array}
\right.  . \label{e.2.3}%
\end{equation}
By the construction of $M_{t}^{f}\left(  \ell\right)  $ in Eq. (\ref{e.2.1}),
$\dot{M}_{t}^{f}\left(  \ell\right)  =A\,dx\left\langle \dot{\ell}\left(
t\right)  \right\rangle $ and so (writing $\pt_{t}^{A}\left(  \ell\right)  $
for $\pt_{t}^{A\,dx}\left(  \ell\right)  ),$ we should have%
\begin{equation}
\frac{d}{dt}\pt_{t}^{A}\left(  \ell\right)  +\dot{M}_{t}^{f}\left(
\ell\right)  \pt_{t}^{A}\left(  \ell\right)  =0\text{ with }\pt_{a}^{A}\left(
\ell\right)  =I\in K. \label{e.2.4}%
\end{equation}
In the stochastic setting $f$ and $A$ are distribution valued random variables
and hence Eqs. (\ref{e.2.2}) and (\ref{e.2.4}) will require proper
interpretation which is provided in Definitions \ref{def.2.6} and
\ref{def.2.7} below. Before these key definition we need a little more preparation.

\begin{definition}
\label{def.2.4}If $f$ is a white noise define $\hat{f}$ to be the white noise
determined by%
\[
\left\langle \hat{f},u\right\rangle =\left\langle f,\hat{u}\right\rangle
\text{ for all }u\in L^{2}\left(  \mathbb{R}^{2},m;\mathfrak{k}\right)  .
\]

\end{definition}

Let $\mathcal{B}_{\mathbb{R}^{2}}^{0}$ denote the \textbf{finite area} Borel
subsets of $\mathbb{R}^{2}.$ Notice that if $B\in\mathcal{B}_{\mathbb{R}^{2}%
}^{0},$ then
\[
\hat{f}\left(  B\right)  =\left\langle f,\widehat{1_{B}}\right\rangle
=f\left(  B\cap H_{+}\right)  -f\left(  B\cap H_{-}\right)
\]
where $H_{+}$ and $H_{-}$ denote the upper and lower half planes in
$\mathbb{R}^{2}.$

\begin{definition}
[Filtration]\label{def.2.5}For $s\in\mathbb{R},$ let
\[
\mathcal{F}_{s}^{0}=\sigma\left(  f\left(  B\right)  :B\in\mathcal{B}%
_{\mathbb{R}^{2}}^{0}\text{ s.t. }B\subset\left(  -\infty,s\right)
\times\mathbb{R}\right)
\]
and then let $\left\{  \mathcal{F}_{s}\right\}  _{s\in\mathbb{R}}$ be the
filtration on $\left(  \Omega,\mathcal{F},\mathbb{P}\right)  $ which is the
right continuous version of $\left\{  \mathcal{F}_{s}^{0}\right\}
_{s\in\mathbb{R}}$ augmented by the zero sets of $\mathcal{F}.$
\end{definition}

Defining $M_{t}^{f}\left(  \ell\right)  :=-\hat{f}\left(  R^{\ell}\left(
t\right)  \right)  $ as indicated in Eq. (\ref{e.2.2}) implies that $\left\{
M_{t}^{f}\left(  \ell\right)  :a\leq t\leq b\right\}  $ is a mean zero
$\mathfrak{k}$-valued Gaussian process with
\[
\mathbb{E}\left[  \left\langle M_{t}^{f}\left(  \ell\right)  ,\xi
_{1}\right\rangle \left\langle M_{\tau}^{f}\left(  \ell\right)  ,\xi
_{2}\right\rangle \right]  =\left\langle \xi_{1},\xi_{2}\right\rangle
_{\mathfrak{k}}m^{2}\left(  R^{\ell}\left(  t\wedge\tau\right)  \right)
\text{ }\forall~\xi_{1,}\xi_{2}\in\mathfrak{k}\text{ and }a\leq t,\tau\leq b.
\]
This process has independent increments and is a time change of $\mathfrak{k}%
$-valued Brownian motion and hence has a continuous version. The existence of
a continuous version also follows using Kolmogorov's continuity criteria along
with the observation that (for $t>\tau)$
\[
\mathbb{E}\left[  \left\vert M_{t}^{f}\left(  \ell\right)  -M_{\tau}%
^{f}\left(  \ell\right)  \right\vert _{\mathfrak{k}}^{2}\right]  =\dim K\cdot
m\left(  R^{\ell}\left(  t\right)  \setminus R^{\ell}\left(  \tau\right)
\right)  =\dim K\cdot\int_{\tau}^{t}\left\vert y\left(  x\right)  \right\vert
dx\leq C\left\vert t-\tau\right\vert
\]
which, because $\left\{  M_{t}^{f}\left(  \ell\right)  \right\}  _{a\leq t\leq
b}$ is a Gaussian process, implies for all $p\geq2,$ there exists
$C_{p}<\infty$ such that%
\[
\mathbb{E}\left[  \left\vert M_{t}^{f}\left(  \ell\right)  -M_{\tau}%
^{f}\left(  \ell\right)  \right\vert _{\mathfrak{k}}^{p}\right]  \leq
C_{p}\left\vert t-\tau\right\vert ^{p/2}.
\]

\begin{definition}
[Martingales]\label{def.2.6}From now on we assume that $M_{t}^{f}\left(
\ell\right)  $ refers to a continuous version of $\left[  a,b\right]  \ni
t\rightarrow-\hat{f}\left(  R^{\ell}\left(  t\right)  \right)  .$ This version
becomes a continuous $\mathfrak{k}$-valued martingale adapted to the
filtration $\left\{  \mathcal{F}_{s}\right\}  _{s\in\mathbb{R}}.$
\end{definition}

Motivated by Eq. (\ref{e.2.4}) we now defined the stochastic parallel
translation as follows.

\begin{definition}
[Stochastic Parallel Translation]\label{def.2.7}If $\left[  a,b\right]  \ni
t\rightarrow\ell\left(  t\right)  =\left(  t,y\left(  t\right)  \right)
\in\mathbb{R}^{2}$ is a horizontal\ curve as in Definition \ref{def.2.3},
let\footnote{To emphasize that the white noise is the fundamental input we are
now writing $//_{t}^{f}\left(  \ell\right)  $ in place of $//_{t}^{A}\left(
\ell\right)  .$ Although, in the heuristic proof section \ref{sec.3} below, we
will briefly revert back to the old notation.} $\pt_{t}^{f}\left(
\ell\right)  $ be the $K$-valued random process which is defined as the
solution to the stochastic differential equation;%
\begin{equation}
d\pt_{t}^{f}\left(  \ell\right)  +\delta M_{t}^{f}\left(  \ell\right)
\pt_{t}^{f}\left(  \ell\right)  =0\text{ with }\pt_{a}\left(  \ell\right)
=I\in K, \label{e.2.5}%
\end{equation}
where $\delta M_{t}^{f}\left(  \ell\right)  \ $denotes the Fisk-Stratonovich
differential of $M^{\ell}.$ Parallel translation along any purely vertical
path in $\mathbb{R}^{2}$ is defined to be the constant function $I\in K.$ We
further simply write $\pt^{f}\left(  \ell\right)  $ for $\pt_{b}^{f}\left(
\ell\right)  .$
\end{definition}

A directed path $\sigma$ in $\mathbb{R}^{2}$ is \textbf{tame }if it is the
concatenation of finitely many vertical paths and forwards and backwards
horizontal paths. For a tame path $\sigma$ we define $\pt^{f}\left(
\sigma\right)  $ in the natural way as the products of forward parallel
translations for the forward paths and inverses of parallel translations for
the backward horizontal paths. For example if $\sigma$ is the tame path shown
in Figure \ref{fig.2} we let $\pt^{f}\left(  \sigma\right)  =\pt^{f}\left(
\ell_{3}\right)  \pt^{f}\left(  \ell_{2}\right)  ^{-1}\pt^{f}\left(  \ell
_{1}\right)  $ where $\ell_{i}\left(  t\right)  =\left(  t,y_{i}\left(
t\right)  \right)  $ and $y_{1}:\left[  a,c\right]  \rightarrow\mathbb{R},$
$y_{2}:\left[  b,c\right]  \rightarrow\mathbb{R},$ and $y_{3}:\left[
b,d\right]  \rightarrow\mathbb{R}$ are the functions whose graphs are
indicated in Figure \ref{fig.2}.\begin{figure}[ptbh]
\centering
\par
\psize{3.5in} %
\executeiffilenewer{\GraphicsDirectorytame.svg}{\GraphicsDirectorytame.pdf}%
{inkscape -z -D --file=\GraphicsDirectorytame.svg --export-pdf=\GraphicsDirectorytame.pdf --export-latex}%
\input{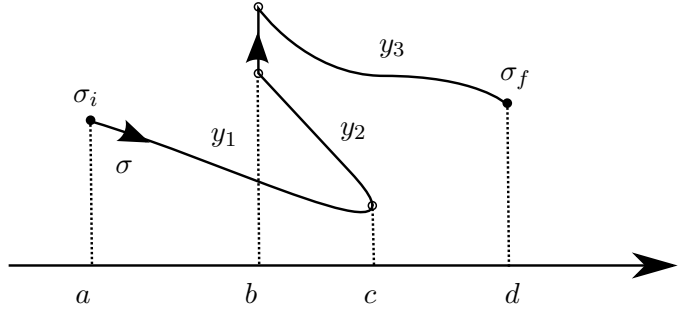}%
 \caption{A tame path, $\sigma,$ with one
vertical segment, two forward horizontal segments, and one backwards
horizontal segment. }%
\label{fig.2}%
\end{figure}

\begin{notation}
\label{not.2.8}Given a directed tame path $\sigma$ in $\mathbb{R}^{2},$ let
$\tilde{\sigma}\subset\mathbb{R}^{2}$ denote the image of $\sigma$ in
$\mathbb{R}^{2},$ $\sigma_{f}$ denote the final point of $\sigma$ and
$\sigma_{i}$ denote the initial point of $\sigma,$ see Figure \ref{fig.2}. A
\textbf{tame graph}, $\mathbb{G},$ in $\mathbb{R}^{2}$ is a finite collection
of directed tame paths such that; if $\sigma,\tau$ are any two distinct
elements in $\mathbb{G}$ then $\tilde{\sigma}\cap\tilde{\tau}\subset\left\{
\sigma_{i},\tau_{i},\sigma_{f},\tau_{f}\right\}  ,$ see for example Figure
\ref{fig.3}. Let $V\left(  \mathbb{G}\right)  =\cup_{\sigma\in\mathbb{G}%
}\left\{  \sigma_{i},\sigma_{f}\right\}  $ denote the \textbf{vertices} of
$\mathbb{G}.$\begin{figure}[ptbh]
\centering
\par
\psize{3.0in} %
\executeiffilenewer{\GraphicsDirectoryloop4.svg}{\GraphicsDirectoryloop4.pdf}%
{inkscape -z -D --file=\GraphicsDirectoryloop4.svg --export-pdf=\GraphicsDirectoryloop4.pdf --export-latex}%
\input{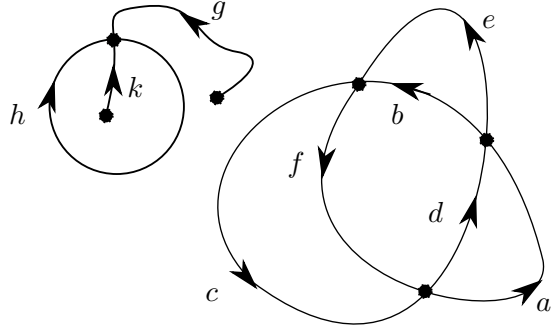}%
 \caption{An example of a tame graph,
$\mathbb{G}$. In this example $g_{f}=k_{f}=h_{f}=h_{i}.$ Note, as in this
graph, there is not assumption that tame graphs are connected.}%
\label{fig.3}%
\end{figure}

\end{notation}

\begin{notation}
\label{not.2.9}To each directed tame graph, $\mathbb{G},$ in $\mathbb{R}^{2},$
let $K^{\mathbb{G}}$ be the the \textbf{configuration space} of all functions,
$\omega:\mathbb{G}\rightarrow K.$ Further let
\[
\pt^{f}\left(  \mathbb{G}\right)  :=\left\{  \pt^{f}\left(  \sigma\right)
:\sigma\in\mathbb{G}\right\}  \in K^{\mathbb{G}}%
\]
be stochastic parallel translation along all the paths in $\mathbb{G}.$ More
accurately we should denote $\pt^{f}\left(  \mathbb{G}\right)  $ by
$\pt^{f}|_{\mathbb{G}}$ but this notation is a bit cumbersome.
\end{notation}

\begin{definition}
[Discrete gauge tranformations]\label{def.2.10}A function $u:V\left(
\mathbb{G}\right)  \rightarrow K$ is called a \textbf{discrete gauge
transformation. }To each such function $u:V\left(  \mathbb{G}\right)
\rightarrow K$ and $\omega\in K^{\mathbb{G}},$ we let $\omega^{u}\in
K^{\mathbb{G}}$ be defined by
\[
\omega^{u}\left(  \sigma\right)  =u\left(  \sigma_{f}\right)  ^{-1}%
\omega\left(  \sigma\right)  u\left(  \sigma_{i}\right)  \text{ for all
}\sigma\in\mathbb{G}.
\]

\end{definition}

\begin{example}
If $x_{0}\in V\left(  \mathbb{G}\right)  \subset\mathbb{R}^{2}$ and $k\in K$
are given, then $u_{x_{0},k}:V\left(  \mathbb{G}\right)  \rightarrow K$
defined by
\[
u_{x_{0},k}\left(  x\right)  =\left\{
\begin{array}
[c]{ccc}%
k & \text{if} & x=x_{0}\\
I & \text{if} & x\neq x_{0}%
\end{array}
\right.
\]
is a discrete gauge transformation. Moreover any discrete gauge transformation
may be written as a finite product of the $u_{x_{0},k}$ via%
\begin{equation}
u=\prod_{x_{0}\in V\left(  \mathbb{G}\right)  }u_{x_{0},u\left(  x_{0}\right)
} \label{e.2.6}%
\end{equation}
where the above pointwise product is independent of any choice of ordering of
the terms.
\end{example}

\begin{definition}
[Discrete gauge invariant functions]\label{def.2.12}A function,
$U:K^{\mathbb{G}}\rightarrow\mathbb{C},$ is said to be \textbf{discrete gauge
invariant at }$x_{0}\in V\left(  \mathbb{G}\right)  $ if $U\left(
\omega^{u_{x_{0},k}}\right)  =U\left(  \omega\right)  $ for all $k\in K$ and
$\omega\in K^{\mathbb{G}}.$ The function, $U:K^{\mathbb{G}}\rightarrow
\mathbb{C},$ is said to be \textbf{discrete gauge invariant }if $U\left(
\omega^{u}\right)  =U\left(  \omega\right)  $ for all $\omega\in
K^{\mathbb{G}}$ and all discrete gauge transformations, $u:V\left(
\mathbb{G}\right)  \rightarrow K.$ [Because of Eq. (\ref{e.2.6}) it is easily
seen that $U$ is discrete gauge invariant iff it is discrete gauge invariant
at each $x_{0}\in V\left(  \mathbb{G}\right)  .]$
\end{definition}

\begin{remark}
\label{rem.2.13}For a tame graph, $\mathbb{G},$ in $\mathbb{R}^{2}$ and
$A\in\mathcal{A},$ let (analogous to Notation \ref{not.2.9})%
\[
\pt^{A}\left(  \mathbb{G}\right)  :=\left\{  \pt^{A^{g}}\left(  \sigma\right)
:\sigma\in\mathbb{G}\right\}  \in K^{\mathbb{G}}.
\]
If $U:K^{\mathbb{G}}\rightarrow\mathbb{C}$ is a discrete gauge invariant
function, then $\Psi:\mathcal{A}\rightarrow\mathbb{C}$ defined by $\Psi\left(
A\right)  =U\left(  \pt^{A}\left(  \mathbb{G}\right)  \right)  $ is a gauge
invariant function on $\mathcal{A}.$ Indeed if $g\in C^{1}\left(
\mathbb{R}^{2}\rightarrow K\right)  ,$ then (see Eq. (\ref{e.A.2}) of Appendix
\ref{sec.A})
\begin{align*}
\Psi\left(  A^{g}\right)   &  =U\left(  \pt^{A}\left(  \mathbb{G}\right)
\right)  =U\left(  \left\{  g\left(  \sigma_{f}\right)  ^{-1}\pt^{A}\left(
\sigma\right)  g\left(  \sigma_{i}\right)  \right\}  _{\sigma\in\mathbb{G}%
}\right) \\
&  =U\left(  \pt^{A}\left(  \mathbb{G}\right)  \right)  =\Psi\left(  A\right)
.
\end{align*}

\end{remark}

When $U:K^{\mathbb{G}}\rightarrow\mathbb{C}$ is a discrete gauge invariant
function, it is shown in \cite{Driver89b} that $\mathbb{E}\left[  U\left(
\pt^{f}\left(  \mathbb{G}\right)  \right)  \right]  $ may be computed as a
finite dimensional integrals relative to a density constructed via certain
products of the convolution heat kernel on $K.$ The next theorem summarizes
the information we need for the purposes of this paper.

\begin{theorem}
[\cite{Driver89b}]\label{thm.2.14}If $\mathbb{G}$ is a directed tame graph and
$U:K^{\mathbb{G}}\rightarrow\mathbb{C}$ is a bounded measurable discrete gauge
invariant function, then
\[
\mathbb{E}\left[  U\left(  \pt^{f}\left(  \mathbb{G}\right)  \right)  \right]
=\int_{K^{\mathbb{G}}}U\left(  \omega\right)  \rho_{\mathbb{G}}\left(
\omega\right)  d\lambda_{\mathbb{G}}\left(  \omega\right)
\]
where $\lambda_{\mathbb{G}}$ is normalized Haar measure on $K^{\mathbb{G}}$
and $\rho_{\mathbb{G}}$ is a smooth density function which is depends only on
the topology of $\mathbb{G}$ and the areas $\left\{  t_{i}\right\}  _{i=1}%
^{N}$ of the bounded connected regions in $\mathbb{R}^{2}\setminus\mathbb{G}.$
Moreover, relative to topological preserving perturbations of $\mathbb{G},$
$\mathbb{E}\left[  U\left(  \pt^{f}\left(  \mathbb{G}\right)  \right)
\right]  $ is a smooth function of $\left(  t_{1},\dots,t_{N}\right)  .$
\end{theorem}

The following corollary follows directly from Theorem \ref{thm.2.14}. This
corollary is also \textquotedblleft explained\textquotedblright\ heuristically
in Meta-Theorems \ref{mthm.B.55} and \ref{mthm.B.56} of Appendix \ref{sec.B}.

\begin{corollary}
[Area preserving diffeomorphism invariance]\label{cor.2.15}If $\varphi$ is any
area preserving diffeomorphism of $\mathbb{R}^{2}$ such that $\varphi
\circ\sigma$ is a tame path for all $\sigma\in\mathbb{G},$ then%
\[
\mathbb{E}\left[  U\left(  \pt^{f}\left(  \mathbb{G}\right)  \right)  \right]
=\mathbb{E}\left[  U\left(  \left\{  \pt^{f}\left(  \varphi\circ\sigma\right)
\right\}  _{\sigma\in\mathbb{G}}\right)  \right]  .
\]

\end{corollary}

\begin{definition}
[$\sigma$-directional derivatives]\label{def.2.16}For $\xi\in\mathfrak{k},$
$\sigma\in\mathbb{G},$ $\omega\in K^{\mathbb{G}},$ and $U:K^{\mathbb{G}%
}\rightarrow\mathbb{C}$ differentiable, let%
\[
\left(  \nabla_{\xi}^{\sigma}U\right)  \left(  \omega\right)  =\frac{d}%
{dt}|_{0}U\left(  \left\{  \omega\left(  b\right)  e^{t\delta_{\sigma,b}\xi
}\right\}  _{b\in\mathbb{G}}\right)  :=\frac{d}{dt}|_{0}U\left(  \omega
_{t}\right)
\]
where, for $t\in\mathbb{R}$ and $b\in\mathbb{G},$%
\[
\omega_{t}\left(  b\right)  =\left\{
\begin{array}
[c]{ccc}%
\omega\left(  b\right)  & \text{if} & b\neq\sigma\\
\omega\left(  \sigma\right)  e^{t\xi} & \text{if} & b=\sigma
\end{array}
.\right.
\]

\end{definition}

\begin{lemma}
\label{lem.2.17}If $\sigma\in\mathbb{G},$ $\xi\in\mathfrak{k},$ and
$U:K^{\mathbb{G}}\rightarrow\mathbb{C}$ is discrete gauge invariant, then
\[
\left(  \nabla_{\xi}^{\sigma}U\right)  \left(  \omega^{u}\right)  =\left(
\nabla_{\mathrm{Ad}_{u\left(  \sigma_{i}\right)  }\xi}^{\sigma}U\right)
\left(  \omega\right)  \text{ }\forall~\omega\in K^{\mathbb{G}}\text{ and
}u:V\left(  \mathbb{G}\right)  \rightarrow K.
\]
In particular if $u\left(  \sigma_{i}\right)  =I,$ then $\left(  \nabla_{\xi
}^{\sigma}U\right)  \left(  \omega^{u}\right)  =\left(  \nabla_{\xi}^{\sigma
}U\right)  \left(  \omega\right)  $ $\forall~\omega\in K^{\mathbb{G}}.$
\end{lemma}

\begin{proof}
Since
\begin{align*}
\omega^{u}\left(  b\right)  e^{t\delta_{\sigma,b}\xi}  &  =u\left(
b_{f}\right)  ^{-1}\omega\left(  b\right)  u\left(  b_{i}\right)
e^{t\delta_{\sigma,b}\xi}\\
&  =u\left(  b_{f}\right)  ^{-1}\omega\left(  b\right)  u\left(  b_{i}\right)
e^{t\delta_{\sigma,b}\xi}u\left(  b_{i}\right)  ^{-1}u\left(  b_{i}\right) \\
&  =u\left(  b_{f}\right)  ^{-1}\omega\left(  b\right)  e^{t\delta_{\sigma
,b}\mathrm{Ad}_{u\left(  b_{i}\right)  }\xi}u\left(  b_{i}\right)  =u\left(
b_{f}\right)  ^{-1}\omega\left(  b\right)  e^{t\delta_{\sigma,b}%
\mathrm{Ad}_{u\left(  \sigma_{i}\right)  }\xi}u\left(  b_{i}\right)  ,
\end{align*}
it follows using the discrete gauge invariance of $U$ that%
\begin{align*}
\left(  \nabla_{\xi}^{\sigma}U\right)  \left(  \omega^{u}\right)   &
=\frac{d}{dt}|_{0}U\left(  \left\{  u\left(  b_{f}\right)  ^{-1}\omega\left(
b\right)  e^{t\delta_{\sigma,b}\mathrm{Ad}_{u\left(  b_{i}\right)  }\xi
}u\left(  b_{i}\right)  \right\}  _{b\in\mathbb{G}}\right) \\
&  =\frac{d}{dt}|_{0}U\left(  \left\{  \omega\left(  b\right)  e^{t\delta
_{\sigma,b}\mathrm{Ad}_{u\left(  \sigma_{i}\right)  }\xi}\right\}
_{b\in\mathbb{G}}\right)  =\left(  \nabla_{\mathrm{Ad}_{u\left(  \sigma
_{i}\right)  }\xi}^{\sigma}U\right)  \left(  \omega\right)  .
\end{align*}

\end{proof}

One consequence of the previous lemma is that $\nabla_{\xi}^{\sigma}U$ is
typically not discrete gauge invariant. In the next proposition, we will show
that certain second order differential operators do preserve discrete gauge
invariant functions.

\begin{notation}
[$\left(  \nabla^{\sigma_{1}}\cdot\nabla^{\sigma_{2}}U\right)  \left(
\Gamma\left(  f\right)  \right)  $]\label{not.2.18}If $\sigma_{1},\sigma
_{2}\in\mathbb{G}$ and $U:K^{\mathbb{G}}\rightarrow\mathbb{C}$ is twice
continuously differentiable, let
\begin{equation}
\nabla^{\sigma_{1}}\cdot\nabla^{\sigma_{2}}U=\sum_{\xi\in\beta}\nabla_{\xi
}^{\sigma_{1}}\nabla_{\xi}^{\sigma_{2}}U \label{e.2.7}%
\end{equation}
where $\beta$ is an orthonormal basis for $\mathfrak{k.}$
\end{notation}

Since $\mathfrak{k}\times\mathfrak{k}\ni\left(  \xi,\eta\right)
\rightarrow\left(  \nabla_{\eta}^{\sigma_{1}}\nabla_{\xi}^{\sigma_{2}%
}U\right)  \left(  \omega\right)  $ is a bilinear form on $\mathfrak{k}%
\times\mathfrak{k},$ it is easily verified that the sum in Eq. (\ref{e.2.7})
is independent of the choice of orthonormal basis, $\beta,$ for $\mathfrak{k.}%
$

\begin{proposition}
\label{pro.2.19}Let $b,c$ be any two bonds in $\mathbb{G}$ such that
$b_{i}=c_{i}=q$ as in Figure \ref{fig.4}. Then if $U:K^{\mathbb{G}}%
\rightarrow\mathbb{C}$ is twice continuously differentiable and discrete gauge
invariant then $\left(  \nabla^{b}\cdot\nabla^{c}U\right)  $ is still discrete
gauge invariant.\begin{figure}[ptbh]
\centering
\par
\psize{1.5in} %
\executeiffilenewer{\GraphicsDirectoryvbonds.svg}{\GraphicsDirectoryvbonds.pdf}%
{inkscape -z -D --file=\GraphicsDirectoryvbonds.svg --export-pdf=\GraphicsDirectoryvbonds.pdf --export-latex}%
\input{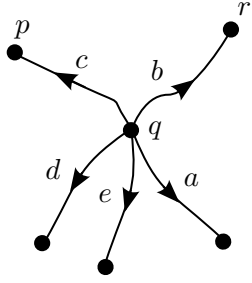}%
\caption{Two bonds among five bonds in
$\mathbb{G}$ sharing $q$ as their initial points.}%
\label{fig.4}%
\end{figure}

\end{proposition}

\begin{proof}
Let us first observe that%
\begin{align*}
\left(  \nabla_{\xi}^{b}\nabla_{\xi}^{c}U\right)  \left(  \omega\right)   &
=\frac{d}{ds}|_{0}\left(  \nabla_{\xi}^{c}U\right)  \left(  \left\{
\omega\left(  \sigma\right)  e^{s\delta_{\sigma,b}\xi}:\sigma\in
\mathbb{G}\right\}  \right) \\
&  =\frac{d}{dt}|_{0}\frac{d}{ds}|_{0}\left(  \nabla_{\xi}^{c}U\right)
\left(  \left\{  \omega\left(  \sigma\right)  e^{s\delta_{\sigma,b}\xi
}e^{t\delta_{\sigma,c}\xi}:\sigma\in\mathbb{G}\right\}  \right) \\
&  =\frac{d}{dt}|_{0}\frac{d}{ds}|_{0}U\left(  \left\{  \omega\left(
\sigma\right)  e^{\left(  s\delta_{\sigma,b}+t\delta_{\sigma,c}\right)  \xi
}:\sigma\in\mathbb{G}\right\}  \right)  .
\end{align*}

Thus if $u:V\left(  \mathbb{G}\right)  \rightarrow G$ is a discrete gauge
transformation and $U$ is a discrete gauge invariant function, then%
\begin{align*}
\left(  \nabla_{\xi}^{b}\nabla_{\xi}^{c}U\right)  \left(  \omega^{u}\right)
&  =\frac{d}{dt}|_{0}\frac{d}{ds}|_{0}U\left(  \left\{  u\left(  \sigma
_{f}\right)  ^{-1}\omega\left(  \sigma\right)  u\left(  \sigma_{i}\right)
e^{\left(  s\delta_{\sigma,b}+t\delta_{\sigma,c}\right)  \xi}:\sigma
\in\mathbb{G}\right\}  \right) \\
&  =\frac{d}{dt}|_{0}\frac{d}{ds}|_{0}U\left(  \left\{  u\left(  \sigma
_{f}\right)  ^{-1}\omega\left(  \sigma\right)  e^{\left(  s\delta_{\sigma
,b}+t\delta_{\sigma,c}\right)  \mathrm{Ad}_{u\left(  \sigma_{i}\right)  }\xi
}u\left(  \sigma_{i}\right)  :\sigma\in\mathbb{G}\right\}  \right) \\
&  =\frac{d}{dt}|_{0}\frac{d}{ds}|_{0}U\left(  \left\{  \omega\left(
\sigma\right)  e^{\left(  s\delta_{\sigma,b}+t\delta_{\sigma,c}\right)
\mathrm{Ad}_{u\left(  q\right)  }\xi}:\sigma\in\mathbb{G}\right\}  \right) \\
&  =\left(  \nabla_{\mathrm{Ad}_{u\left(  q\right)  }\xi}^{b}\nabla
_{\mathrm{Ad}_{u\left(  q\right)  }\xi}^{c}U\right)  \left(  \omega\right)  .
\end{align*}
Summing this equation on $\xi\in\beta$ using $\left\{  \mathrm{Ad}_{u\left(
q\right)  }\xi\right\}  _{\xi\in\beta}$ is still an orthonormal basis for
$\mathfrak{k}$ shows $\left(  \nabla^{b}\cdot\nabla^{c}U\right)  \left(
\omega^{u}\right)  =\left(  \nabla^{b}\cdot\nabla^{c}U\right)  \left(
\omega\right)  .$
\end{proof}

\subsection{Statement of the theorems\label{sec.2.2}}

We are now going to consider graphs, $\mathbb{G},$ in the plane that contain a
simple crossing at some vertex $v\in\mathbb{R}^{2}$ as in Figure \ref{fig.5}
below. \begin{figure}[ptbh]
\centering
\par
\psize{2.0in} %
\executeiffilenewer{\GraphicsDirectorygencross.svg}{\GraphicsDirectorygencross.pdf}%
{inkscape -z -D --file=\GraphicsDirectorygencross.svg --export-pdf=\GraphicsDirectorygencross.pdf --export-latex}%
\input{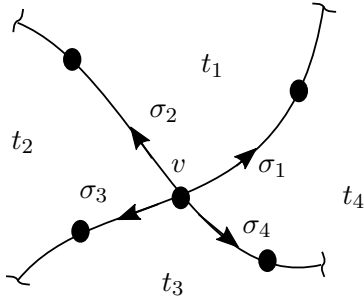}%
\caption{A simple crossing.}%
\label{fig.5}%
\end{figure}

\begin{definition}
\label{def.2.20}Let $\mathbb{G}$ be a directed graph in $\mathbb{R}^{2}$ with
a simple crossing at $v\in V\left(  \mathbb{G}\right)  $ as depicted in Figure
\ref{fig.5}. A function $U:K^{\mathbb{G}}\rightarrow\mathbb{C}$ is said to
have \textbf{extended gauge invariance at}\emph{ }$v\in V\left(
\mathbb{G}\right)  $ if the edges $\sigma_{1},\ldots,\sigma_{4}$ are distinct
and the dependence of $U\left(  \omega\right)  $ on $\left\{  \omega\left(
\sigma_{i}\right)  \right\}  _{i=1}^{4}$ is through $\omega\left(  \sigma
_{1}\right)  \omega\left(  \sigma_{3}\right)  ^{-1}$ and $\omega\left(
\sigma_{2}\right)  \omega\left(  \sigma_{4}\right)  ^{-1}$ only.
\end{definition}

Extended gauge invariance at $v$ should be interpreted to mean that two
otherwise independent paths happened to cross at $v\in\mathbb{R}^{2}$ and
therefore both paths had to artificially be split at $v$ in order to make the
given configuration of paths into a tame graph. An equivalent condition that
$U:K^{\mathbb{G}}\rightarrow\mathbb{C}$ has extended gauge invariance at $v$
is to verify that
\[
K^{2}\ni\left(  x,y\right)  \rightarrow U\left(  \omega\left(  \sigma
_{1}\right)  x,\omega\left(  \sigma_{4}\right)  x,\omega\left(  \sigma
_{2}\right)  y,\omega\left(  \sigma_{4}\right)  y,\left\{  \omega\left(
\sigma\right)  :\sigma\ni\mathbb{G}\setminus\left\{  \sigma_{1},\sigma
_{2},\sigma_{3},\sigma_{4}\right\}  \right\}  \right)
\]
is a constant function for each $\omega\in K^{\mathbb{G}}.$

\begin{theorem}
[Extended M.M. Equations, L\'{e}vy \cite{Levy2017}]\label{thm.2.21}Let
$\mathbb{G}$ be a directed graph in $\mathbb{R}^{2}$ with a simple crossing at
$v\in V\left(  \mathbb{G}\right)  $ as depicted in Figure \ref{fig.5} and
$U\in C^{2}\left(  K^{\mathbb{G}}\rightarrow\mathbb{C}\right)  $ be a gauge
invariant function. If $U$ has extended gauge invariance at $v,$ then%
\[
\mathbb{E}\left[  \left(  \nabla^{\sigma_{1}}\cdot\nabla^{\sigma_{2}}U\right)
\left(  \pt^{f}\left(  \mathbb{G}\right)  \right)  \right]  =-\left(
\frac{\partial}{\partial t_{1}}-\frac{\partial}{\partial t_{2}}+\frac
{\partial}{\partial t_{3}}-\frac{\partial}{\partial t_{4}}\right)
\mathbb{E}\left[  U\left(  \pt^{f}\left(  \mathbb{G}\right)  \right)  \right]
\]
where $\left\{  t_{i}\right\}  _{i=1}^{4}$ refers to the areas described in
Figure \ref{fig.5}. [If any one of these areas are infinite the corresponding
derivative should be omitted from the formula.]
\end{theorem}

We begin with a simple reduction on the geometry of the crossing in Figure
\ref{fig.5}.

\begin{proposition}
\label{pro.2.22}We may find an area preserving diffeomorphism, $\varphi
:\mathbb{R}^{2}\rightarrow\mathbb{R}^{2},$ such that the simple crossing in
Figure \ref{fig.5} may be transformed into the straight line crossing pattern
(i.e. $e_{j}=\varphi\circ\sigma_{j}$ for $1\leq j\leq4)$ going through
$0\in\mathbb{R}^{2}$ as in Figure \ref{fig.6}.
\end{proposition}

\begin{figure}[ptbh]
\centering
\par
\psize{2.0in} %
\executeiffilenewer{\GraphicsDirectorycross.svg}{\GraphicsDirectorycross.pdf}%
{inkscape -z -D --file=\GraphicsDirectorycross.svg --export-pdf=\GraphicsDirectorycross.pdf --export-latex}%
\input{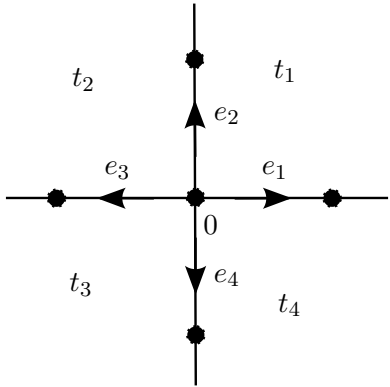}%
\caption{A very simple basic configuration for
$\mathbb{G}$ near $0\in\mathbb{R}^{2}.$}%
\label{fig.6}%
\end{figure}

\begin{proof}
First suppose that $h:\left[  a,b\right]  \rightarrow\mathbb{R}$ is a $C^{1}%
$-function, $\left[  a,b\right]  \ni t\rightarrow\sigma\left(  t\right)
:=\left(  t,h\left(  t\right)  \right)  $ is the associated horizontal curve,
and $t_{0}\in\left(  a,b\right)  $ is a given point. Let $\theta\in
C_{c}^{\infty}\left(  \left(  a,b\right)  ,\left[  0,1\right]  \right)  $ be
chosen so that $\theta=1$ near $t_{0}$ and then define%
\begin{align*}
\varphi\left(  x,y\right)   &  :=\left(  x-t_{0},y-\theta\left(  x\right)
h\left(  x\right)  \right)  \text{ and}\\
\tilde{h}\left(  s\right)   &  :=h\left(  t_{0}+s\right)  \left(
1-\theta\left(  t_{0}+s\right)  \right)  \text{ for }s\in\left[
a-t_{0},b-t_{0}\right]  ,
\end{align*}
where $\theta\left(  x\right)  h\left(  x\right)  \equiv0$ for $x\notin\left[
a,b\right]  .$ It is now a simple matter to verify;

\begin{enumerate}
\item $\varphi$ is an area preserving diffeomorphism.

\item If $t\in\left[  a,b\right]  $ and $s:=t-t_{0},$ then
\[
\varphi\left(  t,h\left(  t\right)  \right)  =\left(  t-t_{0},h\left(
t\right)  -\theta\left(  t\right)  h\left(  t\right)  \right)  =\left(
s,\hat{h}\left(  s\right)  \right)  .
\]
Thus $\varphi$ transforms the horizontal path $\sigma$ to the horizontal path
$\tilde{\sigma}\left(  s\right)  :=\left(  s,\tilde{h}\left(  s\right)
\right)  $ which lies on the $x$-axis for $s$ near $0.$

\item Moreover, $\varphi\left(  t_{0},y\right)  =\left(  0,y-h\left(
t_{0}\right)  \right)  $ so that the vertical path going through
$\sigma\left(  t_{0}\right)  =\left(  t_{0},h\left(  t_{0}\right)  \right)  $
is transformed into the vertical path going through $\tilde{\sigma}\left(
0\right)  =\left(  0,0\right)  .$
\end{enumerate}

We now construct $\varphi$ as in the theorem as a composition of a number or
area preserving diffeomorphisms as follows. Let $\sigma_{2}\ast\bar{\sigma
}_{4}$ denote the path in $\mathbb{R}^{2}$ which follows $\sigma_{4}$
backwards and then continues to follow the path $\sigma_{2}.$ Choose a
rotation which rotates $\sigma_{2}\ast\bar{\sigma}_{4}$ into a path which is
almost horizontal, then apply the above construction to find a $\varphi$ which
makes the resulting curve lie on the $x$-axis, and then rotate this curve back
to a vertical curve. Finally apply the above construction to the image of
$\sigma_{1}\ast\bar{\sigma}_{3}$ under the above area preserving
diffeomorphisms so that the resulting image curves are as in Figure
\ref{fig.6}.
\end{proof}

Because of Propositions \ref{pro.2.22} and \ref{pro.2.19} along with the
structure of the Yang-Mills expectations as described in Theorem
\ref{thm.2.14} (in particular the invariance under area preserving
diffeomorphisms) it suffices to prove Theorem \ref{thm.2.21} in the special
case where $\mathbb{G}$ is a graph in $\mathbb{R}^{2}$ such that $0\in
V\left(  \mathbb{G}\right)  $ and $\mathbb{G}$ contains a cross lined up with
the coordinate axes as in Figure \ref{fig.6}.

\begin{assumption}
\label{ass.1}For the body of this paper we will assume that $\mathbb{G}$ is a
tame graph as just described.
\end{assumption}

\begin{theorem}
\label{thm.2.23}Let $\mathbb{G}$ and $\pt^{f}\left(  \mathbb{G}\right)
:=\pt^{f}|_{\mathbb{G}}$ be as in Notation \ref{not.2.9} and further assume
$\mathbb{G}$ contains the bonds $\left\{  e_{1},\dots,e_{4}\right\}  $ as
above. If $U\left(  \pt^{f}\left(  \mathbb{G}\right)  \right)  $ is gauge
invariant at $0,$ then
\begin{equation}
\mathbb{E}\left[  \left(  \nabla^{e_{1}}\cdot\nabla^{e_{2}}U\right)  \left(
\pt^{f}\left(  \mathbb{G}\right)  \right)  \right]  =-\frac{1}{\left\vert
Q\right\vert }\mathbb{E}\left[  \left(  \nabla_{f\left(  Q\right)  }^{e_{2}%
}U-\nabla_{f\left(  RQ\right)  }^{e_{2}}U\right)  \left(  \pt^{f}\left(
\mathbb{G}\right)  \right)  \right]  . \label{e.2.8}%
\end{equation}
where $R\left(  x,y\right)  :=\left(  x,-y\right)  $ is reflection across the
$x$-axis and $Q$ and $RQ$ are the regions shown in Figure \ref{fig.7} and
$\left\vert Q\right\vert =m\left(  Q\right)  $ is the area of $Q.$%
\begin{figure}[ptbh]
\centering
\par
\psize{2.0in} %
\executeiffilenewer{\GraphicsDirectorycross_curved_regions.svg}{\GraphicsDirectorycross_curved_regions.pdf}%
{inkscape -z -D --file=\GraphicsDirectorycross_curved_regions.svg --export-pdf=\GraphicsDirectorycross_curved_regions.pdf --export-latex}%
\input{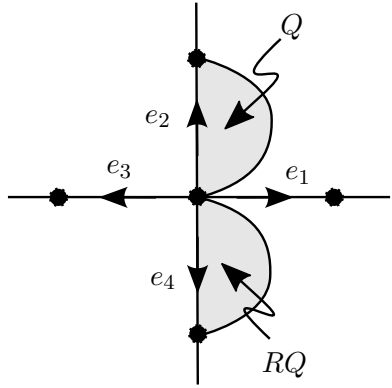}%
\caption{The regions $Q$ and
$RQ$ used in Eq. (\ref{e.2.8}). }%
\label{fig.7}%
\end{figure}

\end{theorem}

\begin{corollary}
\label{cor.2.24}Continuing the setup in Theorem \ref{thm.2.23}, if $U$ is
further assumed to have extended gauge invariance at $0$ (see Definition
\ref{def.2.20}), then
\begin{equation}
\mathbb{E}\left[  \left(  \nabla^{e_{1}}\cdot\nabla^{e_{2}}U\right)  \left(
\pt^{f}\left(  \mathbb{G}\right)  \right)  \right]  =-\frac{1}{\left\vert
Q\right\vert }\mathbb{E}\left[  \left(  \nabla_{f\left(  Q\right)  }^{e_{2}%
}U+\nabla_{f\left(  RQ\right)  }^{e_{4}}U\right)  \left(  \pt^{f}\left(
\mathbb{G}\right)  \right)  \right]  . \label{e.2.9}%
\end{equation}

\end{corollary}

The proof of Theorem \ref{thm.2.21} will now amount to taking the limit (see
Theorem \ref{thm.2.27}) in Eq. (\ref{e.2.9}) where $Q$ shrinks down to the the
line segment, $e_{2}.$ To be more precise we introduce the following notation.

\begin{notation}
\label{not.2.25}For $0<\varepsilon$ small, let $Q_{\varepsilon}$ be a slender
non-empty region as in Figure \ref{fig.8} and $e_{2}^{\varepsilon}$ and
$e_{4}^{\varepsilon}$ be the deformations of $e_{2}$ and $e_{4}$ bounding the
right side of $Q_{\varepsilon}$ and $RQ_{\varepsilon}$ respectively as in the
same figure. We further suppose that $Q_{\varepsilon}$ is indexed in such a
way that
\[
\varepsilon=\max\left\{  x>0:\left[  \left\{  x\right\}  \times\mathbb{R}%
\right]  \cap Q_{\varepsilon}\right\}  .
\]

\end{notation}

\begin{notation}
\label{not.2.26}Let $\mathbb{G}_{\varepsilon,+}$ be the graph $\mathbb{G}$
with $e_{2}$ replaced the deformed path $e_{2}^{\varepsilon}$ and
$\mathbb{G}_{\varepsilon,-}$ be the graph $\mathbb{G}$ where $e_{4}$ is
replaced by the deformed path $e_{4}^{\varepsilon}$ as shown in Figure
\ref{fig.8}. \begin{figure}[ptbh]
\centering
\par
\psize{2.0in} %
\executeiffilenewer{\GraphicsDirectorycross_curved_deformation.svg}{\GraphicsDirectorycross_curved_deformation.pdf}%
{inkscape -z -D --file=\GraphicsDirectorycross_curved_deformation.svg --export-pdf=\GraphicsDirectorycross_curved_deformation.pdf --export-latex}%
\input{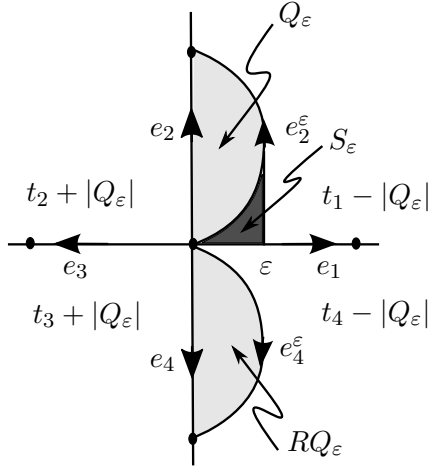}%
\caption{Deforming $e_{2}$
and $e_{4}$ in order to construct $\mathbb{G}_{+,\varepsilon}$ and
$\mathbb{G}_{-,\varepsilon}$ respectively. }%
\label{fig.8}%
\end{figure}

If $U$ is a function on $K^{\mathbb{G}}$ we may consider $U\mathbb{\ }$to also
be a function on both $K^{\mathbb{G}_{+,\varepsilon}}$ by replacing the
argument $\omega\left(  e_{2}\right)  $ in $U$ by $\omega\left(
e_{2}^{\varepsilon}\right)  .$ Similarly by replacing the argument
$\omega\left(  e_{4}\right)  $ in $U$ by $\omega\left(  e_{4}^{\varepsilon
}\right)  $ we may also view $U$ as a function on $K^{\mathbb{G}%
_{-,\varepsilon}}.$
\end{notation}

\begin{assumption}
\label{ass.2}For the rest of this paper we now assume that $\left\{
Q_{\varepsilon}\right\}  $ as in Notation \ref{not.2.25} is chosen in such as
way that there exists $c<\infty$ so that the \textquotedblleft
shadow\textquotedblright\ region $S_{\varepsilon}$ (as indicated in Figure
\ref{fig.8}) satisfies $\left\vert S_{\varepsilon}\right\vert \leq c\left\vert
Q_{\varepsilon}\right\vert $ as $\varepsilon\downarrow0.$
\end{assumption}

\begin{theorem}
[Loop expansion]\label{thm.2.27}Continuing the notation above while keeping
Assumptions \ref{ass.1} and \ref{ass.2} in force, if $U:K^{\mathbb{G}%
}\rightarrow\mathbb{C}$ is a $C^{3}$-function, then
\begin{equation}
-\mathbb{E}\left[  \left(  \nabla_{f\left(  Q_{\varepsilon}\right)  }^{e_{2}%
}U+\nabla_{f\left(  RQ_{\varepsilon}\right)  }^{e_{4}}U\right)  \left(
\pt^{f}\left(  \mathbb{G}\right)  \right)  \right]  =\mathbb{E}\left[
U\left(  \pt^{f}\left(  \mathbb{G}_{+,\varepsilon}\right)  \right)  -U\left(
\pt^{f}\left(  \mathbb{G}_{-,\varepsilon}\right)  \right)  \right]  +O\left(
\sqrt{\varepsilon}\left\vert Q_{\varepsilon}\right\vert \right)  .
\label{e.2.10}%
\end{equation}

\end{theorem}

\subsection{Proof of Theorem \ref{thm.2.21}}

It is now a simple matter to use the previous results to prove Theorem
\ref{thm.2.21}.

\begin{proof}
[Proof of Theorem \ref{thm.2.21}]By Proposition \ref{pro.2.22} it suffices to
assume the crossing configuration of $\mathbb{G}$ is at $v=0$ and lies on the
coordinated axes as in Figure \ref{fig.6}. It then follows from Corollary
\ref{cor.2.24} and Theorem \ref{thm.2.27} that
\begin{align}
\mathbb{E}\left[  \left(  \nabla^{e_{1}}\cdot\nabla^{e_{2}}U\right)  \left(
\pt^{f}\left(  \mathbb{G}\right)  \right)  \right]   &  =-\lim_{\varepsilon
\downarrow0}\frac{1}{\left\vert Q_{\varepsilon}\right\vert }\mathbb{E}\left[
\left(  \nabla_{f\left(  Q_{\varepsilon}\right)  }^{e_{2}}U+\nabla_{f\left(
RQ_{\varepsilon}\right)  }^{e_{4}}U\right)  \left(  \pt^{f}\left(
\mathbb{G}\right)  \right)  \right]  \nonumber\\
&  =\lim_{\varepsilon\downarrow0}\frac{1}{\left\vert Q_{\varepsilon
}\right\vert }\mathbb{E}\left[  U\left(  \pt^{f}\left(  \mathbb{G}%
_{+,\varepsilon}\right)  \right)  -U\left(  \pt^{f}\left(  \mathbb{G}%
_{-,\varepsilon}\right)  \right)  \right]  .\label{e.2.11}%
\end{align}
Fixing all of the bounded areas of $\mathbb{R}^{2}\setminus\mathbb{G}$ except
for those bordering the vertex $0\in V\left(  \mathbb{G}\right)  ,$ let $z$ be
the function such that $z\left(  t_{1},t_{2},t_{3},t_{4}\right)
:=\mathbb{E}\left[  U\left(  \pt^{f}\left(  \mathbb{G}\right)  \right)
\right]  ,$ where $\left\{  t_{i}\right\}  _{i=1}^{4}$ are the areas of the
faces adjoining $0$ as labeled in Figure \ref{fig.6}. [The fact that we can
defined $z$ as a function of these areas and not on the shapes of the regions
relies on Corollary \ref{cor.2.15}.] We then have, see Figure \ref{fig.8},
\begin{align*}
\mathbb{E}\left[  U\left(  \pt^{f}\left(  \mathbb{G}_{+,\varepsilon}\right)
\right)  -U\left(  \pt^{f}\left(  \mathbb{G}_{-,\varepsilon}\right)  \right)
\right]  = &  z\left(  t_{1}-\left\vert Q_{\varepsilon}\right\vert
,t_{2}+\left\vert Q_{\varepsilon}\right\vert ,t_{3},t_{4}\right)  -z\left(
t_{1},t_{2},t_{3}+\left\vert Q_{\varepsilon}\right\vert ,t_{4}-\left\vert
Q_{\varepsilon}\right\vert \right)  \\
= &  z\left(  t_{1}-\left\vert Q_{\varepsilon}\right\vert ,t_{2}+\left\vert
Q_{\varepsilon}\right\vert ,t_{3},t_{4}\right)  -z\left(  t_{1},t_{2}%
,t_{3},t_{4}\right)  \\
&  \qquad-\left[  z\left(  t_{1},t_{2},t_{3}+\left\vert Q_{\varepsilon
}\right\vert ,t_{4}-\left\vert Q_{\varepsilon}\right\vert \right)  -z\left(
t_{1},t_{2},t_{3},t_{4}\right)  \right]  .
\end{align*}
Dividing this identity by $\left\vert Q_{\varepsilon}\right\vert $ and letting
$\varepsilon\downarrow0$ shows
\begin{align}
\lim_{\varepsilon\downarrow0}\frac{1}{\left\vert Q_{\varepsilon}\right\vert
}\mathbb{E}\left[  U\left(  \pt^{f}\left(  \mathbb{G}_{+,\varepsilon}\right)
\right)  -U\left(  \pt^{f}\left(  \mathbb{G}_{-,\varepsilon}\right)  \right)
\right]   &  =\left(  -\frac{\partial}{\partial t_{1}}+\frac{\partial
}{\partial t_{2}}-\frac{\partial}{\partial t_{3}}+\frac{\partial}{\partial
t_{4}}\right)  z\left(  t_{1},t_{2},t_{3},t_{4}\right)  \nonumber\\
&  =-\left(  \frac{\partial}{\partial t_{1}}-\frac{\partial}{\partial t_{2}%
}+\frac{\partial}{\partial t_{3}}-\frac{\partial}{\partial t_{4}}\right)
\mathbb{E}\left[  U\left(  \pt^{f}\left(  \mathbb{G}\right)  \right)  \right]
.\label{e.2.12}%
\end{align}
Combining Eqs. (\ref{e.2.11}) and (\ref{e.2.12}) completes the proof of
Theorem \ref{thm.2.21}.
\end{proof}

\iftrue

\section{Heuristic arguments\label{sec.3}}

Throughout this section, let $\mathbb{G}$ be a graph in $\mathbb{R}^{2}$
satisfying Assumption \ref{ass.1}, i.e. we are assuming $\mathbb{G}$ contains
a cross of bonds $\left(  \left\{  e_{1},e_{2},e_{3},e_{4}\right\}  \right)  $
contained in the coordinate axes as in Figure \ref{fig.6}. In this section we
will use the informal expression for the measure $\mu$ given in Eq.
(\ref{e.1.6}) and therefore many of the \textquotedblleft
results\textquotedblright\ in this section are not rigorous. We indicate the
non-rigorous results by writing Meta-Theorem, Meta-Lemma, etc. Most of the
results in this section will have a corresponding rigorous version in either
Section \ref{sec.4} or Section \ref{sec.5} below.

\subsection{Heuristic integration by parts\label{sec.3.1}}

\begin{metatheorem}
[Gaussian IBP]\label{mthm.3.1}If $\Psi:\mathcal{A}_{0}\rightarrow\mathbb{C}$
and $\eta\in\mathcal{A}_{0},$ then%
\begin{equation}
\int_{\mathcal{A}_{0}}\Psi\left(  A+\eta\right)  d\mu\left(  A\right)
=e^{-\frac{1}{2}\left\Vert \partial_{y}\eta\right\Vert ^{2}}\int
_{\mathcal{A}_{0}}\Psi\left(  A\right)  \exp\left(  \left\langle \partial
_{y}A,\partial_{y}\eta\right\rangle \right)  d\mu\left(  A\right)
\label{e.3.1}%
\end{equation}
and%
\begin{equation}
\int_{\mathcal{A}_{0}}\partial_{\eta}\Psi\left(  A\right)  d\mu\left(
A\right)  =\int_{\mathcal{A}_{0}}\Psi\left(  A\right)  \left\langle
\partial_{y}A,\partial_{y}\eta\right\rangle d\mu\left(  A\right)  .
\label{e.3.2}%
\end{equation}

\end{metatheorem}

\begin{proof}
Using the formal translation invariance of $\mathcal{D}A$ we find
\begin{align*}
\int_{\mathcal{A}_{0}}\Psi\left(  A+\eta\right)  d\mu\left(  A\right)   &
=\frac{1}{Z}\int_{\mathcal{A}_{0}}\Psi\left(  A+\eta\right)  \exp\left(
-\frac{1}{2}\left\Vert \partial_{y}A\right\Vert ^{2}\right)  \mathcal{D}A\\
&  =\frac{1}{Z}\int_{\mathcal{A}_{0}}\Psi\left(  A\right)  \exp\left(
-\frac{1}{2}\left\Vert \partial_{y}A-\partial_{y}\eta\right\Vert ^{2}\right)
\mathcal{D}A\\
&  =e^{-\frac{1}{2}\left\Vert \partial_{y}\eta\right\Vert ^{2}}\int
_{\mathcal{A}_{0}}\Psi\left(  A\right)  \exp\left(  \left\langle \partial
_{y}A,\partial_{y}\eta\right\rangle \right)  d\mu\left(  A\right)  .
\end{align*}
Replacing $\eta$ by $s\eta$ and then differentiating in $s$ at $s=0$ then
gives the basic Gaussian integration by parts formula in Eq. (\ref{e.3.2}).
\end{proof}

Now we want to make perturbations by $\eta$ (i.e. $\eta dx)$ where in fact
$\eta\left(  x,0\right)  \neq0$ for all $x\in\mathbb{R}.$ Of course we can not
actually do this since this type of perturbation does not preserve the
axial-gauge subspace, $\mathcal{A}_{0}.$ Nevertheless we will see that Eq.
(\ref{e.3.2}) still is valid for such an $\eta$ provided $\Psi$ is $g_{s\eta}%
$-invariant for $s$ near $0,$ see Meta-Theorem \ref{mthm.3.10}. Before proving
this key meta-theorem we need to introduce the general class of perturbations
to be considered.

\begin{notation}
\label{not.3.2}Let $\eta_{y}:\mathbb{R}^{2}\rightarrow\mathfrak{k}$ be a
bounded measurable function with compact support and set
\begin{equation}
\eta\left(  x,y\right)  :=\int_{-\infty}^{y}\eta_{y}\left(  x,y^{\prime
}\right)  dy^{\prime}\text{ }\forall~\left(  x,y\right)  \in\mathbb{R}^{2}.
\label{e.3.3}%
\end{equation}

\end{notation}

Notice that for each $x\in\mathbb{R},$ $y\rightarrow\eta\left(  x,y\right)  $
is absolutely continuous and $\frac{\partial\eta}{\partial y}\left(
x,y\right)  =\eta_{y}\left(  x,y\right)  $ for a.e. $y.$ Because of this
observation, we will often informally describe $\eta$ as in Notation
\ref{not.3.2} by saying that $\eta$ is a function from $\mathbb{R}^{2}$ to
$\mathfrak{k}$ such that $\eta_{y}=\partial\eta/\partial y$ is bounded and
compactly supported with the understanding that $\eta$ is given as in Eq.
(\ref{e.3.3}).

\begin{notation}
\label{not.3.3}Given $A:\mathbb{R}^{2}\rightarrow\mathfrak{k}$ bounded and
measurable, let
\[
\bar{A}\left(  x,y\right)  :=A\left(  x,y\right)  -A\left(  x,0\right)  .
\]

\end{notation}

\begin{remark}
\label{rem.3.4}If $\eta_{y}$ and $\eta$ are as in Notation \ref{not.3.2}, then
for all $y\in\mathbb{R},$%
\begin{align*}
\bar{\eta}\left(  x,y\right)   &  =\eta\left(  x,y\right)  -\eta\left(
x,0\right) \\
&  =\int_{-\infty}^{y}\eta_{y}\left(  x,y^{\prime}\right)  dy^{\prime}%
-\int_{-\infty}^{0}\eta_{y}\left(  x,y^{\prime}\right)  dy^{\prime}=\int
_{0}^{y}\eta_{y}\left(  x,y^{\prime}\right)  dy^{\prime}.
\end{align*}

\end{remark}

The next example contains the only class of $\eta$'s that are actually needed
for the purposes of this paper.

\begin{example}
\label{ex.3.5}For $\xi\in\mathfrak{k}$ and $Q$ a compact region in the first
quadrant as shown in Figure \ref{fig.7}, let $\eta$ be as in Eq. (\ref{e.3.3})
with
\[
\eta_{y}\left(  x,y\right)  =\left[  1_{RQ}\left(  x,y\right)  -1_{Q}\left(
x,y\right)  \right]  \cdot\xi.
\]
In this case, $\eta$ is compactly supported with $\mathrm{supp}\left(
\eta\right)  =\bar{Q}\cup R\bar{Q}.$
\end{example}

\begin{definition}
\label{def.3.6}To each $\eta$ as in Notation \ref{not.3.2}, let $g_{\eta}\in
C\left(  \mathbb{R},K\right)  $ be the absolutely continuous function
satisfying the ODE%
\begin{equation}
\frac{d}{dx}g_{\eta}\left(  x\right)  +\eta\left(  x,0\right)  g_{\eta}\left(
x\right)  =0\text{ for a.e. }x\text{ with }g_{\eta}\left(  0\right)  =I.
\label{e.3.4}%
\end{equation}
[Equation (\ref{e.3.4}) should be interpreted in integral form as%
\[
g_{\eta}\left(  x\right)  =I-\int_{0}^{x}\eta\left(  x^{\prime},0\right)
g_{\eta}\left(  x^{\prime}\right)  dx^{\prime}\text{ }\forall~x\in\mathbb{R}.
\]
This integral equation has a unique absolutely continuous solution.]
\end{definition}

For our purposes, we will only deal with $g\in C\left(  \mathbb{R}%
^{2},K\right)  $ such that $g\left(  x,y\right)  $ is independent of $y$ and
in this setting we will identify $g$ with $g\left(  \cdot,0\right)  \in
C\left(  \mathbb{R},K\right)  .$ Thus we will abuse notation and write
\textquotedblleft$g\left(  x,y\right)  =g\left(  x\right)  $\textquotedblright%
\ for any $g\in C\left(  \mathbb{R},K\right)  $ and in particular apply this
identification to $g_{\eta}$ of Definition \ref{def.3.6}.

\begin{proposition}
\label{pro.3.7}If $A\,dx$ is a connection one form, $\bar{A}$ is as in
Notation \ref{not.3.3}, and $g_{A}$ is a in Definition \ref{def.3.6} with
$\eta$ replaced by $A,$ then%
\begin{align*}
\left(  A\,dx\right)  ^{g_{A}}\left(  x,y\right)   &  =\mathrm{Ad}%
_{g_{A}\left(  x\right)  ^{-1}}\bar{A}\left(  x,y\right)  dx\in\mathcal{A}%
_{0},\text{ and }\\
f^{\left(  A\,dx\right)  ^{g_{A}}}\left(  x,y\right)   &  =\mathrm{Ad}%
_{g_{A}\left(  x\right)  ^{-1}}f^{A\,dx}\left(  x,y\right)  .
\end{align*}

\end{proposition}

\begin{proof}
From the definitions we have%
\begin{align*}
\left(  A\,dx\right)  ^{g_{A}}  &  =\mathrm{Ad}_{g_{A}^{-1}}A\,dx+g_{A}%
^{-1}dg_{A}=\left[  \mathrm{Ad}_{g_{A}\left(  x\right)  ^{-1}}A\left(
x,y\right)  -\mathrm{Ad}_{g_{A}\left(  x\right)  ^{-1}}A\left(  x,0\right)
\right]  dx\\
&  =\mathrm{Ad}_{g_{A}\left(  x\right)  ^{-1}}\bar{A}\left(  x,y\right)  dx
\end{align*}
which proves the first equality. The second equality follows directly from
Theorem \ref{thm.A.1} of the appendix or may be proved directly as follows,%
\begin{align*}
f^{\left(  A\,dx\right)  ^{g_{A}}}\left(  x,y\right)   &  =-\partial
_{y}\left[  \mathrm{Ad}_{g_{A}\left(  x\right)  ^{-1}}\bar{A}\left(
x,y\right)  \right]  =-\mathrm{Ad}_{g_{A}\left(  x\right)  ^{-1}}\partial
_{y}\bar{A}\left(  x,y\right) \\
&  =-\mathrm{Ad}_{g_{A}\left(  x\right)  ^{-1}}\partial_{y}A\left(
x,y\right)  =\mathrm{Ad}_{g_{A}\left(  x\right)  ^{-1}}f^{A\,dx}\left(
x,y\right)  .
\end{align*}

\end{proof}

Finally we introduce a class of $K$-valued functions which enable us to
describe how parallel translation transforms under the perturbations,
$A\,dx\rightarrow\left[  \left(  A+\eta\right)  dx\right]  ^{g_{\eta}}.$

\begin{definition}
\label{def.3.8}For $A\in\mathcal{A}_{0}$ (i.e. $A\,dx\in\mathcal{A}_{0}),$
$\eta:\mathbb{R}^{2}\rightarrow\mathfrak{k}$ as in Notation \ref{not.3.2}, and
a horizontal path, $\left[  a,b\right]  \ni x\rightarrow\ell\left(  x\right)
=\left(  x,y\left(  x\right)  \right)  ,$ let $\left[  a,b\right]  \ni
x\rightarrow k_{x}\left(  \ell\right)  $ denote absolutely continuous function
satisfying,%
\begin{equation}
\frac{d}{dx}k_{x}\left(  \ell\right)  +\left[  \mathrm{Ad}_{\pt_{x}^{A}\left(
\ell\right)  ^{-1}}\eta\left(  x,y\left(  x\right)  \right)  \right]
k_{x}\left(  \ell\right)  =0\text{ a.e. }x\text{ with }k_{0}\left(
\ell\right)  =I \label{e.3.5}%
\end{equation}
and let $\mathbf{k}^{\left(  A,\eta\right)  }\left(  \ell\right)
:=k_{b}\left(  \ell\right)  .$ [Again this equation is to be interpreted in
its integral form.]
\end{definition}

\begin{corollary}
\label{cor.3.9}If $A\,dx\in\mathcal{A}_{0}$ (so $A\left(  x,0\right)
\equiv0),$ $\eta$ is as in Notation \ref{not.3.2}, and $g_{\eta}$ is as in
Definition \ref{def.3.6}, then%
\begin{align}
\left[  \left(  A+\eta\right)  dx\right]  ^{g_{\eta}}  &  =\mathrm{Ad}%
_{g_{\eta}\left(  x\right)  ^{-1}}\left[  A\left(  x,y\right)  +\bar{\eta
}\left(  x,y\right)  \right]  dx\in\mathcal{A}_{0},\text{ and}\label{e.3.6}\\
f_{\eta}\left(  x,y\right)   &  :=f^{\left[  \left(  A+\eta\right)  dx\right]
^{g_{\eta}}}\left(  x,y\right)  =\mathrm{Ad}_{g_{\eta}\left(  x\right)  ^{-1}%
}\left[  f\left(  x,y\right)  -\eta_{y}\left(  x,y\right)  \right]  .
\label{e.3.7}%
\end{align}
Moreover, if $\left[  a,b\right]  \ni x\rightarrow\ell\left(  x\right)
=\left(  x,y\left(  x\right)  \right)  \in\mathbb{R}^{2}$ is a horizontal
curve, then
\begin{equation}
\pt^{\left[  \left(  A+\eta\right)  dx\right]  ^{g_{\eta}}}\left(
\ell\right)  =g_{\eta}\left(  b\right)  ^{-1}\pt_{b}^{A}\left(  \ell\right)
\mathbf{k}^{\left(  A,\eta\right)  }\left(  \ell\right)  g_{\eta}\left(
a\right)  . \label{e.3.8}%
\end{equation}

\end{corollary}

\begin{proof}
Most of this corollary is an easy consequence of Proposition \ref{pro.3.7}
upon observing that
\[
\overline{A+\eta}\left(  x,y\right)  =A\left(  x,y\right)  +\bar{\eta}\left(
x,y\right)  \text{ and }~\left(  A+\eta\right)  \left(  x,0\right)
=\eta\left(  x,0\right)
\]
and therefore; $g_{A+\eta}\left(  x,y\right)  =g_{\eta}\left(  x\right)  ,$%
\[
\left[  \left(  A+\eta\right)  dx\right]  ^{g_{A+\eta}}=\mathrm{Ad}_{g_{\eta
}\left(  x\right)  ^{-1}}\left[  A\left(  x,y\right)  +\bar{\eta}\left(
x,y\right)  \right]  dx,
\]
and%
\[
f_{\eta}\left(  x,y\right)  :=-\partial_{y}\left(  \mathrm{Ad}_{g_{\eta
}\left(  x\right)  ^{-1}}\left[  A\left(  x,y\right)  +\bar{\eta}\left(
x,y\right)  \right]  \right)  =\mathrm{Ad}_{g_{\eta}\left(  x\right)  ^{-1}%
}\left[  f\left(  x,y\right)  -\partial_{y}\eta\left(  x,y\right)  \right]  .
\]
So it only remains to prove Eq. (\ref{e.3.8}).

Differentiating the identity, $\pt_{x}^{A}\left(  \ell\right)  ^{-1}%
\pt_{x}^{A}\left(  \ell\right)  =I$ while making use of the Definition
\ref{def.1.2} shows%
\[
\frac{d}{dx}\pt_{x}^{A}\left(  \ell\right)  ^{-1}=\pt_{x}^{A}\left(
\ell\right)  ^{-1}A\left(  x,y\left(  x\right)  \right)
\]
and therefore,%
\begin{align*}
\frac{d}{dx}  &  \left(  \pt_{x}^{A}\left(  \ell\right)  ^{-1}\pt_{x}^{\left(
A+\eta\right)  dx}\left(  \ell\right)  \right) \\
&  =\pt_{x}^{A}\left(  \ell\right)  ^{-1}\left(  A\left(  x,y\left(  x\right)
\right)  -\left[  A\left(  x,y\left(  x\right)  \right)  +\eta\left(
x,y\left(  x\right)  \right)  \right]  \right)  \pt_{x}^{\left(
A+\eta\right)  dx}\left(  \ell\right) \\
&  =-\left(  \mathrm{Ad}_{\pt_{x}^{A}\left(  \ell\right)  ^{-1}}\eta\left(
x,y\left(  x\right)  \right)  \right)  \left(  \pt_{x}^{A}\left(  \ell\right)
^{-1}\pt_{x}^{\left(  A+\eta\right)  dx}\left(  \ell\right)  \right)
\end{align*}
with $\pt_{x}^{A}\left(  \ell\right)  ^{-1}\pt_{x}^{\left(  A+\eta\right)
dx}\left(  \ell\right)  |_{x=a}=I.$ By uniqueness of solutions to ODEs we
conclude that
\[
\pt_{x}^{A}\left(  \ell\right)  ^{-1}\pt_{x}^{\left(  A+\eta\right)
dx}\left(  \ell\right)  =k_{x}\left(  \ell\right)  \implies\pt_{b}^{\left(
A+\eta\right)  dx}\left(  \ell\right)  =\pt_{b}^{A}\left(  \ell\right)
\mathbf{k}^{\left(  A,\eta\right)  }\left(  \ell\right)  .
\]
This equation along with Eq. (\ref{e.A.2}) of Appendix \ref{sec.A} then gives
Eq. (\ref{e.3.8}).

\end{proof}

\begin{metatheorem}
\label{mthm.3.10}If $\Psi:\mathcal{A}\rightarrow\mathbb{C}$ is an
\textquotedblleft integrable function\textquotedblright\ and $\eta=\eta dx$ is
a connection one form (we do not assume $\eta\left(  x,0\right)  =0$ for all
$x\in\mathbb{R}),$ then
\begin{equation}
\int_{\mathcal{A}_{0}}\Psi\left(  \left(  A+\eta\right)  ^{g_{\eta}}\right)
d\mu\left(  A\right)  =e^{-\frac{1}{2}\left\Vert \partial_{y}\eta\right\Vert
^{2}}\int_{\mathcal{A}_{0}}\Psi\left(  A\right)  \exp\left(  \left\langle
\partial_{y}A,\mathrm{Ad}_{g_{\eta}^{-1}}\partial_{y}\eta\right\rangle
\right)  d\mu\left(  A\right)  . \label{e.3.9}%
\end{equation}

\end{metatheorem}

\begin{proof}
By Corollary \ref{cor.3.9},%
\[
\left(  A+\eta\right)  ^{g_{\eta}}=\mathrm{Ad}_{g_{\eta}^{-1}}\left[
A+\bar{\eta}\right]  \in\mathcal{A}_{0},
\]
where $\bar{\eta}\left(  x,y\right)  =\eta\left(  x,y\right)  -\eta\left(
x,0\right)  .$ Because $g_{\eta}$ depends only on $x$ and $\mathrm{Ad}%
_{g_{\eta}^{-1}}$ acts isometrically on $\mathfrak{k}$ it follows that
$d\mu\left(  A\right)  $ is invariant under $A\rightarrow\mathrm{Ad}_{g_{\eta
}^{-1}}A$ and we conclude, with the aid of Eq. (\ref{e.3.1}) and $\partial
_{y}\bar{\eta}=\partial_{y}\eta,$ that%
\begin{align*}
\int_{\mathcal{A}_{0}}\Psi\left(  \left(  A+\eta\right)  ^{g_{\eta}}\right)
d\mu\left(  A\right)   &  =\int_{\mathcal{A}_{0}}\Psi\left(  \mathrm{Ad}%
_{g_{\eta}^{-1}}\left[  A+\bar{\eta}\right]  \right)  d\mu\left(  A\right) \\
&  =e^{-\frac{1}{2}\left\Vert \partial_{y}\bar{\eta}\right\Vert ^{2}}%
\int_{\mathcal{A}_{0}}\Psi\left(  \mathrm{Ad}_{g_{\eta}^{-1}}A\right)
\exp\left(  \left\langle \partial_{y}A,\partial_{y}\bar{\eta}\right\rangle
\right)  d\mu\left(  A\right) \\
&  =e^{-\frac{1}{2}\left\Vert \partial_{y}\eta\right\Vert ^{2}}\int
_{\mathcal{A}_{0}}\Psi\left(  A\right)  \exp\left(  \left\langle \partial
_{y}\mathrm{Ad}_{g_{\eta}}A,\partial_{y}\eta\right\rangle \right)  d\mu\left(
A\right) \\
&  =e^{-\frac{1}{2}\left\Vert \partial_{y}\eta\right\Vert ^{2}}\int
_{\mathcal{A}_{0}}\Psi\left(  A\right)  \exp\left(  \left\langle \partial
_{y}A,\mathrm{Ad}_{g_{\eta}^{-1}}\partial_{y}\eta\right\rangle \right)
d\mu\left(  A\right)  .
\end{align*}
This proves Eq. (\ref{e.3.9}).
\end{proof}

\begin{metacorollary}
\label{mcor.3.11}If $\Psi:\mathcal{A}\rightarrow\mathbb{C}$ is
\textquotedblleft smooth\textquotedblright\ and $g_{s\eta}$-invariant in the
sense that $\Psi\left(  B^{g_{s\eta}}\right)  =\Psi\left(  B\right)  $ for all
$s$ near $0$ and $B\in\mathcal{A}$ with $B=Bdx$ ($B\left(  x,0\right)  $ not
assumed to be zero), then Eq. (\ref{e.3.2}) still holds, i.e.%
\begin{equation}
\int_{\mathcal{A}_{0}}\left(  \partial_{\eta}\Psi\right)  \left(  A\right)
d\mu\left(  A\right)  =\int_{\mathcal{A}_{0}}\Psi\left(  A\right)
\left\langle \partial_{y}A,\partial_{y}\eta\right\rangle d\mu\left(  A\right)
=-\int_{\mathcal{A}_{0}}\Psi\left(  A\right)  \left\langle f^{A},\partial
_{y}\eta\right\rangle d\mu\left(  A\right)  . \label{e.3.10}%
\end{equation}
[Note that we are substituting the assumption that $\eta\in\mathcal{A}_{0}$ by
the assumption that $\Psi$ is $g_{s\eta}$-invariant.]
\end{metacorollary}

\begin{proof}
[Meta-Proof]This result is a special case of Meta-Corollary \ref{mcor.B.32} of
Appendix \ref{sec.B}. Nevertheless we will give a second \textquotedblleft
proof\textquotedblright\ here which will be closer in line with the rigorous
proof of Corollary \ref{cor.4.12} below.

Replacing $\eta$ by $s\eta$ in Eq. (\ref{e.3.9}) and differentiating the
result leads to,%
\begin{align}
\frac{d}{ds}|_{0}\int_{\mathcal{A}_{0}}\Psi\left(  \left(  A+s\eta\right)
^{g_{s\eta}}\right)  d\mu\left(  A\right)   &  =\int_{\mathcal{A}_{0}}%
\Psi\left(  A\right)  \exp\left(  \left\langle \partial_{y}A,\partial_{y}%
\eta\right\rangle \right)  d\mu\left(  A\right) \nonumber\\
&  =\int_{\mathcal{A}_{0}}\Psi\left(  A\right)  \exp\left(  \left\langle
\partial_{y}A,\partial_{y}\eta\right\rangle \right)  d\mu\left(  A\right)  .
\label{e.3.11}%
\end{align}
If we now further assume that $\Psi$ is $g_{s\eta}$-invariant, then%
\[
\frac{d}{ds}|_{0}\Psi\left(  \left(  A+s\eta\right)  ^{g_{s\eta}}\right)
=\frac{d}{ds}|_{0}\Psi\left(  A+s\eta\right)  =\left(  \partial_{\eta}%
\Psi\right)  \left(  A\right)
\]
which combined with Eq. (\ref{e.3.11}) verifies Eq. (\ref{e.3.10}).

\end{proof}

\begin{notation}
\label{n.3.12}If $\eta$ is as in Notation \ref{not.3.2} and $\sigma=\left(
\sigma_{1},\sigma_{2}\right)  :\left[  a_{\sigma},b_{\sigma}\right]
\rightarrow\mathbb{R}^{2}$ is a piecewise $C^{1}$-path, let
\[
\zeta_{\eta}^{A}\left(  \sigma\right)  :=\int_{a_{\sigma}}^{b_{\sigma}%
}\mathrm{Ad}_{\pt_{t}^{A}\left(  \sigma\right)  }\eta\left(  \sigma\left(
t\right)  \right)  \dot{\sigma}_{1}\left(  t\right)  dt\in\mathfrak{k.}%
\]

\end{notation}

\begin{lemma}
\label{lem.3.13}If $\mathbb{G}$ is a tame graph, $U\in C^{1}\left(
K^{\mathbb{G}}\rightarrow\mathbb{C}\right)  ,$ and $\Psi\left(  A\right)
=U\left(  \pt^{A}\left(  \mathbb{G}\right)  \right)  ,$ then
\[
\left(  \partial_{\eta}\Psi\right)  \left(  A\right)  =-\left(  \tilde{\zeta
}_{\eta}^{A}U\right)  \left(  \pt^{A}\left(  \mathbb{G}\right)  \right)
=-\sum_{\sigma\in\mathbb{G}}\left(  \nabla_{\zeta_{\eta}^{A}\left(
\sigma\right)  }^{\sigma}U\right)  \left(  \pt^{A}\left(  \mathbb{G}\right)
\right)
\]
where%
\begin{equation}
\left(  \tilde{\zeta}_{\eta}^{A}U\right)  \left(  \pt^{A}\left(
\mathbb{G}\right)  \right)  :=\frac{d}{ds}|_{0}U\left(  \left\{
\pt^{A}\left(  \sigma\right)  e^{s\zeta_{\eta}^{A}\left(  \sigma\right)
}\right\}  _{\sigma\in\mathbb{G}}\right)  . \label{e.3.12}%
\end{equation}

\end{lemma}

\begin{proof}
By Proposition \ref{pro.A.6} of Appendix \ref{sec.A},
\begin{align*}
\partial_{\eta}\left[  A\rightarrow U\left(  \pt^{A}\left(  \mathbb{G}\right)
\right)  \right]   &  =\frac{d}{ds}|_{0}U\left(  \left\{  \pt^{A}\left(
\sigma\right)  e^{-s\zeta_{\eta}^{A}\left(  \sigma\right)  }\right\}
_{\sigma\in\mathbb{G}}\right) \\
&  =-\left(  \tilde{\zeta}_{\eta}^{A}U\right)  \left(  \pt^{A}\left(
\mathbb{G}\right)  \right) \\
&  =-\sum_{\sigma\in\mathbb{G}}\left(  \nabla_{\zeta_{\eta}^{A}\left(
\sigma\right)  }^{\sigma}U\right)  \left(  \pt^{A}\left(  \mathbb{G}\right)
\right)  ,
\end{align*}
where the second equality is a consequence of the chain rule.
\end{proof}

\begin{metatheorem}
\label{mthm.3.14}If $\eta$ is a $\mathfrak{k}$-valued one form on
$\mathbb{R}^{2},$ and $U:K^{\mathbb{G}}\rightarrow\mathbb{C}$ is a $g_{s\eta
}|_{V\left(  \mathbb{G}\right)  }$-invariant for $s$ near $0,$ then
\[
\mathbb{E}\left[  \sum_{\sigma\in\mathbb{G}}\left(  \nabla_{\zeta_{\eta}%
^{A}\left(  \sigma\right)  }^{\sigma}U\right)  \left(  \pt^{A}\left(
\mathbb{G}\right)  \right)  \right]  =\mathbb{E}\left[  U\left(
\pt^{A}\left(  \mathbb{G}\right)  \right)  \cdot\left\langle f^{A}%
,\frac{\partial\eta}{\partial y}\right\rangle \right]
\]
where $f^{A}=-\partial_{y}A.$

\end{metatheorem}

\begin{proof}
According to Meta-Corollary \ref{mcor.3.11}, under the given assumptions we
have%
\begin{align*}
\int_{\mathcal{A}_{0}}\partial_{\eta}\left[  A\rightarrow U\left(
\pt^{A}\left(  \mathbb{G}\right)  \right)  \right]  d\mu\left(  A\right)   &
=\int_{\mathcal{A}_{0}}\Psi\left(  A\right)  \left\langle \partial
_{y}A,\partial_{y}\eta\right\rangle d\mu\left(  A\right) \\
&  =-\int_{\mathcal{A}_{0}}\Psi\left(  A\right)  \left\langle f^{A}%
,\partial_{y}\eta\right\rangle d\mu\left(  A\right)  .
\end{align*}
This identity along with Lemma \ref{lem.3.13} completes the proof.
\end{proof}

\subsection{Heuristic proof of Theorem \ref{thm.2.23}\label{sec.3.2}}

We are now prepared to give a heuristic proof of Theorem \ref{thm.2.23} as a
corollary of Meta-Theorem \ref{mthm.3.14}.

\begin{proof}
[Heuristic proof of Theorem \ref{thm.2.23}]Let $h_{y}:\mathbb{R}%
^{2}\rightarrow\mathbb{R}$ be a bounded measurable function with compact
support,
\begin{equation}
h\left(  x,y\right)  :=\int_{-\infty}^{y}h_{y}\left(  x,y^{\prime}\right)
dy^{\prime}, \label{e.3.13}%
\end{equation}
$\xi\in\mathfrak{k,}$ and $\eta_{y}:=h_{y}\cdot\xi$ so that $\eta:=h\xi.$ By
Lemma \ref{lem.2.17}, $\nabla_{\xi}^{e_{2}}U$ is still invariant under
discrete gauge transformations, $u:V\left(  \mathbb{G}\right)  \rightarrow K,$
such that $u\left(  0\right)  =I.$ Thus we may apply Meta-Theorem
\ref{mthm.3.14} with $U$ replaced by $\nabla_{\xi}^{e_{2}}U,$ to find%
\begin{equation}
\mathbb{E}\left[  \sum_{\sigma\in\mathbb{G}}\left(  \nabla_{\zeta_{h\xi}%
^{A}\left(  \sigma\right)  }^{\sigma}\left(  \nabla_{\xi}^{e_{2}}U\right)
\right)  \left(  \pt^{A}\left(  \mathbb{G}\right)  \right)  \right]
=\mathbb{E}\left[  \left(  \nabla_{\xi}^{e_{2}}U\right)  \left(
\pt^{A}\left(  \mathbb{G}\right)  \right)  \cdot\left\langle f^{A},h_{y}%
\xi\right\rangle \right]  . \label{e.3.14}%
\end{equation}

Notice that $\zeta_{h\xi}^{A}\left(  \sigma\right)  =0$ for $\sigma\in\left\{
e_{2},e_{4}\right\}  $ since $\dot{\sigma}_{1}\left(  t\right)  =0$ whenever
$\sigma$ is a vertical path. Let us now further assume that $h$ is supported
in small neighborhood of $0$ so that $\zeta_{h\xi}^{A}\left(  \sigma\right)
=0$ unless $\sigma\in\left\{  e_{1},e_{3}\right\}  .$ Since $\pt_{t}%
^{A}\left(  e_{j}\right)  =I$ for $j\in\left\{  1,3\right\}  ,$ it follows
that
\begin{align*}
\zeta_{h\xi}^{A}\left(  e_{1}\right)   &  :=\left[  \int_{0}^{\infty}h\left(
t,0\right)  dt\right]  \xi=:\alpha_{1}\xi\text{ and }\\
\zeta_{h\xi}^{A}\left(  e_{3}\right)   &  :=\left[  \int_{0}^{-\infty}h\left(
t,0\right)  dt\right]  \xi=:\alpha_{3}\xi
\end{align*}
and so Eq. (\ref{e.3.14}) becomes,%
\[
\mathbb{E}\left[  \sum_{j=1,3}\alpha_{j}\left(  \nabla_{\xi}^{e_{j}}%
\nabla_{\xi}^{e_{2}}U\right)  \left(  \pt^{A}\left(  \mathbb{G}\right)
\right)  \right]  =\mathbb{E}\left[  \left(  \nabla_{\xi}^{e_{2}}U\right)
\left(  \pt^{A}\left(  \mathbb{G}\right)  \right)  \cdot\left\langle
f^{A},h_{y}\xi\right\rangle \right]  .
\]
Summing the last equation over $\xi$ in an orthonormal basis for
$\mathfrak{k}$ shows,%
\begin{equation}
\mathbb{E}\left[  \sum_{j=1,3}\alpha_{j}\left(  \nabla^{e_{j}}\cdot
\nabla^{e_{2}}U\right)  \left(  \pt^{A}\left(  \mathbb{G}\right)  \right)
\right]  =\mathbb{E}\left[  \left(  \nabla_{\left\langle f^{A},h_{y}%
\right\rangle }^{e_{2}}U\right)  \left(  \pt^{A}\left(  \mathbb{G}\right)
\right)  \right]  . \label{e.3.15}%
\end{equation}

Finally, we take $\eta$ as Example \ref{ex.3.5}, i.e. $\eta=h\xi$ where
\begin{equation}
h_{y}\left(  x,y\right)  =1_{RQ}\left(  x,y\right)  -1_{Q}\left(  x,y\right)
\label{e.3.16}%
\end{equation}
with $Q$ being a compact region in the first quadrant as shown in Figure
\ref{fig.7}. With this choice,%
\begin{equation}
h\left(  x,y\right)  =\int_{-\infty}^{y}\left[  1_{RQ}\left(  x,y^{\prime
}\right)  -1_{Q}\left(  x,y^{\prime}\right)  \right]  dy^{\prime},
\label{e.3.17}%
\end{equation}
$\alpha_{3}=0,$ and%
\[
\alpha_{1}=\int_{0}^{\infty}h\left(  x,0\right)  dx=\int_{0}^{\infty}%
dx\int_{-\infty}^{0}dy^{\prime}1_{RQ}\left(  x,y^{\prime}\right)  =m\left(
Q\right)  =\left\vert Q\right\vert .
\]
Combining these identities with Eq. (\ref{e.3.15}) shows%
\[
\left\vert Q\right\vert \cdot\mathbb{E}\left[  \left(  \nabla^{e_{1}}%
\cdot\nabla^{e_{2}}U\right)  \left(  \pt^{A}\left(  \mathbb{G}\right)
\right)  \right]  =\mathbb{E}\left[  \left(  \nabla_{\left[  f^{A}\left(
RQ\right)  -f^{A}\left(  Q\right)  \right]  }^{e_{2}}U\right)  \left(
\pt^{A}\left(  \mathbb{G}\right)  \right)  \right]
\]
which is equivalent to Eq. (\ref{e.2.8}).
\end{proof}

\subsection{Path expansions\label{sec.3.3}}

In order to deduce Theorem \ref{thm.2.27} from Theorem \ref{thm.2.23} we need
to approximate
\begin{equation}
f^{A}\left(  Q\right)  :=-\int_{Q}\partial_{y}A\left(  x,y\right)  dxdy.
\label{e.3.18}%
\end{equation}
in terms of parallel translation around the boundary of $Q.$ To be more
precise, let $\ell$ be the right boundary of $Q$ and for $\varepsilon>0$
(small) let $\varepsilon\ell$ and $\varepsilon Q$ be $\ell$ and $Q$ scaled by
$\varepsilon$ as depicted in Figure \ref{fig.9}. \begin{figure}[ptbh]
\centering
\par
\psize{1.0in} %
\executeiffilenewer{\GraphicsDirectoryellpath.svg}{\GraphicsDirectoryellpath.pdf}%
{inkscape -z -D --file=\GraphicsDirectoryellpath.svg --export-pdf=\GraphicsDirectoryellpath.pdf --export-latex}%
\input{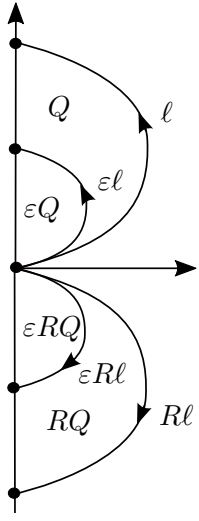}%
\caption{Here $Q$ ($\varepsilon Q)$ is the
region bounded by the path $\ell$ ($\varepsilon\ell)$ and the $y$-axis. }%
\label{fig.9}%
\end{figure}

\begin{notation}
[Right derivatives]\label{not.3.15}For $\psi\in C^{1}\left(  K,\mathbb{C}%
\right)  $ and $\xi\in\mathfrak{k}$ let $\hat{\nabla}_{\xi}\psi:K\rightarrow
\mathbb{C}$ be defined by,%
\begin{equation}
\left(  \hat{\nabla}_{\xi}\psi\right)  \left(  k\right)  :=\frac{d}{dt}%
|_{0}\psi\left(  e^{t\xi}k\right)  . \label{e.3.19}%
\end{equation}

\end{notation}

\begin{theorem}
\label{thm.3.16}Let $A\,dx$ be a smooth\footnote{For example, it would suffice
for $f=-\partial_{y}A$ to be continuous.} element of $\mathcal{A}_{0},$
$\varepsilon>0,$ and $\varepsilon\ell$ be the path shown in Figure
\ref{fig.9}. If $\psi\in C^{2}\left(  K,\mathbb{C}\right)  ,$ then%
\begin{equation}
\psi\left(  \pt_{1}^{A}\left(  \varepsilon\ell\right)  \right)  =\psi\left(
I\right)  -\left(  \nabla_{f^{A}\left(  \varepsilon Q\right)  }\psi\right)
\left(  e\right)  +O\left(  \varepsilon^{4}\right)  \label{e.3.20}%
\end{equation}
where $f^{A}\left(  \varepsilon Q\right)  $ is as in Eq. (\ref{e.3.18}) with
$Q$ replaced by $\varepsilon Q.$ Similarly if $R\left(  x,y\right)  :=\left(
x,-y\right)  $ is reflection across $x$-axis, then
\begin{equation}
\psi\left(  \pt_{1}^{A}\left(  \varepsilon R\ell\right)  \right)  =\psi\left(
I\right)  +\left(  \nabla_{f^{A}\left(  \varepsilon RQ\right)  }\psi\right)
\left(  e\right)  +O\left(  \varepsilon^{4}\right)  . \label{e.3.21}%
\end{equation}

\end{theorem}

\begin{proof}
Let us suppose $\ell$ is parametrized by $t\in\left[  0,1\right]  $ and let
$g\left(  t\right)  :=\pt_{t}^{A}\left(  \ell\right)  .$ Then $g\left(
t\right)  $ satisfies,%
\begin{equation}
\dot{g}\left(  t\right)  +\dot{\beta}_{\ell}\left(  t\right)  g\left(
t\right)  =0\text{ with }g\left(  0\right)  =I \label{e.3.22}%
\end{equation}
where%
\begin{align}
\beta_{\ell}\left(  t\right)   &  :=\int_{0}^{t}\left(  A\,dx\right)  \left(
\dot{\ell}\left(  \tau\right)  \right)  d\tau=-\int_{0}^{t}d\tau\dot{\ell}%
_{1}\left(  \tau\right)  \int_{0}^{\ell_{2}\left(  \tau\right)  }%
dyf^{A}\left(  \ell_{1}\left(  \tau\right)  ,y\right) \nonumber\\
&  =-\int_{0}^{t}d\tau\dot{\ell}_{1}\left(  \tau\right)  \ell_{2}\left(
\tau\right)  \int_{0}^{1}dsf^{A}\left(  \ell_{1}\left(  \tau\right)
,s\ell_{2}\left(  \tau\right)  \right)  . \label{e.3.23}%
\end{align}
By the fundamental theorem of calculus along with the ODE for $g$ in Eq.
(\ref{e.3.22}) we find,%
\begin{align}
\psi\left(  g\left(  t\right)  \right)   &  =\psi\left(  e\right)  +\int
_{0}^{t}\frac{d}{d\tau}\psi\left(  g\left(  \tau\right)  \right)
d\tau\nonumber\\
&  =\psi\left(  e\right)  -\int_{0}^{t}\left(  \hat{\nabla}_{\dot{\beta}%
_{\ell}\left(  \tau\right)  }\psi\right)  \left(  g\left(  \tau\right)
\right)  d\tau. \label{e.3.24}%
\end{align}
Applying Eq. (\ref{e.3.24}) with $\psi$ replaced by $\hat{\nabla}_{\dot{\beta
}_{\ell}\left(  \tau\right)  }\psi$ shows
\[
\left(  \hat{\nabla}_{\dot{\beta}_{\ell}\left(  \tau\right)  }\psi\right)
\left(  g\left(  \tau\right)  \right)  =\left(  \hat{\nabla}_{\dot{\beta
}_{\ell}\left(  \tau\right)  }\psi\right)  \left(  e\right)  -\int_{0}^{\tau
}ds\left(  \hat{\nabla}_{\dot{\beta}_{\ell}\left(  s\right)  }\hat{\nabla
}_{\dot{\beta}_{\ell}\left(  \tau\right)  }\psi\right)  \left(  g\left(
s\right)  \right)
\]
and then substituting this expression back into Eq. (\ref{e.3.24}) implies
(taking $t=1)$ that
\begin{equation}
\psi\left(  \pt_{1}^{A}\left(  \ell\right)  \right)  =\psi\left(  e\right)
-\left(  \hat{\nabla}_{\beta_{\ell}\left(  1\right)  }\psi\right)  \left(
e\right)  +\int_{0}^{1}d\tau\int_{0}^{\tau}ds\left(  \hat{\nabla}_{\dot{\beta
}_{\ell}\left(  s\right)  }\hat{\nabla}_{\dot{\beta}_{\ell}\left(
\tau\right)  }\psi\right)  \left(  g\left(  s\right)  \right)  .
\label{e.3.25}%
\end{equation}
Using the fact that $A\,dx\equiv0$ on the $y$-axis along with Green's (or
Stokes') theorem it follows that%
\begin{align}
\beta_{\ell}\left(  1\right)   &  =\int_{0}^{1}\left(  A\,dx\right)  \left(
\dot{\ell}\left(  t\right)  \right)  dt=\int_{\partial Q}A\,dx\nonumber\\
&  =\int_{Q}F^{A\,dx}=\int_{Q}f^{A}\left(  x,y\right)  dxdy=f^{A}\left(
Q\right)  . \label{e.3.26}%
\end{align}

Replacing $\ell$ by $\varepsilon\ell$ in Eqs. (\ref{e.3.23}), (\ref{e.3.25}),
and (\ref{e.3.26}) then gives Eq. (\ref{e.3.20}). Equation (\ref{e.3.21}) is
proved similarly noting that $R$ takes counterclockwise loops to clockwise
loops which changes a sign in Green's theorem. This then explains the change
of sign in the gradient term when passing from Eq. (\ref{e.3.20}) to Eq.
(\ref{e.3.21}).
\end{proof}

\subsection{Heuristic proof of Theorem \ref{thm.2.27}}

We are now almost ready to give a heuristic argument of Theorem \ref{thm.2.27}%
. In order to apply Theorem \ref{thm.3.16}, let $e_{2}^{\varepsilon}$ denote
the perturbation of $e_{2}$ consisting of traversing the path $\varepsilon
\ell$ followed by the straight line vertical path from $\varepsilon\ell\left(
1\right)  $ to $\ell\left(  1\right)  $ as shown in Figure \ref{fig.10}.
Similarly let $e_{4}^{\varepsilon}=Re_{2}^{\varepsilon}$ be the reflection of
$e_{2}^{\varepsilon}$ so that $e_{4}^{\varepsilon}$ is a perturbation of
$e_{4}.$ In order to simplify notation also let $\pt^{A}\left(  \mathbb{G}%
_{\pm},\varepsilon\right)  =\left\{  \pt_{1}^{A}\left(  \sigma\right)
:\sigma\in\mathbb{G}_{\pm},\varepsilon\right\}  .$\begin{figure}[ptbh]
\centering
\par
\psize{1.5in} %
\executeiffilenewer{\GraphicsDirectorysweep2a.svg}{\GraphicsDirectorysweep2a.pdf}%
{inkscape -z -D --file=\GraphicsDirectorysweep2a.svg --export-pdf=\GraphicsDirectorysweep2a.pdf --export-latex}%
\input{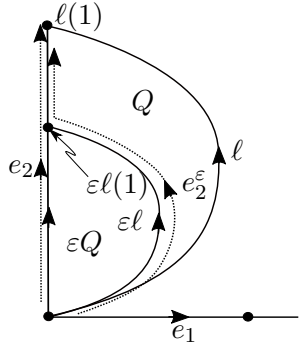}%
\caption{This figure shows the paths $e_{2}$
being deformed by the path $e_{2}^{\varepsilon}$ which consists of
$\varepsilon\ell_{1}$ followed by the straight line vertical path going from
$\varepsilon a$ to $a.$ We will also consider the reflection of this path
across the $x$-axis, $e_{4}^{\varepsilon}:=Re_{2}^{\varepsilon}.$}%
\label{fig.10}%
\end{figure}

\begin{corollary}
\label{cor.3.17}Let $\mathbb{G}_{\pm,\varepsilon}$ be the perturbations of
$\mathbb{G}$ described in Notation \ref{not.2.26} where $e_{2}^{\varepsilon}$
and $e_{4}^{\varepsilon}$ are the perturbations of $e_{2}$ and $e_{4}$
described above. If $A\,dx$ is a smooth element of $\mathcal{A}_{0}$ and
$U:K^{\mathbb{G}}\rightarrow\mathbb{C}$ is a $C^{2}$-function, then%
\begin{equation}
U\left(  \pt^{A}\left(  \mathbb{G}_{+,\varepsilon}\right)  \right)  -U\left(
\pt^{A}\left(  \mathbb{G}_{-,\varepsilon}\right)  \right)  =-\left(
\nabla_{f^{A}\left(  \varepsilon Q\right)  }^{e_{2}}U+\nabla_{f^{A}\left(
\varepsilon RQ\right)  }^{e_{4}}U\right)  \left(  \pt^{A}\left(
\mathbb{G}\right)  \right)  +O\left(  \varepsilon^{4}\right)  \label{e.3.27}%
\end{equation}

\end{corollary}

\begin{proof}
Freezing $\omega\left(  \sigma\right)  \in K$ for $\sigma\in\mathbb{G}%
\setminus\left\{  e_{2},e_{4}\right\}  $ and letting $\psi_{\pm}\in
C^{2}\left(  K,\mathbb{C}\right)  $ be chosen so that
\begin{align*}
\psi_{+}\left(  \omega\left(  e_{2}\right)  \right)   &  =U\left(
\omega\right)  \text{ with }\omega\left(  e_{4}\right)  :=e\text{ and}\\
\psi_{-}\left(  \omega\left(  e_{4}\right)  \right)   &  =U\left(
\omega\right)  \text{ with }\omega\left(  e_{2}\right)  :=e.
\end{align*}
It then follows from Eq. (\ref{e.3.20}) with $\psi=\psi_{+}$ and Eq.
(\ref{e.3.21}) with $\psi=\psi_{-}$ that%
\begin{align*}
U\left(  \pt^{A}\left(  \mathbb{G}_{+,\varepsilon}\right)  \right)   &
=U\left(  \pt^{A}\left(  \mathbb{G}\right)  \right)  -\left(  \nabla
_{f^{A}\left(  \varepsilon Q\right)  }^{e_{2}}U\right)  \left(  \pt^{A}\left(
\mathbb{G}\right)  \right)  +O\left(  \varepsilon^{4}\right)  \text{ and}\\
U\left(  \pt^{A}\left(  \mathbb{G}_{-,\varepsilon}\right)  \right)   &
=U\left(  \pt^{A}\left(  \mathbb{G}\right)  \right)  +\left(  \nabla
_{f^{A}\left(  \varepsilon RQ\right)  }^{e_{4}}U\right)  \left(
\pt^{A}\left(  \mathbb{G}\right)  \right)  +O\left(  \varepsilon^{4}\right)
\end{align*}
from which Eq. (\ref{e.3.27}) easily follows.
\end{proof}

\begin{proof}
[Heuristic Proof of Theorem \ref{thm.2.27}]By taking expectations of Eq.
(\ref{e.3.27}) we expect
\begin{equation}
\mathbb{E}\left[  U\left(  \pt^{A}\left(  \mathbb{G}_{+,\varepsilon}\right)
\right)  -U\left(  \pt^{A}\left(  \mathbb{G}_{-,\varepsilon}\right)  \right)
\right]  =-\mathbb{E}\left[  \left(  \nabla_{f^{A}\left(  \varepsilon
Q\right)  }^{e_{2}}U+\nabla_{f^{A}\left(  \varepsilon RQ\right)  }^{e_{4}%
}U\right)  \left(  \pt^{A}\left(  \mathbb{G}\right)  \right)  \right]
+O\left(  \varepsilon^{4}\right)  \label{e.3.28}%
\end{equation}
which would imply Theorem \ref{thm.2.27} in the setting described here as
$\left\vert \varepsilon Q\right\vert =\varepsilon^{2}\left\vert Q\right\vert $
so that $O\left(  \varepsilon^{4}\right)  =O\left(  \left\vert \varepsilon
Q\right\vert ^{2}\right)  =o\left(  \left\vert \varepsilon Q\right\vert
\right)  .$
\end{proof}

\begin{remark}
\label{rem.3.18}The heuristic \textquotedblleft proof\textquotedblright\ given
above follows the spirit of the arguments in \cite{MM}. However, as we will
see below the argument is too naive since the random connection one-forms
$A\,dx$ are very rough. In reality, $f^{A}\left(  Q\right)  $ fluctuates on
order of $\sqrt{\left\vert Q\right\vert }$ and the error term in Eqs.
(\ref{e.3.20}) and (\ref{e.3.21}) are really $O\left(  \left\vert Q\right\vert
\right)  $ rather than $o\left(  \left\vert Q\right\vert \right)  ,$ see
Theorem \ref{thm.5.4} below for a precise statement. Nevertheless we will see
below in subsection \ref{sec.5.1} that Theorem \ref{thm.2.27} is in fact true
because of a fortuitous cancellation and the simple covariance estimate, Lemma
\ref{lem.3.19} below.
\end{remark}

\fi

\section{Rigorous integration by parts\label{sec.4}}

Our goal in this section is to give a rigorous stochastic proof of Theorem
\ref{thm.2.23}. In order to prove the required integration by parts formula it
is necessary to understand the distribution of a white noise after it has been
transformed by rotations and translations. The key result is Corollary
\ref{cor.4.12} which specializes the abstract white noise result reviewed in
Theorem \ref{thm.4.5}.

\subsection{Rotating and translating white noise\label{sec.4.1}}

We begin by formally describing rotations of the white noise. If
$\mathcal{O}:L^{2}\left(  \mathbb{R}^{2},m;\mathfrak{k}\right)  \rightarrow
L^{2}\left(  \mathbb{R}^{2},m;\mathfrak{k}\right)  $ is any orthogonal
transformation and $f$ is a white noise then we define a new white noise,
$f^{\mathcal{O}},$ by%
\[
\left\langle f^{\mathcal{O}},u\right\rangle :=\left\langle f,\mathcal{O}%
u\right\rangle \text{ for all }u\in L^{2}\left(  \mathbb{R}^{2},m;\mathfrak{k}%
\right)  .
\]
If $\mathcal{R}:L^{2}\left(  \mathbb{R}^{2},m;\mathfrak{k}\right)  \rightarrow
L^{2}\left(  \mathbb{R}^{2},m;\mathfrak{k}\right)  $ is another isometry then
\[
\left\langle \left(  f^{\mathcal{R}}\right)  ^{\mathcal{O}},u\right\rangle
:=\left\langle f^{\mathcal{R}},\mathcal{O}u\right\rangle =\left\langle
f,\mathcal{RO}u\right\rangle =\left\langle f^{\mathcal{RO}},u\right\rangle
\]
from which it follows that $\left(  f,\mathcal{O}\right)  \rightarrow
f^{\mathcal{O}}$ is a right action. It should be clear that $f^{\mathcal{O}}$
and $f$ have the same distributions as mean zero Gaussian processes (like the
white noise) are completely determined by their covariances. We will be
interested here only in two special cases of this construction. The first is
the transformation, $u\rightarrow\hat{u}$ and correspondingly $f\rightarrow
\hat{f}$ given in Definition \ref{def.2.4} above and the second is given in
then next definition.

\begin{definition}
\label{def.4.1}If $g\in C\left(  \mathbb{R}^{2},K\right)  ,$ then we let
$f^{g}=f^{\mathrm{Ad}_{g}},$ i.e. $\left\langle f^{g},u\right\rangle
=\left\langle f,\mathrm{Ad}_{g}u\right\rangle $ for all $u\in L^{2}\left(
\mathbb{R}^{2},m;\mathfrak{k}\right)  .$
\end{definition}

If the white noise were a continuous process, then we would have $f^{g}\left(
p\right)  =\mathrm{Ad}_{g^{-1}\left(  p\right)  }f\left(  p\right)  $ for all
$p\in\mathbb{R}^{2}.$ Lemma \ref{lem.4.3} below gives a rigorous
interpretation of this informal representation of $f^{g}.$

\begin{notation}
[Oscillation semi-norms]\label{not.4.2}Suppose that $B\in\mathcal{B}%
_{\mathbb{R}^{2}}$ is a bounded set, $g\in C^{2}\left(  \mathbb{R}%
^{2},K\right)  ,$ and $\Pi\subset\mathcal{B}_{\mathbb{R}^{2}}$ denotes a
finite partition of $B.$ Then we let
\begin{align*}
\left\vert \Pi\right\vert  &  :=\max\left\{  \operatorname{diam}\left(
A\right)  :A\in\Pi\right\}  \text{ and }\\
\operatorname{osc}_{\Pi}\left(  g\right)   &  :=\max_{A\in\Pi}\sup_{p,q\in
A}\left\vert g\left(  q\right)  -g\left(  p\right)  \right\vert .
\end{align*}

\end{notation}

\begin{lemma}
\label{lem.4.3}Let $g\in C\left(  \mathbb{R}^{2},K\right)  ,$ $B\in
\mathcal{B}_{\mathbb{R}^{2}}$ be a bounded set, and $\left\{  \Pi_{n}\right\}
_{n=1}^{\infty}\subset\mathcal{B}_{\mathbb{R}^{2}}$ be a sequence of finite
partitions of $B$ such that $\lim_{n\rightarrow\infty}\operatorname{osc}%
_{\Pi_{n}}\left(  g\right)  =0,$ then
\begin{equation}
f^{g}\left(  B\right)  =L^{2}\left(  \mathbb{P}\right)  \text{-}%
\lim_{n\rightarrow\infty}\sum_{A\in\Pi_{n}}\mathrm{Ad}_{g\left(  p_{A}\right)
^{-1}}f\left(  A\right)  \label{e.4.1}%
\end{equation}
where $p_{A}$ denotes any choice of a point in $A$ for all $A\in\cup
_{n=1}^{\infty}\Pi_{n}.$
\end{lemma}

\begin{proof}
Let $\xi\in\mathfrak{k}$ be fixed so that
\[
\left\langle f^{g}\left(  B\right)  ,\xi\right\rangle :=\left\langle f^{g}%
,\xi1_{B}\right\rangle :=\left\langle f,\mathrm{Ad}_{g}\xi1_{B}\right\rangle
.
\]
By the assumption, $\lim_{n\rightarrow\infty}\operatorname{osc}_{\Pi_{n}%
}\left(  g\right)  =0,$ along with the dominated convergence theorem,%
\[
\mathrm{Ad}_{g\left(  \cdot\right)  }\xi1_{B}\left(  \cdot\right)
=L^{2}\left(  m\right)  \text{-}\lim_{n\rightarrow\infty}\sum_{A\in\Pi_{n}%
}\mathrm{Ad}_{g\left(  p_{A}\right)  }\xi1_{A}.
\]
This identity, the $\mathrm{Ad}_{g}$ -invariance of the inner product on
$\mathfrak{k,}$ and the isometry property of the white noise (see Definition
\ref{def.2.1}) then implies,%
\[
\left\langle f,\mathrm{Ad}_{g}\xi1_{B}\right\rangle =L^{2}\left(
\mathbb{P}\right)  \text{-}\lim_{n\rightarrow\infty}\sum_{A\in\Pi_{n}%
}\left\langle f\left(  A\right)  ,\mathrm{Ad}_{g\left(  p_{A}\right)  }%
\xi\right\rangle =L^{2}\left(  \mathbb{P}\right)  \text{-}\lim_{n\rightarrow
\infty}\sum_{A\in\Pi_{n}}\left\langle \mathrm{Ad}_{g\left(  p_{A}\right)
^{-1}}f\left(  A\right)  ,\xi\right\rangle .
\]
As $\xi\in\mathfrak{k}$ is arbitrary, Eq. (\ref{e.4.1}) is proved.
\end{proof}

Recall, as mentioned after Definition \ref{def.3.6}, we will routinely
identify $g\in C\left(  \mathbb{R},K\right)  $ with $g\circ p\in C\left(
\mathbb{R}^{2},K\right)  $ where $p:\mathbb{R}^{2}\rightarrow\mathbb{R}$ is
projection onto the first factor.

\begin{lemma}
\label{lem.4.4}If $g\in C\left(  \mathbb{R},K\right)  \subset C\left(
\mathbb{R}^{2},K\right)  ,$ then $\left(  \hat{f}\right)  ^{g}=\widehat{f^{g}%
}$ where $\hat{f}$ is as in Definition \ref{def.2.4}.
\end{lemma}

\begin{proof}
If $u\in L^{2}\left(  \mathbb{R}^{2};\mathfrak{k}\right)  ,$ then
\[
\left\langle \left(  \hat{f}\right)  ^{g},u\right\rangle =\left\langle \hat
{f},\mathrm{Ad}_{g}u\right\rangle =\left\langle f,\widehat{\mathrm{Ad}_{g}%
u}\right\rangle =\left\langle f,\mathrm{Ad}_{g}\widehat{u}\right\rangle
=\left\langle f^{g},\hat{u}\right\rangle =\left\langle \widehat{f^{g}%
},u\right\rangle .
\]

\end{proof}

\begin{theorem}
[Affine change of variables]\label{thm.4.5}If $\alpha\in L^{2}\left(
\mathbb{R}^{2};\mathfrak{k}\right)  ,$ $g\in C\left(  \mathbb{R}^{2},K\right)
,$ and $\psi\left(  f\right)  $ is a bounded measurable function of the white
noise, $f,$ then
\begin{equation}
\mathbb{E}\left[  \psi\left(  f^{g}-\mathrm{Ad}_{g^{-1}}\alpha\right)
\right]  =\mathbb{E}\left[  \psi\left(  f\right)  e^{-\left\langle
f,\alpha\right\rangle -\frac{1}{2}\left\Vert \alpha\right\Vert ^{2}}\right]  .
\label{e.4.2}%
\end{equation}
In particular the laws of $f^{g}-\mathrm{Ad}_{g^{-1}}\alpha$ and $f$ are
mutually absolutely continuous relative to one another.
\end{theorem}

\begin{proof}
By the multiplicative system theorem (see Dellacherie \cite[p. 14]%
{Dellacherie1972} or Janson \cite[Appendix A., p. 309]{Janson1997}) it
suffices to prove Eq. (\ref{e.4.2}) when $\psi$ is a cylinder functions of the
form,%
\[
\psi\left(  f\right)  =\tilde{\psi}\left(  \left\langle f,u_{1}\right\rangle
,\dots,\left\langle f,u_{k}\right\rangle \right)
\]
with $u_{i}\in L^{2}\left(  \mathbb{R}^{2};\mathfrak{k}\right)  $ or $u_{i}\in
C_{c}^{\infty}\left(  \mathbb{R}^{2};\mathfrak{k}\right)  $ if we prefer. We
may further assume that $\left\{  u_{i}\right\}  _{i=1}^{\infty}$ is an
orthonormal basis for $L^{2}\left(  \mathbb{R}^{2};\mathfrak{k}\right)  $ in
which case $\left\{  \mathrm{Ad}_{g}u_{i}\right\}  _{i=1}^{\infty}$ is also an
orthonormal basis for $L^{2}\left(  \mathbb{R}^{2};\mathfrak{k}\right)  .$
Since%
\begin{align*}
\psi\left(  f^{g}-\mathrm{Ad}_{g^{-1}}\alpha\right)   &  =\tilde{\psi}\left(
\left\langle f-\alpha,\mathrm{Ad}_{g}u_{1}\right\rangle ,\dots,\left\langle
f-\alpha,\mathrm{Ad}_{g}u_{k}\right\rangle \right) \\
&  =\tilde{\psi}\left(  \left\langle f,\mathrm{Ad}_{g}u_{1}\right\rangle
-\left\langle \alpha,\mathrm{Ad}_{g}u_{1}\right\rangle ,\dots,\left\langle
f,\mathrm{Ad}_{g}u_{k}\right\rangle -\left\langle \alpha,\mathrm{Ad}_{g}%
u_{k}\right\rangle \right)  ,
\end{align*}
it follows by a finite dimensional change of variables and the fact that
$\left\{  \left\langle f,\mathrm{Ad}_{g}u_{i}\right\rangle \right\}
_{i=1}^{\infty}$ are i.i.d. standard normal random variables that%
\begin{align}
\mathbb{E}\left[  \psi\left(  f^{g}-\mathrm{Ad}_{g^{-1}}\right)  \right]   &
=\mathbb{E}\left[  \tilde{\psi}\left(  \left\langle f,\mathrm{Ad}_{g}%
u_{1}\right\rangle -\left\langle \alpha,\mathrm{Ad}_{g}u_{1}\right\rangle
,\dots,\left\langle f,\mathrm{Ad}_{g}u_{k}\right\rangle -\left\langle
\alpha,\mathrm{Ad}_{g}u_{k}\right\rangle \right)  \right] \nonumber\\
&  =\mathbb{E}\left[  \tilde{\psi}\left(  \left\langle f,\mathrm{Ad}_{g}%
u_{1}\right\rangle ,\dots,\left\langle f,\mathrm{Ad}_{g}u_{k}\right\rangle
\right)  Z_{\tilde{k}}\right]  \label{e.4.3}%
\end{align}
where for any $\tilde{k}\geq k,$
\[
Z_{\tilde{k}}=\exp\left(  -\sum_{j=1}^{\tilde{k}}\left[  \left\langle
f,\mathrm{Ad}_{g}u_{j}\right\rangle \left\langle \alpha,\mathrm{Ad}_{g}%
u_{j}\right\rangle -\frac{1}{2}\left\vert \left\langle \alpha,\mathrm{Ad}%
_{g}u_{j}\right\rangle \right\vert ^{2}\right]  \right)  .
\]
Using%
\[
L^{\infty-}\text{-}\lim_{\tilde{k}\uparrow\infty}Z_{\tilde{k}}=\exp\left(
-\left\langle f,\alpha\right\rangle -\frac{1}{2}\left\Vert \alpha\right\Vert
^{2}\right)  ,
\]
we may pass to the limit as $\tilde{k}\rightarrow\infty$ in Eq. (\ref{e.4.3})
to arrive at Eq. (\ref{e.4.2}).
\end{proof}

\begin{proposition}
\label{pro.4.6}Suppose that $f$ is a $\mathfrak{k}$ -- valued white noise,
$\left[  a,b\right]  \ni t\rightarrow\ell\left(  t\right)  :=\left(
t,y\left(  t\right)  \right)  $ is a horizontal curve in $\mathbb{R}^{2},$ and
$\left\{  M_{t}^{f}\left(  \ell\right)  \right\}  _{t\in\left[  a,b\right]  }$
is the $\mathfrak{k}$-valued martingale as defined in Definition
\ref{def.2.6}. If $g\in C\left(  \mathbb{R},K\right)  $ such that $g\left(
0\right)  =I,$ then $dM_{t}^{f^{g}}\left(  \ell\right)  =\mathrm{Ad}%
_{g^{-1}\left(  t\right)  }dM_{t}^{f}\left(  \ell\right)  ,$ i.e.
\[
M_{t}^{f^{g}}\left(  \ell\right)  =\int_{a}^{t}\mathrm{Ad}_{g^{-1}\left(
\tau\right)  }dM_{\tau}^{f}\left(  \ell\right)
\]
where the latter integral is a It\^{o} (or essentially Wiener) stochastic integral.
\end{proposition}

\begin{proof}
Let $\Pi=\left\{  a=t_{0}<t_{1}<\dots<t_{n}=t\right\}  $ denote a partition of
$\left[  a,t\right]  .$ By Lemma \ref{lem.4.3},%
\begin{align*}
M_{t}^{\ell}\left(  f^{g}\right)   &  =-\hat{f}^{g}\left(  R_{t}^{\ell
}\right)  =L^{2}\left(  \mathbb{P}\right)  \text{-}\lim_{\left\vert
\Pi\right\vert \rightarrow0}\sum_{j=1}^{n}-\mathrm{Ad}_{g\left(
t_{j-1}\right)  }\hat{f}\left(  R_{t_{j}}^{\ell}\setminus R_{t_{j-1}}^{\ell
}\right) \\
&  =L^{2}\left(  \mathbb{P}\right)  \text{-}\lim_{\left\vert \Pi\right\vert
\rightarrow0}\sum_{j=1}^{n}-\mathrm{Ad}_{g\left(  t_{j-1}\right)  }\left[
\hat{f}\left(  R_{t_{j}}^{\ell}\right)  -\hat{f}\left(  R_{t_{j-1}}^{\ell
}\right)  \right] \\
&  =L^{2}\left(  \mathbb{P}\right)  \text{-}\lim_{\left\vert \Pi\right\vert
\rightarrow0}\sum_{j=1}^{n}\mathrm{Ad}_{g\left(  t_{j-1}\right)  }\left[
M_{t_{j}}^{\ell}\left(  f\right)  -M_{t_{j-1}}^{\ell}\left(  f\right)  \right]
\\
&  =\int_{a}^{t}\mathrm{Ad}_{g\left(  \tau\right)  }dM_{\tau}^{f}\left(
\ell\right)  .
\end{align*}

\end{proof}

\subsection{Perturbations of $f,$ $M^{f},$ and $\pt^{f}$\label{sec.4.2}}

We start by making precise the perturbation, $f_{\eta},$ of $f$ which was
introduced informally in Eq. (\ref{e.3.7}).

\begin{definition}
\label{def.4.7}For $\eta,\eta_{y}:\mathbb{R}^{2}\rightarrow\mathfrak{k}$ as in
Notation \ref{not.3.2}, let
\begin{equation}
f_{\eta}:=f^{g_{\eta}}-\mathrm{Ad}_{g_{\eta}^{-1}}\eta_{y} \label{e.4.4}%
\end{equation}
where $g_{\eta}\in C\left(  \mathbb{R\rightarrow}K\right)  $ is the solution
to the ODE in Eq. (\ref{e.3.4}) in Definition \ref{def.3.6}.
\end{definition}

\begin{theorem}
[Martingale perturbations]\label{thm.4.8}Let $\left[  a,b\right]  \ni
x\rightarrow\ell\left(  x\right)  :=\left(  x,y\left(  x\right)  \right)  $ be
a horizontal curve in $\mathbb{R}^{2},$ $f$ be a $\mathfrak{k}$ -- valued
white noise, $\eta,\eta_{y}:\mathbb{R}^{2}\rightarrow\mathfrak{k}$ be as in
Notation \ref{not.3.2}, and $\left\{  M_{t}^{f}\left(  \ell\right)  \right\}
_{t\in\left[  a,b\right]  }$ be the martingale defined in Definition
\ref{def.2.6}. Then $\left\{  M_{t}^{f_{\eta}}\left(  \ell\right)  \right\}
_{t\in\left[  a,b\right]  }$ is the semi-martingale given by the following
It\^{o} integrals;
\begin{equation}
M_{t}^{f_{\eta}}\left(  \ell\right)  =\int_{a}^{t}\mathrm{Ad}_{g_{\eta}\left(
x\right)  ^{-1}}dM_{x}^{f}\left(  \ell\right)  +\int_{a}^{t}\mathrm{Ad}%
_{g_{\eta}\left(  x\right)  ^{-1}}\bar{\eta}\left(  x,y\left(  x\right)
\right)  dx. \label{e.4.5}%
\end{equation}
Alternatively, the differential form of Eq. (\ref{e.4.5}) is
\begin{equation}
dM_{x}^{f_{\eta}}\left(  \ell\right)  =\mathrm{Ad}_{g_{\eta}\left(  x\right)
^{-1}}\left[  dM_{x}^{f}\left(  \ell\right)  +\bar{\eta}\left(  x,y\left(
x\right)  \right)  dx\right]  . \label{e.4.6}%
\end{equation}
[Recall from Notation \ref{not.3.3} and Remark \ref{rem.3.4} that
\[
\bar{\eta}\left(  x,y\right)  :=\eta\left(  x,y\right)  -\eta\left(
x,0\right)  =\int_{0}^{y}\eta_{y}\left(  x,y^{\prime}\right)  dy^{\prime}.]
\]

\end{theorem}

\begin{proof}
Let us first observe that it makes sense to replace $f$ by $f_{\eta}$ in
$M_{x}^{f}\left(  \ell\right)  $ since (as a consequence of Theorem
\ref{thm.4.5}) the laws of $f$ and $f_{\eta}$ are mutually absolutely
continuous relative to one another. The identity in Eq. (\ref{e.4.5}) is now a
matter of unwinding the definitions;%
\begin{align*}
M_{t}^{f_{\eta}}\left(  \ell\right)   &  =-\hat{f}_{\eta}\left(  R_{t}^{\ell
}\right)  =-\widehat{f^{g_{\eta}}}\left(  R_{t}^{\ell}\right)  +\int
_{R_{t}^{\ell}}\mathrm{Ad}_{g_{\eta}^{-1}\left(  x\right)  }\mathrm{sgn}%
\left(  y\right)  \eta_{y}\left(  x,y\right)  dxdy\\
&  =M_{t}^{f^{g_{\eta}}}\left(  \ell\right)  +\int_{a}^{t}dx\int_{0}^{y\left(
x\right)  }dy\mathrm{Ad}_{g_{\eta}^{-1}\left(  x\right)  }\eta_{y}\left(
x,y\right) \\
&  =M_{t}^{f^{g_{\eta}}}\left(  \ell\right)  +\int_{a}^{t}\mathrm{Ad}%
_{g_{\eta}\left(  x\right)  ^{-1}}\bar{\eta}\left(  x,y\left(  x\right)
\right)  dx.
\end{align*}
The desired result now follows from this equation along with Proposition
\ref{pro.4.6}.
\end{proof}

We now introduce the white noise variant of Definition \ref{def.3.8}.

\begin{definition}
\label{def.4.9}If $\left[  a,b\right]  \ni x\rightarrow\ell\left(  x\right)
=\left(  x,y\left(  x\right)  \right)  $ is a horizontal path, let
$k_{x}\left(  \ell\right)  $ denote the solution to the ODE,%
\begin{equation}
\frac{d}{dx}k_{x}\left(  \ell\right)  +\left[  \mathrm{Ad}_{\pt_{x}^{f}\left(
\ell\right)  ^{-1}}\eta\left(  x,\ell\left(  x\right)  \right)  \right]
k_{x}\left(  \ell\right)  =0\text{ with }k_{0}=I. \label{e.4.7}%
\end{equation}
Further let $\mathbf{k}^{\eta}\left(  \ell\right)  =k_{b}\left(  \ell\right)
.$
\end{definition}

Although suppressed from the notation, the functions, $x\rightarrow
k_{x}\left(  \ell\right)  $ are in general random and depend on the white
noise $f$ through the dependence of Eq. (\ref{e.4.7}) on $\left\{  \pt_{x}%
^{f}\left(  \ell\right)  \right\}  _{x\in\left[  a,b\right]  }.$ The key
result of this section is the following stochastic analogue of Eq.
(\ref{e.3.8}) of Corollary \ref{cor.3.9}.

\begin{theorem}
[Perturbed parallel translation]\label{thm.4.10}If $\eta$ is as above and
$\left[  a,b\right]  \ni x\rightarrow\ell\left(  x\right)  =\left(  x,y\left(
x\right)  \right)  $ is a horizontal path, then%
\begin{equation}
\pt^{f_{\eta}}\left(  \ell\right)  =g_{\eta}\left(  b\right)  ^{-1}%
\pt^{f}\left(  \ell\right)  \mathbf{k}^{\eta}\left(  \ell\right)  g_{\eta
}\left(  a\right)  . \label{e.4.8}%
\end{equation}
[Again, it makes sense to replace $f$ by $f_{\eta}$ in $\pt^{f}\left(
\ell\right)  $ since (as a consequence of Theorem \ref{thm.4.5}) the laws of
$f$ and $f_{\eta}$ are mutually absolutely continuous relative to one another.]
\end{theorem}

\begin{proof}
From Definition \ref{def.3.6} and Eq. (\ref{e.4.6}),
\begin{align*}
\delta\left[  g_{\eta}\left(  x\right)  \pt_{x}^{f_{\eta}}\left(  \ell\right)
\right]   &  =-\eta\left(  x,0\right)  g_{\eta}\left(  x\right)
\pt_{x}^{f_{\eta}}\left(  \ell\right)  dx-g_{\eta}\left(  x\right)  \delta
M_{x}^{f_{\eta}}\left(  \ell\right)  \pt_{x}^{f_{\eta}}\left(  \ell\right) \\
&  =-\eta\left(  x,0\right)  dxg_{\eta}\left(  x\right)  \pt_{x}^{f_{\eta}%
}\left(  \ell\right)  -g_{\eta}\left(  x\right)  \mathrm{Ad}_{g_{\eta}\left(
x\right)  ^{-1}}\left[  \delta M_{x}^{f}\left(  \ell\right)  +\bar{\eta
}\left(  x,y\left(  x\right)  \right)  dx\right]  \pt_{x}^{f_{\eta}}\left(
\ell\right) \\
&  =-\left[  \delta M_{x}^{f}\left(  \ell\right)  +\eta\left(  x,y\left(
x\right)  \right)  dx\right]  g_{\eta}\left(  x\right)  \pt_{x}^{f_{\eta}%
}\left(  \ell\right)  .
\end{align*}
As in the proof of Corollary \ref{cor.3.9}, taking the Stratonovich
differential of the identity, $\pt_{x}^{f}\left(  \ell\right)  ^{-1}%
\pt_{x}^{f_{\eta}}\left(  \ell\right)  =I,$ shows
\[
\delta\pt_{x}^{f}\left(  \ell\right)  ^{-1}=\pt_{x}^{f}\left(  \ell\right)
^{-1}\delta M_{x}^{f}\left(  \ell\right)  .
\]

Combining the previous two equations, it follows that $V_{x}:=\pt_{x}%
^{f}\left(  \ell\right)  ^{-1}g_{\eta}\left(  x\right)  \pt_{x}^{f_{\eta}%
}\left(  \ell\right)  $ satisfies,%
\begin{align*}
\delta V_{x}  &  =\pt_{x}^{f}\left(  \ell\right)  ^{-1}\delta M_{x}^{f}\left(
\ell\right)  g_{\eta}\left(  x\right)  \pt_{x}^{f_{\eta}}\left(  \ell\right)
-\pt_{x}^{f}\left(  \ell\right)  ^{-1}\left[  \delta M_{x}^{f}\left(
\ell\right)  +\eta\left(  x,y\left(  x\right)  \right)  dx\right]  g_{\eta
}\left(  x\right)  \pt_{x}^{f_{\eta}}\left(  \ell\right) \\
&  =-\pt_{x}^{f}\left(  \ell\right)  ^{-1}\eta\left(  x,y\left(  x\right)
\right)  g_{\eta}\left(  x\right)  \pt_{x}^{f_{\eta}}\left(  \ell\right)
dx=-\left[  \mathrm{Ad}_{\pt_{x}^{f}\left(  \ell\right)  ^{-1}}\eta\left(
x,y\left(  x\right)  \right)  \right]  V_{x}dx,
\end{align*}
i.e. $V_{x}$ satisfies the same ODE that $k_{x}g_{\eta}\left(  a\right)  $
satisfies. Therefore by the uniqueness of solutions to ODEs, $V_{x}%
=k_{x}g_{\eta}\left(  a\right)  $ and hence%
\[
\pt_{x}^{f_{\eta}}\left(  \ell\right)  =g_{\eta}\left(  x\right)  ^{-1}%
\pt_{x}^{f}\left(  \ell\right)  V_{x}=g_{\eta}\left(  x\right)  ^{-1}%
\pt_{x}^{f}\left(  \ell\right)  k_{x}g_{\eta}\left(  a\right)  .
\]

\end{proof}

\subsection{Proof of Theorem \ref{thm.2.23}\label{sec.4.3}}

For the proof of Theorem \ref{thm.2.23}, we first need to deduce the required
integration by parts formulas from the results in the previous subsection.

\begin{corollary}
\label{cor.4.11}Let $\mathbf{k}^{\eta}\left(  \sigma\right)  $ be as in
Definition \ref{def.4.9} and $g_{\eta}\in K$ be as in Definition
\ref{def.3.6}. If $U:K^{\mathbb{G}}\rightarrow\mathbb{R}$ is a smooth function
which is discrete gauge invariant under the action determined by $g_{\eta}$
(i.e. $u\left(  v\right)  =g_{\eta}\left(  v\right)  $ for all $v\in V\left(
\mathbb{G}\right)  ),$ then%
\begin{equation}
\mathbb{E}\left[  U\left(  \left\{  \pt^{f}\left(  \sigma\right)
\mathbf{k}^{\eta}\left(  \sigma\right)  \right\}  _{\sigma\in\mathbb{G}%
}\right)  \right]  =\mathbb{E}\left[  U\left(  \pt^{f}\left(  \mathbb{G}%
\right)  \right)  \cdot\exp\left(  -\left\langle f,\eta_{y}\right\rangle
-\frac{1}{2}\left\Vert \eta_{y}\right\Vert ^{2}\right)  \right]  .
\label{e.4.9}%
\end{equation}

\end{corollary}

\begin{proof}
Using the gauge invariance assumption along with Theorem \ref{thm.4.10} we
find,%
\begin{align*}
\mathbb{E}\left[  U\left(  \left\{  \pt^{f}\left(  \sigma\right)
\mathbf{k}^{\eta}\left(  \sigma\right)  \right\}  _{\sigma\in\mathbb{G}%
}\right)  \right]   &  =\mathbb{E}\left[  U\left(  \left\{  g_{\eta}\left(
\sigma_{f}\right)  ^{-1}\pt^{f}\left(  \sigma\right)  \mathbf{k}^{\eta}\left(
\sigma\right)  g_{\eta}\left(  \sigma_{i}\right)  \right\}  _{\sigma
\in\mathbb{G}}\right)  \right] \\
&  =\mathbb{E}\left[  U\left(  \left\{  \pt^{f^{g_{\eta}}-\mathrm{Ad}%
_{g_{\eta}^{-1}}\eta_{y}}\left(  \sigma\right)  \right\}  _{\sigma
\in\mathbb{G}}\right)  \right]  .
\end{align*}
This equation along with Theorem \ref{thm.4.5} then completes the proof.
\end{proof}

The next corollary is a rigorous version of Theorem \ref{mthm.3.14} above.

\begin{corollary}
[Key IBP\ formula]\label{cor.4.12}Continuing the notation and assumptions of
Corollary \ref{cor.4.11} and further letting%
\begin{equation}
\zeta_{\eta}\left(  \sigma\right)  :=\int_{a_{\sigma}}^{b_{\sigma}}%
\mathrm{Ad}_{\pt_{x}^{f}\left(  \sigma\right)  }\eta\left(  x,y\left(
x\right)  \right)  dx\text{ with }\sigma\left(  x\right)  =\left(  x,y\left(
x\right)  \right)  , \label{e.4.10}%
\end{equation}
we have the integration by parts formula,%
\begin{equation}
\mathbb{E}\left[  \sum_{\sigma\in\mathbb{G}}\left(  \nabla_{\zeta_{\eta
}\left(  \sigma\right)  }^{\sigma}U\right)  \left(  \pt^{f}\left(
\mathbb{G}\right)  \right)  \right]  =\mathbb{E}\left[  U\left(
\pt^{f}\left(  \mathbb{G}\right)  \right)  \cdot\left\langle f,\eta
_{y}\right\rangle \right]  . \label{e.4.11}%
\end{equation}

\end{corollary}

\begin{proof}
Let $\mathbf{k}^{s\eta}\left(  \sigma\right)  $ be defined as in Definition
\ref{def.4.9} with $\eta$ replaced by $s\eta$ in which case Eq. (\ref{e.4.9})
reads,%
\[
\mathbb{E}\left[  U\left(  \left\{  \pt^{f}\left(  \sigma\right)
\mathbf{k}^{s\eta}\left(  \sigma\right)  \right\}  _{\sigma\in\mathbb{G}%
}\right)  \right]  =\mathbb{E}\left[  U\left(  \pt^{f}\left(  \mathbb{G}%
\right)  \right)  \cdot\exp\left(  -s\left\langle f,\eta_{y}\right\rangle
-\frac{s^{2}}{2}\left\Vert \eta_{y}\right\Vert ^{2}\right)  \right]
\]
Differentiating this equation with respect to $s$ then gives the integration
by parts formula,%
\begin{equation}
\mathbb{E}\left[  \frac{d}{ds}|_{0}U\left(  \left\{  \pt^{f}\left(
\sigma\right)  \mathbf{k}^{s\eta}\left(  \sigma\right)  \right\}  _{\sigma
\in\mathbb{G}}\right)  \right]  =-\mathbb{E}\left[  U\left(  \pt^{f}\left(
\mathbb{G}\right)  \right)  \cdot\left\langle f,\eta_{y}\right\rangle \right]
. \label{e.4.12}%
\end{equation}
Letting $\kappa_{x}^{\sigma}:=\frac{d}{ds}|_{0}k_{x}^{s\eta}\left(
\sigma\right)  $ we find, by differentiating the ODE,%
\[
\frac{d}{dx}k_{x}^{s\eta}=-s\left(  u_{x}^{-1}\eta\left(  x,y\left(  x\right)
\right)  u_{x}\right)  k_{x}^{s\eta},
\]
for $k_{x}^{s\eta}$ at $s=0$ that
\begin{align*}
\frac{d}{dx}\kappa_{x}^{\sigma}  &  :=\frac{d}{ds}|_{0}\left[  -s\left(
u_{x}^{-1}\eta\left(  x,y\left(  x\right)  \right)  u_{x}\right)  k_{x}%
^{s\eta}\right] \\
&  =-\mathrm{Ad}_{u_{x}^{-1}}\eta\left(  x,y\left(  x\right)  \right)
=-\mathrm{Ad}_{\pt_{x}\left(  \sigma\right)  }\eta\left(  x,y\left(  x\right)
\right)
\end{align*}
and so
\[
\frac{d}{ds}|_{0}\pt^{f}\left(  \sigma\right)  \cdot\mathbf{k}^{s\eta}\left(
\sigma\right)  =-\pt^{f}\left(  \sigma\right)  \cdot\int_{a}^{b}%
\mathrm{Ad}_{\pt_{x}^{f}\left(  \sigma\right)  }\eta\left(  x,y\left(
x\right)  \right)  dx=-\pt^{f}\left(  \sigma\right)  \cdot\zeta_{\eta}\left(
\sigma\right)  .
\]
Therefore%
\[
\frac{d}{ds}|_{0}U\left(  \left\{  \pt^{f}\left(  \sigma\right)
\mathbf{k}^{s\eta}\left(  \sigma\right)  \right\}  _{\sigma\in\mathbb{G}%
}\right)  =-\left(  \tilde{\zeta}_{\eta}U\right)  \left(  \pt^{f}\left(
\mathbb{G}\right)  \right)  =-\sum_{\sigma\in\mathbb{G}}\left(  \nabla
_{\zeta_{\eta}\left(  \sigma\right)  }^{\sigma}U\right)  \left(
\pt^{f}\left(  \mathbb{G}\right)  \right)
\]
which combined with Eq. (\ref{e.4.12}) gives Eq. (\ref{e.4.11}).
\end{proof}

We are now ready to prove Theorem \ref{thm.2.23}.

\begin{proof}
[Proof of Theorem \ref{thm.2.23}]Let $\xi\in\mathfrak{k}:=\operatorname*{Lie}%
\left(  K\right)  $ and $h$ be the function defined in Eq. (\ref{e.3.17}) with
$Q$ sufficiently small. Following the heuristic proof in Subsection
\ref{sec.3.2}, we then find $\zeta_{\eta}\left(  \sigma\right)  =0$ unless
$\sigma=e_{1}$ and for $\sigma=e_{1},$
\[
\zeta_{\eta_{\varepsilon}}\left(  e_{1}\right)  :=\int_{0}^{\infty}h\left(
x,0\right)  \mathrm{Ad}_{\pt_{x}^{f}\left(  e_{1}\right)  }\xi dx=\int
_{0}^{\infty}h\left(  x,0\right)  dx\cdot\xi=\left\vert Q\right\vert \xi.
\]
By Lemma \ref{lem.2.17}, $\nabla_{\xi}^{e_{2}}U$ is still invariant under
discrete gauge transformations, $u:V\left(  \mathbb{G}\right)  \rightarrow K,$
such that $u\left(  0\right)  =I.$ Hence, applying Corollary \ref{cor.4.12}
with $U$ replaced by $\nabla_{\xi}^{e_{2}}U$ shows,%
\begin{align*}
\left\vert Q\right\vert \mathbb{E}\left[  \nabla_{\xi}^{e_{1}}\nabla_{\xi
}^{e_{2}}U\left(  \pt^{f}\left(  \mathbb{G}\right)  \right)  \right]   &
=\mathbb{E}\left[  \tilde{\zeta}_{\eta}\nabla_{\xi}^{e_{2}}U\left(
\pt^{f}\left(  \mathbb{G}\right)  \right)  \right] \\
&  =\mathbb{E}\left[  \left(  \nabla_{\xi}^{e_{2}}U\right)  \left(
\pt^{f}\left(  \mathbb{G}\right)  \right)  \cdot\left\langle f,\left(
\partial_{y}h\right)  \xi\right\rangle \right] \\
&  =\mathbb{E}\left[  \left(  \nabla_{\xi}^{e_{2}}U\right)  \left(
\pt^{f}\left(  \mathbb{G}\right)  \right)  \cdot\left\langle f\left(
RQ\right)  -f\left(  Q\right)  ,\xi\right\rangle _{\mathfrak{k}}\right]
\end{align*}
wherein we have used Eq. (\ref{e.3.16}) for the last equality. Summing this
equation on $\xi\in\beta$ (an orthonormal basis for $\mathfrak{k)}$ then
completes the proof of Eq. (\ref{e.2.8}).
\end{proof}

\section{Loop expansion of parallel translation\label{sec.5}}

As above, $f$ is the $\mathfrak{k}$ -- valued white noise on $\mathbb{R}^{2}.$
Let $u,v:\left[  0,1\right]  \rightarrow\mathbb{R}$ be continuous functions
such that either $0\leq u\left(  t\right)  \leq v\left(  t\right)  $ or $0\geq
u\left(  t\right)  \geq v\left(  t\right)  ,$ $\sigma\left(  t\right)
=\left(  t,u\left(  t\right)  \right)  $ and $\gamma\left(  t\right)  =\left(
t,v\left(  t\right)  \right)  $ be the associated horizontal paths, and
$Q_{t}$ be the region bounded by $y=u\left(  t\right)  ,$ $y=v\left(
t\right)  ,$ $x=0,$ and $x=t,$ see Figure \ref{fig.11}. Further let $M_{t}%
^{f}\left(  \gamma\right)  $ and $M_{t}^{f}\left(  \sigma\right)  $ be the
associated martingales,
\begin{equation}
b_{t}:=M_{t}^{f}\left(  \gamma\right)  -M_{t}^{f}\left(  \sigma\right)
=-\hat{f}\left(  Q_{t}\right)  \label{e.5.1}%
\end{equation}
and
\[
a_{t}^{\gamma}:=\int_{0}^{t}\left\vert v\left(  \tau\right)  \right\vert
d\tau,\text{ }a_{t}^{\sigma}:=\int_{0}^{t}\left\vert u\left(  \tau\right)
\right\vert d\tau,\text{ and }a_{t}:=\int_{0}^{t}\left\vert v\left(
\tau\right)  -u\left(  \tau\right)  \right\vert d\tau=\left\vert
Q_{t}\right\vert .
\]
\begin{figure}[ptbh]
\centering
\par
\psize{2.5in} %
\executeiffilenewer{\GraphicsDirectorysigmagamma_tp.svg}{\GraphicsDirectorysigmagamma_tp.pdf}%
{inkscape -z -D --file=\GraphicsDirectorysigmagamma_tp.svg --export-pdf=\GraphicsDirectorysigmagamma_tp.pdf --export-latex}%
\input{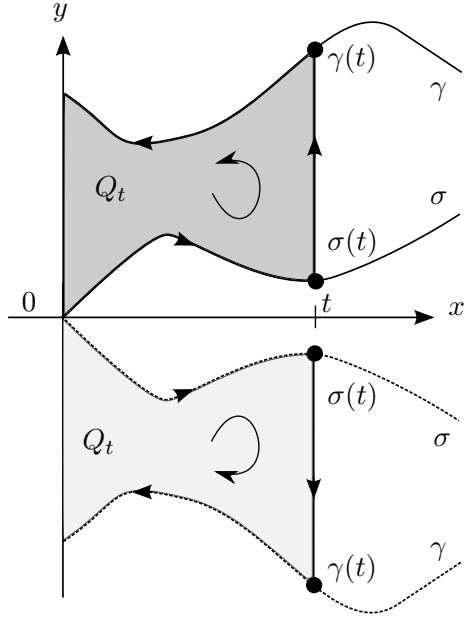}%
\caption{Both scenarios on the
ordering of the functions $u$ and $v$ are depicted in this picture. The reader
should refer to the figures in the upper (lower) half plane when $0\leq u\leq
v$ ($v\leq u\leq0).$}%
\label{fig.11}%
\end{figure}

In this section, we will often make use of the Burkholder-Davis-Gundy
inequalities in the form we now describe. Let $V$ be a finite dimensional
inner product space, $n\in\mathbb{N},$ $L\left(  \mathfrak{k}^{n},V\right)  $
be the linear transformations from $\mathfrak{k}^{n}$ to $V,$ $\left\{
\beta_{t}\right\}  _{t\geq0}$ be a $\mathfrak{k}^{n}$-valued continuous square
integrable martingale with with independent increments, and $\left\{
\left\langle \beta\right\rangle _{t}\right\}  _{t\geq0}$ be the quadratic
variation of $\left\{  \beta_{t}\right\}  _{t\geq0}.$ If $\left\{  u_{t}\in
L\left(  \mathfrak{k}^{n},V\right)  \right\}  _{t\geq0}$ is an adapted
continuous process and $\left\{  \alpha_{t}\right\}  _{t\geq0}$ is an
increasing process dominating $\left\langle \beta\right\rangle _{t}$ (i.e.
there exists $c<\infty$ such that $d\left\langle \beta\right\rangle _{t}\leq
cd\alpha_{t}$ for all $t),$ then there exists a constant $C_{p}<\infty$
depending only on $p,$ $c,n,$ and $V,$ such that
\begin{equation}
\left\Vert \int_{0}^{T}w_{\tau}d\beta_{\tau}\right\Vert _{p}\leq C_{p}%
\sqrt{\int_{0}^{T}\left\Vert w_{\tau}\right\Vert _{p}^{2}d\alpha_{\tau}%
}~\forall~0<T<\infty. \label{e.5.3}%
\end{equation}
To prove this estimate we may assume $\int_{0}^{T}\left\Vert w_{\tau
}\right\Vert _{p}^{2}d\alpha_{\tau}<\infty$ since otherwise Eq. (\ref{e.5.3})
is trivial. Under this assumption, $M_{t}=\int_{0}^{t}w_{\tau}d\beta_{\tau}$
is a square integrable $V$-valued martingale satisfying, $\left\langle
M\right\rangle _{T}\leq c\int_{0}^{T}\left\vert w_{\tau}\right\vert
^{2}d\alpha_{\tau}$ where $\left\vert w_{\tau}\right\vert $ is an appropriate
Hilbert-Schmidt norm of $w_{\tau}.$ If we let $M_{T}^{\ast}=\max_{0\leq t\leq
T}\left\vert M_{t}\right\vert ,$ then the Burkholder-Davis-Gundy inequalities
state that $\left\Vert M_{T}^{\ast}\right\Vert _{p}\asymp\left\Vert
\sqrt{\left\langle M\right\rangle _{T}}\right\Vert _{p},$ i.e. there exist
finite constants, $c_{p}$ and $C_{p}$ such that%
\[
c_{p}\left\Vert \sqrt{\left\langle M\right\rangle _{T}}\right\Vert _{p}%
\leq\left\Vert M_{T}^{\ast}\right\Vert _{p}\leq C_{p}\left\Vert \sqrt
{\left\langle M\right\rangle _{T}}\right\Vert _{p}.
\]
From the second inequality we then have,%
\begin{align*}
\left\Vert \int_{0}^{T}w_{\tau}d\beta_{\tau}\right\Vert _{p}\lesssim\left\Vert
\sqrt{\int_{0}^{T}\left\vert w_{\tau}\right\vert ^{2}d\alpha_{\tau}%
}\right\Vert _{p}  &  =\left\Vert \int_{0}^{T}\left\vert w_{\tau}\right\vert
^{2}d\alpha_{\tau}\right\Vert _{p/2}^{1/2}\\
&  \leq\sqrt{\int_{0}^{T}\left\Vert \left\vert w_{\tau}\right\vert
^{2}\right\Vert _{p/2}d\alpha_{\tau}}=\sqrt{\int_{0}^{T}\left\Vert w_{\tau
}\right\Vert _{p}^{2}d\alpha_{\tau}}.
\end{align*}
which proves Eq. (\ref{e.5.3}).

\begin{example}
\label{ex.5.1}Let $\left\{  b_{t}\right\}  _{t\geq0}$ be as in Eq.
(\ref{e.5.1}), $n=1,$ $\beta_{t}=b_{t},$ and $\alpha_{t}=a_{t}=\left\vert
Q_{t}\right\vert $ so that $d\left\langle b\right\rangle _{t}=\dim
\mathfrak{k}\cdot d\alpha_{t}.$ Taking $V=\mathfrak{k}\otimes\mathfrak{k}$ and
$w_{t}=b_{t}\otimes\left(  \cdot\right)  $ in Eq. (\ref{e.5.3}) then gives,
\[
\left\Vert \int_{0}^{T}b_{t}\otimes db_{t}\right\Vert _{p}\lesssim\sqrt
{\int_{0}^{T}\left\Vert b_{\tau}\right\Vert _{p}^{2}d\alpha_{\tau}}%
\lesssim\sqrt{\int_{0}^{T}\alpha_{\tau}d\alpha_{\tau}}=\sqrt{\frac{1}{2}%
\alpha_{T}^{2}}=\sqrt{\frac{1}{2}}a_{T},
\]
wherein we have used $\left\Vert b_{\tau}\right\Vert _{p}^{2}\asymp\left\Vert
b_{\tau}\right\Vert _{2}^{2}=\dim\mathfrak{k}\cdot\alpha_{\tau}$ because
$b_{\tau}$ is a Gaussian random vector. Similarly one shows $\left\Vert
\int_{0}^{T}db_{t}\otimes b_{t}\right\Vert _{p}\lesssim a_{T}.$
\end{example}

Recall that $K$ is assumed to be a matrix Lie group imbedded $\mathbb{C}%
^{D\times D}$ -- the space of $D\times D$-matrices with complex entries. For a
function, $\psi:\mathbb{C}^{D\times D}\rightarrow\mathbb{C},$ let $D^{j}%
\psi=\psi^{\left(  j\right)  }$ denote the $j^{\text{th}}$ differential of
$\psi$ as function on $\mathbb{C}^{D\times D}$ thought of as a real vector
space. Also let $\Delta_{K}$ be the Laplacian on $K,$ i.e.%
\[
\left(  \Delta_{K}\psi\right)  \left(  k\right)  =\sum_{\xi\in\beta}%
\frac{d^{2}}{dt^{2}}|_{0}\psi\left(  ke^{t\xi}\right)
\]
where $\beta$ is any orthonormal basis for $\mathfrak{k.}$

\begin{lemma}
\label{lem.5.2}If $\psi:\mathbb{C}^{D\times D}\rightarrow\mathbb{C}$ is a
smooth function in a neighborhood of $I\in K\subset\mathbb{C}^{D\times D},$
then
\begin{equation}
\left(  \Delta_{K}\psi\right)  \left(  I\right)  =\sum_{\xi\in\beta}\left[
\psi^{\prime\prime}\left(  I\right)  \xi\otimes\xi+\psi^{\prime}\left(
I\right)  \xi^{2}\right]  =\sum_{\xi\in\beta}\psi^{\prime\prime}\left(
I\right)  \xi\otimes\xi+\psi^{\prime}\left(  I\right)  \kappa, \label{e.5.4}%
\end{equation}
where again $\beta$ is any orthonormal basis for $\mathfrak{k}$ and
$\kappa:=\sum_{\xi\in\beta}\xi^{2}$ is the \textbf{Casmir matrix. }[As usual,
it easily verified that $\Delta_{K}$ and the matrix $\kappa$ are independent
of the choice of orthonormal basis of $\mathfrak{k.]}.$
\end{lemma}

\begin{proof}
For $\xi\in\mathfrak{k}$ we have%
\[
\left(  \tilde{\xi}^{2}\psi\right)  \left(  I\right)  =\frac{d^{2}}{dt^{2}%
}|_{0}\psi\left(  e^{t\xi}\right)  =\frac{d}{dt}|_{0}\left[  \psi^{\prime
}\left(  e^{t\xi}\right)  \xi e^{t\xi}\right]  =\psi^{\prime\prime}\left(
I\right)  \xi\otimes\xi+\psi^{\prime}\left(  I\right)  \xi^{2}.
\]
Summing the above identity on $\xi\in\beta$ gives Eq. (\ref{e.5.4}).
\end{proof}

To simplify notation in the statements and the proofs to follow we will adopt
the following notation.

\begin{notation}
[$O$-notation]\label{not.5.3}If $\left\{  A_{t}\right\}  _{0<t<\delta}$ is a
collection of random variables and $\left(  0,\delta\right)  \ni t\rightarrow
a_{t}\in\left(  0,\infty\right)  $ is a positive function, we write
$A_{t}=O\left(  a_{t}^{3/2}\right)  $ provided; for all $1\leq p<\infty$ there
exists $C=C\left(  p,a,A\right)  <\infty$ such that $\left\Vert A_{t}%
\right\Vert _{p}\leq Ca_{t}$ for $0<t<\delta.$
\end{notation}

\begin{theorem}
[Loop Expansions]\label{thm.5.4}Let $u,$ $v,$ $\sigma,$ $\tau,$ and $Q_{t}$ be
as described at the start of this section, see Figure \ref{fig.11} and let
$\partial Q_{t}$ denote the path traversing the boundary of $Q_{t}$ in the
counter-clockwise (clockwise) direction starting at $0\in\mathbb{R}^{2}$ when
$0\leq u\leq v$ $\left(  v\leq u\leq0\right)  .$ Then
\begin{equation}
\pt^{f}\left(  \partial Q_{t}\right)  =I-\hat{f}\left(  Q_{t}\right)
+\frac{1}{2}\kappa a_{t}+\mathcal{R}_{t}+R_{t} \label{e.5.5}%
\end{equation}
where $\mathcal{R}_{t}$ and $R_{t}$ are matrix valued random variables
satisfying;
\begin{equation}
\mathbb{E}\mathcal{R}_{t}=0,\quad\mathcal{R}_{t}=O_{p}\left(  \sqrt
{a_{t}^{\gamma}a_{t}}\right)  ,\text{ and }R_{t}=O_{p}\left(  a_{t}%
^{3/2}\right)  . \label{e.5.6}%
\end{equation}

\end{theorem}

\begin{proof}
Let $h_{t}=\pt_{t}^{f}\left(  \sigma\right)  $ and $k_{t}=\pt_{t}^{f}\left(
\gamma\right)  $ be parallel translation along $\sigma$ and $\gamma$
respectively so that (in both scenarios)
\[
g_{t}:=\pt^{f}\left(  \partial Q_{t}\right)  =k_{t}^{-1}h_{t}.
\]
Further let (with $b_{t}$ as in Eq. (\ref{e.5.1}))
\begin{equation}
B_{t}:=\int_{0}^{t}\mathrm{Ad}_{h_{\tau}^{-1}}\delta b_{\tau}\text{ and
}\mathbb{B}_{t}:=\int_{0\leq r\leq s\leq t}dB_{r}\otimes dB_{s}=\int_{0}%
^{t}B_{s}\otimes dB_{s} \label{e.5.7}%
\end{equation}
and using,
\[
\left[  d\mathrm{Ad}_{h_{\tau}^{-1}}\right]  db_{\tau}=\mathrm{Ad}_{h_{\tau
}^{-1}}\mathrm{ad}_{dM_{\tau}^{f}\left(  \sigma\right)  }\left[  dM_{\tau}%
^{f}\left(  \gamma\right)  -dM_{\tau}^{f}\left(  \sigma\right)  \right]  =0,
\]
we note that
\begin{equation}
B_{t}:=\int_{0}^{t}\mathrm{Ad}_{h_{\tau}^{-1}}db_{\tau}\text{ and }\left[
dB_{t}\right]  ^{2}=\sum_{\xi\in\beta}\left[  \mathrm{Ad}_{h_{\tau}^{-1}}%
\xi\right]  ^{2}da_{t}=\kappa da_{t}, \label{e.5.8}%
\end{equation}
where $\left[  dB_{t}\right]  ^{2}$ is used to denote the differential of the
quadratic variation matrix of $B_{\left(  \cdot\right)  }.$ As we have
mentioned before, $\delta k_{t}^{-1}=k_{t}^{-1}\delta M^{f}\left(
\gamma\right)  ,$ and thus%
\begin{equation}
\delta g_{t}=k_{t}^{-1}\delta M_{t}^{f}\left(  \gamma\right)  h_{t}-k_{t}%
^{-1}\delta M_{t}^{f}\left(  \sigma\right)  h_{t}=k_{t}^{-1}\delta b_{t}%
h_{t}=g_{t}\mathrm{Ad}_{h_{t}^{-1}}\delta b_{t}=g_{t}\delta B_{t}.
\label{e.5.9}%
\end{equation}
The integral form of Eq. (\ref{e.5.9}) expressed in It\^{o}'s form is now
given by
\begin{align}
g_{t}  &  =I+\int_{0}^{t}g_{\tau}\delta B_{\tau}=I+\int_{0}^{t}g_{\tau
}dB_{\tau}+\frac{1}{2}\int_{0}^{t}g_{\tau}\left[  dB_{\tau}\right]
^{2}\nonumber\\
&  =I+\int_{0}^{t}g_{\tau}dB_{\tau}+\frac{1}{2}\int_{0}^{t}g_{\tau}\kappa
da_{\tau}. \label{e.5.10}%
\end{align}

Making use of Eq. (\ref{e.5.3}) (with $\beta_{t}=B_{t})$ it follows that
\begin{align}
\left\Vert g_{t}-I\right\Vert _{p}  &  \leq\left\Vert \int_{0}^{t}g_{\tau
}dB_{\tau}\right\Vert _{p}+\frac{1}{2}\left\Vert \int_{0}^{t}g_{\tau}\kappa
da_{\tau}\right\Vert _{p}\nonumber\\
&  \lesssim\left[  \sqrt{a_{t}}+a_{t}\right]  =O_{p}\left(  \sqrt{a_{t}%
}\right)  . \label{e.5.11}%
\end{align}

Feeding the expansion for $g_{\left(  \cdot\right)  }$ in Eq. (\ref{e.5.10})
back into the right side of Eq. (\ref{e.5.10}) gives%
\begin{align*}
g_{t}  &  =I+\int_{0}^{t}\left[  I+\int_{0}^{\tau}g_{s}dB_{s}+\frac{1}{2}%
\int_{0}^{\tau}g_{s}\kappa da_{s}\right]  dB_{\tau}\\
&  +\frac{1}{2}\int_{0}^{t}\left[  I+\int_{0}^{\tau}g_{s}dB_{s}+\frac{1}%
{2}\int_{0}^{\tau}g_{s}\kappa da_{s}\right]  \kappa da_{\tau}.
\end{align*}
This identity may be rewritten as
\[
g_{t}=I+B_{t}+\frac{1}{2}\kappa a_{t}+\mathcal{R}_{t}^{\prime}+R_{t}%
\]
where
\begin{align*}
\mathcal{R}_{t}^{\prime}  &  =\int_{0}^{t}\left[  \int_{0}^{\tau}g_{s}%
dB_{s}\right]  dB_{\tau}\text{ and }\\
R_{t}  &  =\frac{1}{2}\int_{0}^{t}\left[  \int_{0}^{\tau}g_{s}\kappa
da_{s}\right]  dB_{\tau}+\int_{0}^{t}\left[  \int_{0}^{\tau}g_{s}dB_{s}%
+\frac{1}{2}\int_{0}^{\tau}g_{s}\kappa da_{s}\right]  \kappa da_{\tau}.
\end{align*}

Using basic estimates along with Eq. (\ref{e.5.3}) one easily shows
\begin{align*}
\left\Vert R_{t}\right\Vert _{p}  &  \lesssim a_{t}^{3/2}+a_{t}^{3/2}%
+a_{t}^{2}=O_{p}\left(  a_{t}^{3/2}\right)  ,\\
\mathbb{E}\mathcal{R}_{t}^{\prime}  &  =0,\text{ and }\left\Vert
\mathcal{R}_{t}^{\prime}\right\Vert _{p}\lesssim\sqrt{\int_{0}^{T}\left\Vert
\int_{0}^{\tau}g_{s}dB_{s}\right\Vert _{p}^{2}da_{\tau}}=\sqrt{\int_{0}%
^{T}a_{\tau}da_{\tau}}=\frac{a_{\tau}}{\sqrt{2}}.
\end{align*}
Similarly
\[
B_{t}=\int_{0}^{t}\mathrm{Ad}_{h_{\tau}^{-1}}db_{\tau}=\int_{0}^{t}\left[
I+\int_{0}^{\tau}\mathrm{Ad}_{h_{\tau}^{-1}}\mathrm{ad}_{\delta M_{\tau}%
^{f}\left(  \sigma\right)  }\right]  db_{\tau}=b_{t}+\mathcal{R}_{t}%
^{\prime\prime}%
\]
where
\[
\mathcal{R}_{t}^{\prime\prime}:=\int_{0}^{t}\int_{0}^{\tau}\mathrm{Ad}%
_{h_{\tau}^{-1}}\mathrm{ad}_{\delta M_{\tau}^{f}\left(  \sigma\right)
}db_{\tau}%
\]
satisfies%
\[
\mathbb{E}\mathcal{R}_{t}^{\prime\prime}=0\text{ and }\mathcal{R}_{t}%
^{\prime\prime}=O\left(  \sqrt{a_{t}^{\sigma}\cdot a_{t}}\right)  .
\]
Thus we have shown
\[
g_{t}=I+b_{t}+\frac{1}{2}\kappa a_{t}+\mathcal{R}_{t}+R_{t}%
\]
where $R_{t}=O\left(  a_{t}^{3/2}\right)  ,$ $\mathcal{R}_{t}=\mathcal{R}%
_{t}^{\prime}+\mathcal{R}_{t}^{\prime\prime},$ and
\[
\left\Vert \mathcal{R}_{t}\right\Vert _{p}\leq\left\Vert \mathcal{R}%
_{t}^{\prime}\right\Vert _{p}+\left\Vert \mathcal{R}_{t}^{\prime\prime
}\right\Vert _{p}=O_{p}\left(  a_{t}\right)  +O\left(  \sqrt{a_{t}^{\sigma
}\cdot a_{t}}\right)  =O\left(  \sqrt{a_{t}^{\gamma}a_{t}}\right)  .
\]

This completes the proof since $b_{t}=-\hat{f}\left(  Q_{t}\right)  .$
\end{proof}

Let $\tilde{K}\subset\mathbb{C}^{D\times D}$ be defined by,%
\[
\tilde{K}:=\left\{  I+s\left(  k-I\right)  :0\leq s\leq1\text{ and }k\in
K\right\}  ,
\]
and note that $K\subset\tilde{K}$ and $\tilde{K}$ is compact in $\mathbb{C}%
^{D\times D}.$

\begin{proposition}
\label{pro.5.5}We continue the setup in Theorem \ref{thm.5.4} and further
suppose that $\psi:\mathbb{C}^{D\times D}\rightarrow\mathbb{C}$ is a random
function taking values in $C^{3}\left(  \mathbb{C}^{D\times D},\mathbb{C}%
\right)  .$ If there exists a (non-random) constant, $C<\infty,$ such that
\[
\left\vert D^{j}\psi\left(  k\right)  \right\vert \leq C\text{ }\forall
~k\in\tilde{K},\text{ and }0\leq j\leq3,
\]
then
\begin{equation}
\psi\left(  g_{t}\right)  =\psi\left(  I\right)  -\left(  \nabla_{\hat
{f}\left(  Q_{t}\right)  }\psi\right)  \left(  I\right)  +a_{t}\frac{1}%
{2}\left(  \Delta_{K}\psi\right)  \left(  I\right)  +\psi^{\prime}\left(
I\right)  \mathcal{R}_{t}+\frac{1}{2}\psi^{\prime\prime}\left(  I\right)
\int_{0}^{t}b_{s}\vee db_{s}+O\left(  a_{t}^{3/2}\right)  , \label{e.5.12}%
\end{equation}
where, for $a,b\in\mathfrak{k\subset\mathbb{C}}^{D\times D}\mathfrak{,}$
$a\vee b:=a\otimes b+b\otimes a$ is the symmetrization of $a\otimes b.$
\end{proposition}

\begin{proof}
Let $g_{t}:=\pt^{f}\left(  \partial Q_{t}\right)  $ and $b_{t}=-\hat{f}\left(
Q_{t}\right)  $ (as above) and define%
\[
\delta_{t}:=g_{t}-I=b_{t}+\frac{1}{2}\kappa a_{t}+\mathcal{R}_{t}+R_{t},
\]
where $\mathcal{R}_{t}$ and $R_{t}$ are as in Theorem \ref{thm.5.4}. Notice
that $I+s\delta_{t}\in\tilde{K}$ for all $t\geq0$ and $0\leq s\leq1.$

By Taylor's theorem with integral remainder we have%
\begin{equation}
\psi\left(  g_{t}\right)  =\psi\left(  I\right)  +\psi^{\prime}\left(
I\right)  \delta_{t}+\int_{0}^{1}\psi^{\prime\prime}\left(  I+s\delta
_{t}\right)  \left[  \delta_{t}\otimes\delta_{t}\right]  \left(  1-s\right)
ds. \label{e.5.13}%
\end{equation}
Since $\psi^{\left(  3\right)  }$ is bounded on $\tilde{K},$
\[
\left\vert \psi^{\prime\prime}\left(  I+s\delta_{t}\right)  -\psi
^{\prime\prime}\left(  I\right)  \right\vert \lesssim\left\vert \delta
_{t}\right\vert ,
\]
which along with the estimate in Eq. (\ref{e.5.11}) and H\"{o}lder's
inequality shows%
\[
\left\Vert \int_{0}^{1}\left(  \psi^{\prime\prime}\left(  I+s\delta
_{t}\right)  -\psi^{\prime\prime}\left(  I\right)  \right)  \left[  \delta
_{t}\otimes\delta_{t}\right]  \left(  1-s\right)  ds\right\Vert _{p}%
\lesssim\left\Vert \left\vert \delta_{t}\right\vert ^{3}\right\Vert
_{p}=O\left(  a_{t}^{3/2}\right)  .
\]
Combining these estimates with Eq. (\ref{e.5.13}) implies%
\begin{equation}
\psi\left(  g_{t}\right)  =\psi\left(  I\right)  +\psi^{\prime}\left(
I\right)  \delta_{t}+\frac{1}{2}\psi^{\prime\prime}\left(  I\right)  \left[
\delta_{t}\otimes\delta_{t}\right]  +O\left(  a_{t}^{3/2}\right)  ,
\label{e.5.14}%
\end{equation}
Using $R_{t}=O\left(  a_{t}^{3/2}\right)  $ and $\delta_{t}\otimes\delta
_{t}-b_{t}\otimes b_{t}=O\left(  a_{t}^{3/2}\right)  $ in Eq. (\ref{e.5.14})
then shows%
\begin{align*}
\psi\left(  g_{t}\right)   &  =\psi\left(  I\right)  +\psi^{\prime}\left(
I\right)  \left[  b_{t}+\frac{1}{2}\kappa a_{t}+\mathcal{R}_{t}\right]
+\frac{1}{2}\psi^{\prime\prime}\left(  I\right)  \left[  b_{t}\otimes
b_{t}\right]  +O\left(  a_{t}^{3/2}\right) \\
&  =\psi\left(  I\right)  -\left(  \nabla_{\hat{f}\left(  Q_{t}\right)  }%
\psi\right)  +\frac{1}{2}\left[  a_{t}\psi^{\prime}\left(  I\right)
\kappa+\psi^{\prime\prime}\left(  I\right)  \left[  b_{t}\otimes b_{t}\right]
\right]  +\psi^{\prime}\left(  I\right)  \mathcal{R}_{t}+O\left(  a_{t}%
^{3/2}\right)  .
\end{align*}

By It\^{o}'s formula,%
\[
b_{t}\otimes b_{t}=\int_{0}^{t}b_{s}\vee db_{s}+a_{t}\cdot\sum_{\xi\in\beta
}\xi\otimes\xi,
\]
and so by Lemma \ref{lem.5.2},
\[
a_{t}\psi^{\prime}\left(  I\right)  \kappa+\psi^{\prime\prime}\left(
I\right)  \left[  b_{t}\otimes b_{t}\right]  =a_{t}\left(  \Delta_{K}%
\psi\right)  \left(  I\right)  +\psi^{\prime\prime}\left(  I\right)  \int
_{0}^{t}b_{s}\vee db_{s}.
\]
Combining these identities and estimates gives Eq. (\ref{e.5.12}).
\end{proof}

\begin{lemma}
\label{lem.5.6}If $u$ is a $C^{2}$-function defined in a neighborhood of $I\in
K$ and $\tilde{\psi}\left(  k\right)  :=\psi\left(  k^{-1}\right)  ,$ then
$\nabla\tilde{\psi}\left(  I\right)  =-\nabla\psi\left(  I\right)  $ and
$\left(  \Delta_{K}\tilde{\psi}\right)  \left(  I\right)  =\left(  \Delta
_{K}\psi\right)  \left(  I\right)  .$
\end{lemma}

\begin{proof}
The elementary proof is left to the reader.
\end{proof}

\begin{notation}
\label{not.5.7}Let $\mathcal{R}_{\varepsilon}^{\pm}=\mathcal{R}_{\varepsilon}$
and $b_{t}^{\pm}=b_{t}$ be as in Theorem \ref{thm.5.4} and Eq. (\ref{e.5.1})
respectively when $Q_{t}$ is in the upper/lower half plane. [In what follows
the lower half plane region will be $RQ_{t}$ where $Q_{t}$ is the region in
the upper half plane.]
\end{notation}

\begin{corollary}
\label{cor.5.8}Let $\mathbb{G},$ $\varepsilon>0,$ $\mathbb{G}_{\varepsilon
,\pm},$ and $U:K^{\mathbb{G}}\rightarrow\mathbb{C}$ be as in Theorem
\ref{thm.2.27} and further assume that $U$ has been extended (arbitrarily) to
a smooth function on $\left[  \mathbb{C}^{D\times D}\right]  ^{\mathbb{G}}.$
Then
\begin{equation}
U\left(  \pt^{f}\left(  \mathbb{G}_{+,\varepsilon}\right)  \right)  -U\left(
\pt^{f}\left(  \mathbb{G}_{-,\varepsilon}\right)  \right)  =-\left(
\nabla_{_{f\left(  Q_{\varepsilon}\right)  }}^{e_{2}}U+\nabla_{_{f\left(
RQ_{\varepsilon}\right)  }}^{e_{4}}U\right)  \left(  \pt^{f}\left(
\mathbb{G}\right)  \right)  +E_{\varepsilon}^{+}-E_{\varepsilon}^{-}+O\left(
a_{\varepsilon}^{3/2}\right)  \label{e.5.15}%
\end{equation}
where%
\begin{align}
E_{\varepsilon}^{+}  &  :=\left(  D^{e_{2}}U\right)  \left(  \pt^{f}\left(
\mathbb{G}\right)  \right)  \mathcal{R}_{\varepsilon}^{+}+\frac{1}{2}\left(
D^{e_{2}}D^{e_{2}}U\right)  \left(  \pt^{f}\left(  \mathbb{G}\right)  \right)
\int_{0}^{\varepsilon}b_{s}^{+}\vee db_{s}^{+}\text{ and}\label{e.5.16}\\
E_{\varepsilon}^{-}  &  :=\left(  D^{e_{4}}U\right)  \left(  \pt^{f}\left(
\mathbb{G}\right)  \right)  \mathcal{R}_{\varepsilon}^{-}+\frac{1}{2}\left(
D^{e_{4}}D^{e_{4}}U\right)  \left(  \pt^{f}\left(  \mathbb{G}\right)  \right)
\int_{0}^{\varepsilon}b_{s}^{-}\vee db_{s}^{-}. \label{e.5.17}%
\end{align}

\end{corollary}

\begin{proof}
Let $\mathbb{G}^{\prime}:=\mathbb{G}\setminus\left\{  e_{2},e_{3}\right\}  $
and let us write $U\left(  \omega\right)  $ as $U\left(  \omega\left(
e_{2}\right)  ,\omega\left(  e_{4}\right)  ,\omega\left(  \mathbb{G}^{\prime
}\right)  \right)  .$ We then have,
\[
U\left(  \pt^{f}\left(  \mathbb{G}_{+,\varepsilon}\right)  \right)  =U\left(
\pt^{f}\left(  \partial Q_{\varepsilon}\right)  ,I,\pt^{f}\left(
\mathbb{G}^{\prime}\right)  \right)
\]
and by Eq. (\ref{e.5.12}) with $\psi\left(  k\right)  :=U\left(
k,I,\pt^{f}\left(  \mathbb{G}^{\prime}\right)  \right)  $ it follows that
\begin{equation}
U\left(  \pt^{f}\left(  \mathbb{G}_{+,\varepsilon}\right)  \right)  =U\left(
\pt^{f}\left(  \mathbb{G}\right)  \right)  -\left(  \nabla_{_{\hat{f}\left(
Q_{\varepsilon}\right)  }}^{e_{2}}U\right)  \left(  \pt^{f}\left(
\mathbb{G}\right)  \right)  +a_{\varepsilon}\frac{1}{2}\left(  \Delta
_{K}^{e_{2}}U\right)  \left(  \pt^{f}\left(  \mathbb{G}\right)  \right)
+E_{\varepsilon}^{+}+O\left(  a_{\varepsilon}^{3/2}\right)  . \label{e.5.18}%
\end{equation}
Similarly%
\[
U\left(  \pt^{f}\left(  \mathbb{G}_{-,\varepsilon}\right)  \right)  =U\left(
I,\pt^{f}\left(  \partial RQ_{\varepsilon}\right)  ,\pt^{f}\left(
\mathbb{G}^{\prime}\right)  \right)
\]
and by Eq. (\ref{e.5.12}) with $\psi\left(  k\right)  :=U\left(
I,k,\pt^{f}\left(  \mathbb{G}^{\prime}\right)  \right)  $ it follows that
\begin{equation}
U\left(  \pt^{f}\left(  \mathbb{G}_{-,\varepsilon}\right)  \right)  =U\left(
\pt^{f}\left(  \mathbb{G}\right)  \right)  -\left(  \nabla_{_{\hat{f}\left(
RQ_{\varepsilon}\right)  }}^{e_{4}}U\right)  \left(  \pt^{f}\left(
\mathbb{G}\right)  \right)  +a_{\varepsilon}\frac{1}{2}\left(  \Delta
_{K}^{e_{4}}U\right)  \left(  \pt^{f}\left(  \mathbb{G}\right)  \right)
+E_{\varepsilon}^{-}+O\left(  a_{\varepsilon}^{3/2}\right)  . \label{e.5.19}%
\end{equation}
Since $U$ has extended gauge invariance at $0,$ Lemma \ref{lem.5.6} implies
$\left(  \Delta_{K}^{e_{4}}U\right)  \left(  \pt^{f}\left(  \mathbb{G}\right)
\right)  =\left(  \Delta_{K}^{e_{2}}U\right)  \left(  \pt^{f}\left(
\mathbb{G}\right)  \right)  .$ Using the previous identities and the fact that
$\hat{f}\left(  RQ_{\varepsilon}\right)  =-f\left(  RQ_{\varepsilon}\right)
,$ we may subtract Eqs. (\ref{e.5.19}) from Eq. (\ref{e.5.18}) to arrive at
Eq. (\ref{e.5.15}).
\end{proof}

In order to complete the proof of Theorem \ref{thm.2.27} it remains to
estimate the error terms, $E_{\varepsilon}^{\pm},$ in Eqs. (\ref{e.5.16}) and
(\ref{e.5.17}) which we will do with the aid of Lemma \ref{lem.3.19}. Lemma
\ref{lem.5.10} below contains the key estimate needed to make this scheme work.

\begin{notation}
\label{not.5.9}For $t>0,$ let $J_{t}:=\left[  0,t\right]  \times\mathbb{R}$
and
\begin{equation}
\mathcal{B}_{t}:=\sigma\left\{  f\left(  R\right)  :R\in\mathcal{B}%
_{\mathbb{R}^{2}}^{o}\text{ with }R\subset J_{t}\right\}  \label{e.5.20}%
\end{equation}
be the $\sigma$-algebra generated by the white noise over $J_{t}.$
\end{notation}

In the following lemma, let $\hat{\nabla}^{\sigma}$ be the \textquotedblleft
right\textquotedblright\ analogue of $\nabla^{\sigma},$ that is replace
$\omega\left(  b\right)  e^{t\delta_{\sigma,b}\xi}$ by $e^{t\delta_{\sigma
,b}\xi}\omega\left(  b\right)  $ in the definition of $\nabla^{\sigma}$ in
Definition \ref{def.2.16}.

\begin{lemma}
\label{lem.5.10}Let $\varepsilon>0,$ $\Lambda$ be a finite collection of
horizontal curves over $\left[  0,\varepsilon\right]  ,$ and $V:K^{\Lambda
}\rightarrow\mathbb{C}$ is a random function independent of $\mathcal{B}%
_{\varepsilon}$ such that $V$ takes values in $C^{2}\left(  K^{\Lambda
},\mathbb{C}\right)  .$ If there exists a (non-random) constant, $C<\infty,$
such that%
\begin{equation}
\sup_{\left\vert \xi\right\vert _{\mathfrak{k}}=1~\&~\left\vert \eta
\right\vert _{\mathfrak{k}}=1}\left[  \left\vert V\right\vert +\sum_{\sigma
\in\Lambda}\left\vert \hat{\nabla}_{\xi}^{\sigma}V\right\vert +\sum
_{\sigma,\tau\in\Lambda}\left\vert \hat{\nabla}_{\xi}^{\sigma}\hat{\nabla
}_{\eta}^{\tau}V\right\vert \right]  \leq C\text{ on }K^{\Lambda},
\label{e.5.21}%
\end{equation}
then
\begin{equation}
\left\Vert V\left(  \pt^{f}\left(  \Lambda\right)  \right)  -V\left(
\mathbf{I}\right)  \right\Vert _{2}=O\left(  \sqrt{\varepsilon}\right)
\label{e.5.22}%
\end{equation}
where $\mathbf{I}$ is the identity in $K^{\Lambda}$ and $\pt^{f}\left(
\Lambda\right)  :=\left\{  \pt^{f}\left(  \sigma\right)  :\sigma\in
\Lambda\right\}  \in K^{\Lambda}$ as in Notation \ref{not.2.9}.
\end{lemma}

\begin{proof}
For $0\leq t\leq\varepsilon,$ let $G_{t}:=\pt_{t}^{f}\left(  \Lambda\right)
\in K^{\Lambda}$ and $M_{t}$ be the $L\left(  \mathfrak{k}^{\Lambda
},\mathfrak{k}^{\Lambda}\right)  $-valued martingale which is block diagonal
having $M_{t}^{f}\left(  \sigma\right)  $ in the $\sigma$-$\sigma$ block for
each $\sigma\in\Lambda.$ With this notation $G_{t}$ solves the Stratonovich
differential equation,%
\[
\delta G_{t}=-\delta M_{t}G_{t}\text{ with }G_{0}=\mathbf{I}.
\]
Although $V$ is a random function, because it is independent of $\mathcal{B}%
_{\varepsilon},$ we may still use the adapted stochastic calculus to find,%
\begin{align}
V\left(  \pt_{\varepsilon}^{f}\left(  \Lambda\right)  \right)   &  -V\left(
\mathbf{I}\right)  =V\left(  G_{\varepsilon}\right)  -V\left(  G_{0}\right)
=-\sum_{\sigma\in\Lambda}\int_{0}^{\varepsilon}\left(  \hat{\nabla}_{\delta
M_{t}^{f}\left(  \sigma\right)  }^{\sigma}V\right)  \left(  G_{t}\right)
\nonumber\\
=  &  -\sum_{\sigma\in\Lambda}\int_{0}^{\varepsilon}\left(  \hat{\nabla
}^{\sigma}V\right)  \left(  G_{t}\right)  dM_{t}^{f}\left(  \sigma\right)
+\frac{1}{2}\sum_{\sigma,\tau\in\Lambda}\int_{0}^{\varepsilon}\left(
\hat{\nabla}^{\tau}\hat{\nabla}^{\sigma}V\right)  \left(  G_{t}\right)
\left[  dM_{t}^{f}\left(  \tau\right)  \otimes dM_{t}^{f}\left(
\sigma\right)  \right]  . \label{e.5.23}%
\end{align}
Since $\Lambda$ is a finite set, the quadratic covariances of $M_{t}%
^{f}\left(  \tau\right)  $ and $M_{t}^{f}\left(  \sigma\right)  $ for all
$\sigma,\tau\in\Lambda$ are controlled by $dt$ and therefore,%
\[
\left\vert \sum_{\sigma,\tau\in\Lambda}\int_{0}^{\varepsilon}\left(
\hat{\nabla}^{\tau}\hat{\nabla}^{\sigma}V\right)  \left(  G_{t}\right)
\left[  dM_{t}^{f}\left(  \tau\right)  \otimes dM_{t}^{f}\left(
\sigma\right)  \right]  \right\vert \leq C_{1}\varepsilon
\]
and (from the estimate in Eq. (\ref{e.5.3}))
\[
\left\Vert \sum_{\sigma\in\Lambda}\int_{0}^{\varepsilon}\left(  \hat{\nabla
}^{\sigma}V\right)  \left(  G_{t}\right)  dM_{t}^{f}\left(  \sigma\right)
\right\Vert _{2}\leq C_{2}\sqrt{\varepsilon}.
\]
These estimates along with Eq. (\ref{e.5.23}) prove Eq. (\ref{e.5.22}).
\end{proof}

\subsection{Proof of Theorem \ref{thm.2.27}\label{sec.5.1}}

We start by recording the following elementary covariance estimate.

\begin{lemma}
\label{lem.3.19}If $Z,$ $\bar{Z},$ and $N$ are square integrable random
variables such that $\mathbb{E}N=0$ and $\bar{Z}$ is independent of $N,$ then
\[
\left\vert \mathbb{E}\left[  NZ\right]  \right\vert =\left\vert \mathbb{E}%
\left[  N\left(  Z-\bar{Z}\right)  \right]  \right\vert \leq\left\Vert
N\right\Vert _{2}\cdot\left\Vert Z-\bar{Z}\right\Vert _{2}.
\]

\end{lemma}

As in the statement of Theorem \ref{thm.2.27} we will assume there is a fixed
constant, $C,$ independent of $\varepsilon>0$ such that $\left\vert
S_{\varepsilon}\right\vert \leq C\left\vert Q_{\varepsilon}\right\vert $ for
all $\varepsilon>0$ where $S_{\varepsilon}$ is the \textquotedblleft
shadow\textquotedblright\ region as in Figure \ref{fig.8}. This means in
practice that $a_{t}^{\gamma}\leq Ca_{t}$ and so the error, $\mathcal{R}%
_{t}=O_{p}\left(  \sqrt{a_{t}^{\gamma}a_{t}}\right)  ,$ in Eq. (\ref{e.5.6})
may be rewritten as $\mathcal{R}_{t}=O_{p}\left(  a_{t}\right)  .$ We are now
in a position to give a rigorous proof of Theorem \ref{thm.2.27}.

\begin{proof}
[Proof of Theorem \ref{thm.2.27}]Using $a_{\varepsilon}\leq c\varepsilon$ for
some $c>0,$ according to Corollary \ref{cor.5.8} to finish the proof of
Theorem \ref{thm.2.27} it suffices to show $\left\vert \mathbb{E}\left[
E_{\varepsilon}^{+}-E_{\varepsilon}^{-}\right]  \right\vert =O\left(
\sqrt{\varepsilon}a_{\varepsilon}\right)  .$ The error term, $E_{\varepsilon
}^{+}-E_{\varepsilon}^{-},$ is a sum of four terms all of the form,
$\mathcal{U}\left(  \pt^{f}\left(  \mathbb{G}\right)  \right)  W_{\varepsilon
},$ where $\mathcal{U}\in C^{2}\left(  K^{\mathbb{G}},Y^{\ast}\right)  $ with
$Y=\mathfrak{k}$ or $\mathfrak{k}\otimes\mathfrak{k}$ and $W_{\varepsilon}$ is
a $\mathcal{B}_{\varepsilon}$-measurable $Y$-valued random vector such that;
$\mathbb{E}W_{\varepsilon}=0$ and $\left\Vert W_{\varepsilon}\right\Vert
_{2}=O\left(  a_{\varepsilon}\right)  .$ [Here we use $Y^{\ast}$ to denote the
real linear functionals on $Y.$] So to finish the proof it suffices to show
that the expectation of any such expression, $\mathcal{U}\left(
\pt^{f}\left(  \mathbb{G}\right)  \right)  W_{\varepsilon},$ is $O\left(
\sqrt{\varepsilon}a_{\varepsilon}\right)  .$ Before going into the details,
let us give a sketch of the proof.

Let $f_{\varepsilon}:=1_{\mathbb{R}^{2}\setminus J_{\varepsilon}}f$ where $f$
is the white noise. Then $\mathcal{U}\left(  \pt^{f_{\varepsilon}}\left(
\mathbb{G}\right)  \right)  $ now depends only on the white noise over
$\mathbb{R}^{2}\setminus J_{\varepsilon}$ and is therefore independent of
$\mathcal{B}_{\varepsilon}.$ So by Lemma \ref{lem.3.19} with $N=W_{\varepsilon
},$ $Z=\mathcal{U}\left(  \pt^{f}\left(  \mathbb{G}\right)  \right)  ,$ and
$\bar{Z}=\mathcal{U}\left(  \pt^{f_{\varepsilon}}\left(  \mathbb{G}\right)
\right)  ,$ it follows that
\begin{align*}
\left\vert \mathbb{E}\left[  \mathcal{U}\left(  \pt^{f}\left(  \mathbb{G}%
\right)  \right)  W_{\varepsilon}\right]  \right\vert  &  \leq\left\Vert
\mathcal{U}\left(  \pt^{f}\left(  \mathbb{G}\right)  \right)  -\mathcal{U}%
\left(  \pt^{f_{\varepsilon}}\left(  \mathbb{G}\right)  \right)  \right\Vert
_{2}\left\Vert W_{\varepsilon}\right\Vert _{2}\\
&  \leq\left\Vert \mathcal{U}\left(  \pt^{f}\left(  \mathbb{G}\right)
\right)  -\mathcal{U}\left(  \pt^{f_{\varepsilon}}\left(  \mathbb{G}\right)
\right)  \right\Vert _{2}O\left(  a_{\varepsilon}\right)  .
\end{align*}
The proof will be completed by showing, with the aid of Lemma \ref{lem.5.10},
that
\begin{equation}
\left\Vert \mathcal{U}\left(  \pt^{f}\left(  \mathbb{G}\right)  \right)
-\mathcal{U}\left(  \pt^{f_{\varepsilon}}\left(  \mathbb{G}\right)  \right)
\right\Vert _{2}\preceq\sqrt{\varepsilon}. \label{e.5.24}%
\end{equation}
We now proceed to the details.

By subdividing the paths in $\mathbb{G}$ and changing the arbitrary
orientations if necessary, we may assume that all paths in $\mathbb{G}$ are
either purely vertical paths or are horizontal paths oriented from left to
right of the form $\left[  a_{\sigma},b_{\sigma}\right]  \ni x\rightarrow
\sigma\left(  x\right)  =\left(  x,y\left(  x\right)  \right)  \in
\mathbb{R}^{2}$ with $y$ being a continuous function of $x.$ Let
$\mathbb{G}_{h}$ denote the horizontal paths in $\mathbb{G}$ and recall for
$\sigma\in\mathbb{G}_{h}$ that $\tilde{\sigma}=\sigma\left(  \left[
a_{\sigma},b_{\sigma}\right]  \right)  $ is the image of $\sigma$ in
$\mathbb{R}^{2}.$ We now define $\Lambda\subset\mathbb{G}_{h}$ to be the those
paths in $\mathbb{G}_{h}$ which \textquotedblleft cross\textquotedblright\ the
$y$-axis, i.e.
\[
\Lambda=\left\{  \sigma\in\mathbb{G}_{h}:\tilde{\sigma}\cap J_{\varepsilon
}\neq\emptyset\text{ for all }\varepsilon>0\right\}  .
\]
As $\Lambda$ is a finite set there exists $\varepsilon_{0}>0$ such that
$\tilde{\sigma}\cap J_{\varepsilon_{0}}\neq\emptyset$ for all $\sigma
\in\Lambda$ and $\tilde{\sigma}\cap J_{\varepsilon_{0}}=\emptyset$ for all
$\sigma\notin\mathbb{G}_{h}\setminus\Lambda.$ We now assume that
$0<\varepsilon<\varepsilon_{0}$ for the rest of the argument.

For those $\sigma\in\Lambda$ we split $\sigma$ into three paths,
$\sigma|_{\left[  a_{\sigma},0\right]  },$ $\sigma|_{\left[  0,\varepsilon
\right]  },$ and $\sigma|_{\left[  \varepsilon,b_{\sigma}\right]  }$ and note
that
\[
\pt^{f}\left(  \sigma\right)  =\pt^{f}\left(  \sigma|_{\left[  \varepsilon
,b_{\sigma}\right]  }\right)  \pt^{f}\left(  \sigma|_{\left[  0,\varepsilon
\right]  }\right)  \pt^{f}\left(  \sigma|_{\left[  a_{\sigma},0\right]
}\right)
\]
while
\[
\pt^{f_{\varepsilon}}\left(  \sigma\right)  =\pt^{f}\left(  \sigma|_{\left[
\varepsilon,b_{\sigma}\right]  }\right)  \pt^{f}\left(  \sigma|_{\left[
a_{\sigma},0\right]  }\right)  .
\]
Define the random function, $V:K^{\Lambda}\rightarrow Y^{\ast},$ by
\[
V\left(  \omega\right)  :=\mathcal{U}\left(  \left\{  \pt^{f}\left(
\sigma\right)  \right\}  _{\sigma\notin\Lambda},\left\{  \pt^{f}\left(
\sigma|_{\left[  \varepsilon,b_{\sigma}\right]  }\right)  \omega\left(
\sigma\right)  \pt^{f}\left(  \sigma|_{\left[  a_{\sigma},0\right]  }\right)
\right\}  _{\sigma\in\Lambda}\right)  \text{ }\forall~\omega\in K^{\Lambda}.
\]
As $\mathcal{U}$ is $C^{2},$ $V$ is also $C^{2}$ and furthermore $V$ depends
only on the white noise over $\mathbb{R}^{2}\setminus J_{\varepsilon}$ and
hence is independent of $\mathcal{B}_{\varepsilon}$. For $\sigma,\tau
\in\Lambda$ and $\xi,\eta\in\mathfrak{k}$ we have, with
\[
\mathbf{G:=}\left(  \left\{  \pt^{f}\left(  \sigma\right)  \right\}
_{\sigma\notin\Lambda},\left\{  \pt^{f}\left(  \sigma|_{\left[  \varepsilon
,b_{\sigma}\right]  }\right)  \omega\left(  \sigma\right)  \pt^{f}\left(
\sigma|_{\left[  a_{\sigma},0\right]  }\right)  \right\}  _{\sigma\in\Lambda
}\right)  ,
\]
that%
\[
\left(  \hat{\nabla}_{\xi}^{\sigma}V\right)  \left(  \omega\right)  =\left(
\hat{\nabla}_{\mathrm{Ad}_{\pt^{f}\left(  \sigma|_{\left[  \varepsilon
,b_{\sigma}\right]  }\right)  }\xi}^{\sigma}\mathcal{U}\right)  \left(
\mathbf{G}\right)
\]
and similarly%
\[
\left(  \hat{\nabla}_{\eta}^{\tau}\hat{\nabla}_{\xi}^{\sigma}V\right)  \left(
\omega\right)  =\left(  \hat{\nabla}_{\mathrm{Ad}_{\pt^{f}\left(
\sigma|_{\left[  \varepsilon,b_{\sigma}\right]  }\right)  }\eta}^{\tau}%
\hat{\nabla}_{\mathrm{Ad}_{\pt^{f}\left(  \sigma|_{\left[  \varepsilon
,b_{\sigma}\right]  }\right)  }\xi}^{\sigma}\mathcal{U}\right)  \left(
\mathbf{G}\right)  .
\]
Since the inner product on $\mathfrak{k}$ is $\mathrm{Ad}_{K}$-invariant and
$\mathcal{U}$ is a $C^{2}$-function on a compact set, it follows that $V$
satisfies the estimates in Eq. (\ref{e.5.21}). Applying Lemma \ref{lem.5.10}
then shows%
\[
\left\Vert \mathcal{U}\left(  \pt^{f}\left(  \mathbb{G}\right)  \right)
-\mathcal{U}\left(  \pt^{f_{\varepsilon}}\left(  \mathbb{G}\right)  \right)
\right\Vert _{2}=\left\Vert V\left(  \pt^{f}\left(  \Lambda\right)  \right)
-V\left(  \mathbf{I}\right)  \right\Vert _{2}=O\left(  \sqrt{\varepsilon
}\right)  .
\]

\end{proof}

\appendix

\section{Appendix: connections, parallel translation, and
curvature\label{sec.A}}

In this first appendix, we review a few basic facts about covariant
derivatives, parallel translation, and curvature. Recall that we have assumed
that our compact Lie group is a matrix Lie sub-group of $GL\left(
\mathbb{C}^{D}\right)  \subset\mathbb{C}^{D\times D}$ for some $D\in
\mathbb{N}.$

\subsection{Transformation properties\label{sec.A.1}}

The next result recalls how $g\in\mathcal{G}$ acts on covariant
differentiation, parallel translation, and curvature.

\begin{theorem}
[Gauge transformed quantities]\label{thm.A.1}If $A\in\mathcal{A},$
$g\in\mathcal{G},$ $\ell:\left[  a,b\right]  \rightarrow M$ is an absolutely
continuous path in $M,$ and $S:\left[  a,b\right]  \rightarrow\mathbb{C}^{D}$
(or $S:\left[  a,b\right]  \rightarrow\mathbb{C}^{D\times D})$ be a $C^{1}%
$-function, then

\begin{enumerate}
\item The operator $\nabla_{t}^{A^{g}}$ is conjugate to $\nabla_{t}^{A}.$ More
precisely,
\begin{equation}
\nabla_{t}^{A^{g}}S\left(  t\right)  =g\left(  \ell\left(  t\right)  \right)
^{-1}\nabla_{t}^{A}\left[  g\left(  \ell\left(  t\right)  \right)  S\left(
t\right)  \right]  \label{e.A.1}%
\end{equation}
so that $\nabla_{t}^{A^{g}}=M_{g\left(  \ell\left(  t\right)  \right)  ^{-1}%
}\nabla_{t}^{A}M_{g\left(  \ell\left(  t\right)  \right)  }$ where $M_{g}$ is
used to denote multiplication by $g.$

\item For $t\in\left[  a,b\right]  ,$%
\begin{equation}
\pt_{t}^{A^{g}}\left(  \ell\right)  =g\left(  \ell\left(  t\right)  \right)
^{-1}\pt_{t}^{A}\left(  \ell\right)  g\left(  \ell\left(  a\right)  \right)  .
\label{e.A.2}%
\end{equation}

\item The curvature tensor, $F^{A},$ satisfies,%
\begin{equation}
F^{A^{g}}\left\langle v,w\right\rangle =\mathrm{Ad}_{g\left(  x\right)  ^{-1}%
}F^{A}\left\langle v,w\right\rangle \text{ for all }v,w\in T_{x}M\text{ and
}x\in M. \label{e.A.3}%
\end{equation}

\end{enumerate}
\end{theorem}

\begin{proof}
Equation \ref{e.A.1} follows by direct computation using the product rule and
basic calculus. The proof of Eq. (\ref{e.A.2}) is now elementary as
$u_{t}:=g\left(  \ell\left(  t\right)  \right)  ^{-1}\pt_{t}^{A}\left(
\ell\right)  g\left(  \ell\left(  a\right)  \right)  $ satisfies,%
\[
\nabla_{t}^{A^{g}}u_{t}=g\left(  \ell\left(  t\right)  \right)  ^{-1}%
\nabla_{t}^{A}\left[  g\left(  \ell\left(  t\right)  \right)  g\left(
\ell\left(  t\right)  \right)  ^{-1}\pt_{t}^{A}\left(  \ell\right)  g\left(
\ell\left(  a\right)  \right)  \right]  =0\text{ with }u_{a}=I.
\]
Although the curvature assertion in Eq. (\ref{e.A.3}) may be proved by direct
calculation, let us give a more conceptual proof which makes use of the fact
that curvature is related to the commutator of two covariant derivatives. More
precisely, let $\Sigma\left(  s,t\right)  \in M$ and $S\left(  s,t\right)
\in\mathbb{C}^{D}$ (or $\mathbb{C}^{D\times D})$ be two $C^{1}$-functions of
$\left(  s,t\right)  \in\mathbb{R}^{2}$ and let
\[
\nabla_{t}^{A}:=\frac{d}{dt}+A\left(  \dot{\Sigma}\left(  t,s\right)  \right)
\text{ and }\nabla_{s}^{A}:=\frac{d}{ds}+A\left(  \Sigma^{\prime}\left(
t,s\right)  \right)  .
\]
A straightforward computation, using $\left[  \frac{d}{dt},\frac{d}%
{ds}\right]  =0$ and Cartan's formula,
\[
\frac{d}{dt}A\left(  \Sigma^{\prime}\left(  t,s\right)  \right)  -\frac{d}%
{ds}A\left(  \dot{\Sigma}\left(  t,s\right)  \right)  =dA\left(  \dot{\Sigma
}\left(  t,s\right)  ,\Sigma^{\prime}\left(  t,s\right)  \right)  ,
\]
shows%
\begin{equation}
\left[  \nabla_{t}^{A},\nabla_{s}^{A}\right]  S\left(  s,t\right)
=F^{A}\left(  \dot{\Sigma}\left(  t,s\right)  ,\Sigma^{\prime}\left(
t,s\right)  \right)  S\left(  s,t\right)  . \label{e.A.4}%
\end{equation}
Thus it follows that
\begin{align*}
F^{A^{g}}\left(  \dot{\Sigma}\left(  t,s\right)  ,\Sigma^{\prime}\left(
t,s\right)  \right)  S\left(  s,t\right)   &  =\left[  \nabla_{t}^{A^{g}%
},\nabla_{s}^{A^{g}}\right]  S\left(  s,t\right) \\
&  =\left[  M_{g\left(  \Sigma\left(  s,t\right)  \right)  ^{-1}}\nabla
_{t}^{A}M_{g\left(  \Sigma\left(  s,t\right)  \right)  },M_{g\left(
\Sigma\left(  s,t\right)  \right)  ^{-1}}\nabla_{s}^{A}M_{g\left(
\Sigma\left(  s,t\right)  \right)  }\right]  S\left(  s,t\right) \\
&  =M_{g\left(  \Sigma\left(  s,t\right)  \right)  ^{-1}}\left[  \nabla
_{t}^{A},\nabla_{s}^{A}\right]  M_{g\left(  \Sigma\left(  s,t\right)  \right)
}S\left(  s,t\right) \\
&  =M_{g\left(  \Sigma\left(  s,t\right)  \right)  ^{-1}}F^{A}\left(
\dot{\Sigma}\left(  t,s\right)  ,\Sigma^{\prime}\left(  t,s\right)  \right)
M_{g\left(  \Sigma\left(  s,t\right)  \right)  }S\left(  s,t\right)
\end{align*}
from which Eq. (\ref{e.A.3}) is easily deduced.
\end{proof}

\begin{remark}
\label{rem.A.2}One more formula connecting covariant differentiation to
parallel translation is the identity;%
\begin{equation}
\nabla_{t}^{A}S=\pt_{t}^{A}\left(  \ell\right)  \frac{d}{dt}\left[
\pt_{t}^{A}\left(  \ell\right)  ^{-1}S\left(  t\right)  \right]  .
\label{e.A.5}%
\end{equation}
To prove this let $V\left(  t\right)  :=\pt_{t}^{A}\left(  \ell\right)
^{-1}S\left(  t\right)  $ so that $S\left(  t\right)  =\pt_{t}^{A}\left(
\ell\right)  V\left(  t\right)  .$ Now apply the product rule and use
$\nabla_{t}^{A}\pt_{t}^{A}\left(  \ell\right)  =0$ to find,%
\begin{align*}
\nabla_{t}^{A}S\left(  t\right)   &  =\nabla_{t}^{A}\left[  \pt_{t}^{A}\left(
\ell\right)  V\left(  t\right)  \right]  =\left(  \frac{d}{dt}+A\left(
\dot{\ell}\left(  t\right)  \right)  \right)  \left[  \pt_{t}^{A}\left(
\ell\right)  V\left(  t\right)  \right] \\
&  =\left[  \left(  \frac{d}{dt}+A\left(  \dot{\ell}\left(  t\right)  \right)
\right)  \pt_{t}^{A}\left(  \ell\right)  \right]  V\left(  t\right)
+\pt_{t}^{A}\left(  \ell\right)  \dot{V}\left(  t\right)  =\pt_{t}^{A}\left(
\ell\right)  \dot{V}\left(  t\right)
\end{align*}
which is Eq. (\ref{e.A.5}).
\end{remark}

\begin{proposition}
[Connections and diffeomorphisms]\label{pro.A.3}Let $A\in\mathcal{A},$
$\sigma\in C^{1}\left(  \left[  a,b\right]  ,M\right)  ,$ and $\varphi
:M\rightarrow M$ be a diffeomorphism of $M.$ Then $F^{\varphi^{\ast}A}%
=\varphi^{\ast}F^{A}$ and $\pt_{t}^{\varphi^{\ast}A}\left(  \sigma\right)
=\pt_{t}^{A}\left(  \varphi\circ\sigma\right)  $ for all $a\leq t\leq b.$
\end{proposition}

\begin{proof}
Using the basic properties of pull backs on forms we have,%
\[
\varphi^{\ast}F^{A}=\varphi^{\ast}\left[  dA+A\wedge A\right]  =d\left[
\varphi^{\ast}A\right]  +\varphi^{\ast}A\wedge\varphi^{\ast}A=F^{\varphi
^{\ast}A}.
\]
For the second assertion we compute,%
\begin{align*}
0  &  =\frac{\nabla}{dt}\pt_{t}^{A}\left(  \varphi\circ\sigma\right)
:=\left[  \frac{d}{dt}+A\left\langle \frac{d}{dt}\left(  \varphi\circ
\sigma\left(  t\right)  \right)  \right\rangle \right]  \pt_{t}^{A}\left(
\varphi\circ\sigma\right) \\
&  =\left[  \frac{d}{dt}+A\left\langle \varphi_{\ast}\dot{\sigma}\left(
t\right)  \right\rangle \right]  \pt_{t}^{A}\left(  \varphi\circ\sigma\right)
=\left[  \frac{d}{dt}+\left(  \varphi^{\ast}A\right)  \left\langle \dot
{\sigma}\left(  t\right)  \right\rangle \right]  \pt_{t}^{A}\left(
\varphi\circ\sigma\right)
\end{align*}
from which we see that $\pt_{t}^{A}\left(  \varphi\circ\sigma\right)  $
satisfies the same differential equation as $\pt_{t}^{\varphi^{\ast}A}\left(
\sigma\right)  .$
\end{proof}

\begin{corollary}
\label{cor.A.4}If $\dim M=2,$ $A\in\mathcal{A},$ and $\varphi:M\rightarrow M$
is an area preserving diffeomorphism, then $\left\Vert F^{\varphi^{\ast}%
A}\right\Vert ^{2}=\left\Vert F^{A}\right\Vert ^{2}.$
\end{corollary}

\begin{proof}
By definition,%
\begin{equation}
\left\Vert F^{\varphi^{\ast}A}\right\Vert ^{2}=\int_{M}\left\vert
F^{\varphi^{\ast}A}\right\vert ^{2}d\operatorname*{Vol}\nolimits_{g}=\int
_{M}\left\vert \varphi^{\ast}F^{A}\right\vert ^{2}d\operatorname*{Vol}%
\nolimits_{g}. \label{e.A.6}%
\end{equation}
Since $d=2,$ if we let $\omega$ denote the (local) Riemannian volume form on
$M$ then $F^{A}=f\cdot\omega$ for some $f:M\rightarrow\mathfrak{k.}$ The
assumption that $\varphi$ is area preserving means $\varphi^{\ast}\omega
=\pm\omega$ and therefore,%
\[
\varphi^{\ast}F^{A}=f\circ\varphi\cdot\varphi^{\ast}\omega=\pm f\circ
\varphi\cdot\omega.
\]
As $\omega\left(  e_{1},e_{2}\right)  =\pm1$ where $\left\{  e_{1}%
,e_{2}\right\}  $ is any local orthonormal frame on $M,$ we find%
\[
\left\vert \varphi^{\ast}F^{A}\right\vert ^{2}=\left\vert f\circ
\varphi\right\vert ^{2}=\left\vert F^{A}\right\vert ^{2}\circ\varphi.
\]
Using this result in Eq. (\ref{e.A.6}) gives,%
\[
\left\Vert F^{\varphi^{\ast}A}\right\Vert ^{2}=\int_{M}\left\vert
F^{A}\right\vert ^{2}\circ\varphi~d\operatorname*{Vol}\nolimits_{g}=\int
_{M}\left\vert F^{A}\right\vert ^{2}~d\operatorname*{Vol}\nolimits_{g}%
=\left\Vert F^{A}\right\Vert ^{2},
\]
where in the second equality we have the area preserving assumption again,
namely that $\varphi_{\ast}\operatorname*{Vol}_{g}=\operatorname*{Vol}_{g}.$
\end{proof}

\subsection{Differential properties of parallel translation\label{sec.A.2}}

\begin{proposition}
[Connection Comparison]\label{pro.A.5}Suppose that $A$ and $B$ are two
connection $1$-forms, $\ell\in C^{1}\left(  \left[  0,1\right]  ,M\right)  ,$
and $k_{t}:=\pt_{t}^{B}\left(  \ell\right)  \pt_{t}^{A}\left(  \ell\right)
^{-1}$, then
\begin{equation}
\dot{k}_{t}+\left(  \mathrm{Ad}_{\pt_{t}^{A}\left(  \ell\right)  ^{-1}}\left[
B\left(  \dot{\ell}\left(  t\right)  \right)  -A\left(  \dot{\ell}\left(
t\right)  \right)  \right]  \right)  k_{t}=0\text{ with }k_{0}=0.
\label{e.A.7}%
\end{equation}

\end{proposition}

\begin{proof}
Since $\pt_{t}^{B}\left(  \ell\right)  =\pt_{t}^{A}\left(  \ell\right)
k_{t},$ $\nabla_{t}^{B}\pt_{t}^{B}\left(  \ell\right)  =0=\nabla_{t}%
^{A}\pt_{t}^{A}\left(  \ell\right)  $ it follows that%
\begin{align*}
0  &  =\nabla_{t}^{B}\pt_{t}^{B}\left(  \ell\right)  =\left(  \nabla_{t}%
^{B}\pt_{t}^{A}\left(  \ell\right)  \right)  k_{t}+\pt_{t}^{A}\left(
\ell\right)  \dot{k}_{t}\\
&  =\left(  \left[  \nabla_{t}^{B}-\nabla_{t}^{A}\right]  \pt_{t}^{A}\left(
\ell\right)  \right)  k_{t}+\pt_{t}^{A}\left(  \ell\right)  \dot{k}_{t}\\
&  =\left[  B\left(  \dot{\ell}\left(  t\right)  \right)  -A\left(  \dot{\ell
}\left(  t\right)  \right)  \right]  \pt_{t}^{A}\left(  \ell\right)
k_{t}+\pt_{t}^{A}\left(  \ell\right)  \dot{k}_{t}%
\end{align*}
from which Eq. (\ref{e.A.7}) follows.
\end{proof}

\begin{proposition}
[Connection Differentiation]\label{pro.A.6}If $\eta$ is a $\mathfrak{k}$
-valued one form on $M$ and $\ell\in C^{1}\left(  \left[  0,1\right]
,M\right)  ,$ then%
\begin{align}
\partial_{\eta}\left[  A\rightarrow\pt^{A}\left(  \ell\right)  \right]   &
=-\pt^{A}\left(  \ell\right)  \int_{0}^{1}\mathrm{Ad}_{\pt_{t}^{A}\left(
\ell\right)  ^{-1}}\eta\left(  \dot{\ell}\left(  t\right)  \right)
dt\label{e.A.8}\\
&  =-\left[  \int_{0}^{1}\mathrm{Ad}_{\pt_{t}^{A}\left(  \ell\right)  }%
\eta\left(  \dot{\ell}\left(  t\right)  \right)  dt\right]  \pt^{A}\left(
\ell\right)  \label{e.A.9}%
\end{align}

\end{proposition}

\begin{proof}
\textbf{First proof. }Differentiating the identity, $0=\nabla_{t}^{A+s\eta
}\pt_{t}^{A+s\eta}\left(  \ell\right)  $ with respect to $s$ gives,
\begin{align*}
0  &  =\frac{d}{ds}|_{0}\left[  \nabla_{t}^{A+s\eta}\pt_{t}^{A+s\eta}\left(
\ell\right)  \right] \\
&  =\left[  \frac{d}{ds}|_{0}\nabla_{t}^{A+s\eta}\right]  \pt_{t}^{A}\left(
\ell\right)  +\nabla_{t}^{A}\left[  \frac{d}{ds}|_{0}\pt_{t}^{A+s\eta}\left(
\ell\right)  \right] \\
&  =\eta\left(  \dot{\ell}\left(  t\right)  \right)  \pt_{t}^{A}\left(
\ell\right)  +\nabla_{t}^{A}\partial_{\eta}\pt_{t}^{A}\left(  \ell\right) \\
&  =\eta\left(  \dot{\ell}\left(  t\right)  \right)  \pt_{t}^{A}\left(
\ell\right)  +\pt_{t}^{A}\left(  \ell\right)  \frac{d}{dt}\left[  \pt_{t}%
^{A}\left(  \ell\right)  ^{-1}\partial_{\eta}\pt_{t}{}^{A}\left(  \ell\right)
\right]
\end{align*}
wherein we have used Eq. (\ref{e.A.5}) in the last equality. Multiplying this
equation by $\pt_{t}^{A}\left(  \ell\right)  ^{-1}$ and then integrating the
result easily gives Eq. (\ref{e.A.8}) which is equivalent to Eq. (\ref{e.A.9}).

\textbf{Second proof.} Letting $B=A+s\eta$ in Proposition \ref{pro.A.5} shows
$\pt_{t}^{A+s\eta}\left(  \ell\right)  =\pt_{t}^{A}\left(  \ell\right)
k_{t}^{s}$ where
\[
\dot{k}_{t}^{s}+s\left(  \mathrm{Ad}_{\pt_{t}^{A}\left(  \ell\right)  ^{-1}%
}\eta\left(  \dot{\ell}\left(  t\right)  \right)  \right)  k_{t}^{s}=0\text{
with }k_{0}^{s}=0.
\]
Differentiating this equation with respect to $s$ at $s=0$ while using
$k_{t}^{0}=I$ shows
\[
\frac{d}{ds}|_{0}\dot{k}_{t}^{s}+\mathrm{Ad}_{\pt_{t}^{A}\left(  \ell\right)
^{-1}}\eta\left(  \dot{\ell}\left(  t\right)  \right)  =0
\]
and then integrating this result relative to $t$ shows
\[
\frac{d}{ds}|_{0}k_{t}^{s}=-\int_{0}^{t}\mathrm{Ad}_{\pt_{\tau}^{A}\left(
\ell\right)  ^{-1}}\eta\left(  \dot{\ell}\left(  \tau\right)  \right)  d\tau.
\]
Hence it follows that
\begin{align*}
\partial_{\eta}\left[  A\rightarrow\pt_{t}^{A}\left(  \ell\right)  \right]
&  =\frac{d}{ds}|_{0}\pt_{t}^{A+s\eta}\left(  \ell\right)  =\frac{d}{ds}%
|_{0}\pt_{t}^{A}\left(  \ell\right)  k_{t}^{s}\\
&  =-\pt_{t}^{A}\left(  \ell\right)  \int_{0}^{t}\mathrm{Ad}_{\pt_{\tau}%
^{A}\left(  \ell\right)  ^{-1}}\eta\left(  \dot{\ell}\left(  \tau\right)
\right)  d\tau.
\end{align*}

\end{proof}

\begin{proposition}
[Path Differentiation]\label{pro.A.7}Suppose $\ell_{s}$ is a one parameter
family of curves parametrized by an interval, $\left[  0,1\right]  $ such that
$\ell_{s}\left(  0\right)  $ is constant independent of $s$ and let
$\pt_{t}^{A}\left(  \ell\right)  $ denote parallel translation along $\ell
_{s}|_{\left[  0,t\right]  }.$ Then,%
\begin{equation}
\frac{\nabla^{A}}{ds}\pt_{t}^{A}\left(  \ell_{s}\right)  =\pt_{t}^{A}\left(
\ell_{s}\right)  \int_{0}^{t}\mathrm{Ad}_{\pt_{\tau}^{A}\left(  \ell
_{s}\right)  ^{-1}}F^{A}\left(  \dot{\ell}_{s}\left(  \tau\right)  ,\ell
_{s}^{\prime}\left(  \tau\right)  \right)  d\tau\label{e.A.10}%
\end{equation}
and, if we further assume that $\ell_{s}\left(  1\right)  $ is constant
independent of $s$, then%
\begin{equation}
\frac{d}{ds}\pt_{1}^{A}\left(  \ell_{s}\right)  =\pt_{1}^{A}\left(  \ell
_{s}\right)  \int_{0}^{1}\mathrm{Ad}_{\pt_{\tau}^{A}\left(  \ell_{s}\right)
^{-1}}F^{A}\left(  \dot{\ell}_{s}\left(  \tau\right)  ,\ell_{s}^{\prime
}\left(  \tau\right)  \right)  d\tau. \label{e.A.11}%
\end{equation}
Equation (\ref{e.A.10}) may also be expressed as,%
\begin{equation}
\frac{d}{ds}\left[  \pt_{s}^{A}\left(  \ell_{\left(  \cdot\right)  }\left(
t\right)  \right)  ^{-1}\pt_{t}^{A}\left(  \ell_{s}\right)  \right]  =\left[
\pt_{s}^{A}\left(  \ell_{\left(  \cdot\right)  }\left(  t\right)  \right)
^{-1}\pt_{t}^{A}\left(  \ell_{s}\right)  \right]  \int_{0}^{t}\mathrm{Ad}%
_{\pt_{\tau}^{A}\left(  \ell_{s}\right)  ^{-1}}F^{A}\left(  \dot{\ell}%
_{s}\left(  \tau\right)  ,\ell_{s}^{\prime}\left(  \tau\right)  \right)
d\tau. \label{e.A.12}%
\end{equation}

\end{proposition}

\begin{proof}
By Eq. (\ref{e.A.4}) and the fact that $\frac{\nabla^{A}}{dt}\pt_{t}%
^{A}\left(  \ell_{s}\right)  =0,$%
\[
\frac{\nabla^{A}}{dt}\frac{\nabla^{A}}{ds}\pt_{t}^{A}\left(  \ell_{s}\right)
=\left[  \frac{\nabla^{A}}{dt},\frac{\nabla^{A}}{ds}\right]  \pt_{t}%
^{A}\left(  \ell_{s}\right)  =F^{A}\left(  \dot{\ell}_{s}\left(  t\right)
,\ell_{s}^{\prime}\left(  t\right)  \right)  \pt_{t}^{A}\left(  \ell
_{s}\right)  .
\]
By Remark \ref{rem.A.2}, the last identity may be rewritten as,
\[
\frac{d}{dt}\left[  \pt_{t}^{A}\left(  \ell\right)  ^{-1}\frac{\nabla^{A}}%
{ds}\pt_{t}^{A}\left(  \ell_{s}\right)  \right]  =\pt_{t}^{A}\left(
\ell\right)  ^{-1}F^{A}\left(  \dot{\ell}\left(  t\right)  ,\ell_{s}^{\prime
}\left(  t\right)  \right)  \pt_{t}^{A}\left(  \ell\right)  .
\]
Integrating this equation on $t$ gives Eq. (\ref{e.A.10}). If we now assume
$\ell_{s}\left(  1\right)  $ is constant in $s,$ then
\[
\frac{\nabla^{A}}{ds}\pt_{1}^{A}\left(  \ell_{s}\right)  =\left(  \frac{d}%
{ds}+A\left\langle \ell_{s}^{\prime}\left(  1\right)  \right\rangle \right)
\pt_{1}^{A}\left(  \ell_{s}\right)  =\frac{d}{ds}\pt_{1}^{A}\left(  \ell
_{s}\right)
\]
which combined with Eq. (\ref{e.A.10}) at $t=1$ gives Eq. (\ref{e.A.11}).
\end{proof}

For more information on Proposition \ref{pro.A.7} much more related material
to this and the next appendix, see \cite{Gross1985} and \cite{Driver89a}.

\section{Homotopy gauge fixing of Yang-Mills\label{sec.B}}

The goal of this appendix is to motivate the definition of the Yang-Mills
measure as used in this paper. We also wish to give a heuristic argument that
the resulting expectations should be invariant under area preserving
diffeomorphisms. We begin with a few general results in finite dimensions
which we will later apply (illegally) in the infinite dimensional Yang-Mills
context. For an interesting general discussion of gauge fixing from a
differential form point, as apposed to the more measure theoretic view
described in this appendix, see \cite[Section III]{Nguyen2016b}. There is of
course a huge physics literature on various methods of gauge fixing which we
do not attempt to survey here. However, the interested reader might start with
Chapter 13 in \cite[Section 13.6]{DHoker2004} or Chapter 15 in
\cite{Weinberg2005a} and then consult some of the references in
\cite{Nguyen2016b}.

\subsection{Group actions and gauges\label{sec.B.1}}

We will use the following notation throughout this subsection.

\begin{notation}
\label{not.B.1}Let $\left(  \mathcal{A},\mathcal{G},m,\lambda\right)  $ be a
quadruple consisting of a smooth manifold, $\mathcal{A},$ a Lie group
$\mathcal{G},$ a smooth measure $\left(  m\right)  $ on $\mathcal{A},$ and a
right invariant Haar measure $\left(  \lambda\right)  $ on $\mathcal{G}.$ We
assume that there is a given right action of $\mathcal{G}$ on $\mathcal{A}$
and that the measure $m$ is invariant under this right action, i.e. $m$ is
invariant under the transformation, $\mathcal{A\ni}A\rightarrow Ag\in
\mathcal{A}$ for each $g\in\mathcal{G}.$
\end{notation}

\begin{definition}
\label{def.B.2}A \textbf{gauge} is a smooth function, $v:\mathcal{A}%
\rightarrow\mathcal{G},$ such that $v\left(  Ag\right)  =v\left(  A\right)  g$
for all $A\in\mathcal{A}$ and $g\in\mathcal{G}.$ Associated to $v$ we defined
the \textquotedblleft projection map,\textquotedblright\ $\pi_{v}%
:\mathcal{A}\rightarrow\mathcal{A},$ by $\pi_{v}\left(  A\right)  :=A\cdot
v\left(  A\right)  ^{-1}$ and let
\[
\mathcal{A}_{v}:=\pi_{v}\left(  \mathcal{A}\right)  =\left\{  A\cdot v\left(
A\right)  ^{-1}:A\in\mathcal{A}\right\}  .
\]

\end{definition}

\begin{lemma}
\label{lem.B.3}If $v:\mathcal{A}\rightarrow\mathcal{G}$ is a gauge and
$A,B\in\mathcal{A}$, then;

\begin{enumerate}
\item $\pi_{v}$ is constant on gauge orbits,

\item $\pi_{v}\circ\pi_{v}=\pi_{v}$ (i.e. $\pi_{v}|_{\mathcal{A}_{v}}$ is the
identity on $\mathcal{A}_{v}),$

\item $\mathcal{A}_{v}$ may also be expresses as%
\[
\mathcal{A}_{v}=\left\{  A\in\mathcal{A}:v\left(  A\right)  =I\in
\mathcal{G}\right\}  ,
\]

\item $\mathcal{A}_{v}$ is an embedded submanifold of $\mathcal{A},$

\item $\pi_{v}\left(  A\right)  =\pi_{v}\left(  B\right)  $ iff $A$ and $B$
are in the same $\mathcal{G}$-orbit, and

\item the map,
\begin{equation}
\mathcal{A}_{v}\times\mathcal{G}\ni\left(  A,g\right)  \rightarrow A\cdot
g\in\mathcal{A} \label{e.B.1}%
\end{equation}
is as diffeomorphism of smooth manifolds.
\end{enumerate}
\end{lemma}

\begin{proof}
We take each item in turn.

\begin{enumerate}
\item If $A\in\mathcal{A}$ and $g\in\mathcal{G},$ then%
\[
\pi_{v}\left(  Ag\right)  =Ag\cdot v\left(  Ag\right)  ^{-1}=Ag\cdot\left[
v\left(  A\right)  g\right]  ^{-1}=A\cdot v\left(  A\right)  ^{-1}=\pi
_{v}\left(  A\right)
\]
which shows $\pi_{\nu}$ is constant on $\mathcal{G}$-orbits.

\item If $A\in\mathcal{A},$ then $A$ and $\pi_{v}\left(  A\right)  $ are in
the same gauge orbit and hence $\pi_{v}\left(  \pi_{v}\left(  A\right)
\right)  =\pi_{v}\left(  A\right)  .$

\item If $A\in\mathcal{A}_{v}$ then $A=\pi_{\nu}\left(  A\right)  =A\cdot
v\left(  A\right)  ^{-1}$ and therefore,%
\[
v\left(  A\right)  =v\left(  A\cdot v\left(  A\right)  ^{-1}\right)  =v\left(
A\right)  \cdot v\left(  A\right)  ^{-1}=I.
\]
Conversely if $\nu\left(  A\right)  =I,$ then $\pi_{v}\left(  A\right)
=A\in\mathcal{A}_{v}.$

\item If $A\in\mathcal{A\ }$and $\xi\in\operatorname*{Lie}\left(
\mathcal{G}\right)  =T_{I}\mathcal{G},$ then
\[
\frac{d}{dt}|_{0}v\left(  Ae^{t\xi}\right)  =\frac{d}{dt}|_{0}\left[  v\left(
A\right)  e^{t\xi}\right]  =L_{v\left(  A\right)  \ast}\xi
\]
where the latter expression varies over $T_{v\left(  A\right)  }G$ as $\xi$
varies over $\operatorname*{Lie}\left(  \mathcal{G}\right)  .$ This shows
$\nu$ is a submersion and so the level sets of $v$ are all embedded
submanifolds, in particular$\mathcal{A}_{v}=v^{-1}\left(  \left\{  I\right\}
\right)  $ is an embedded submanifold.

\item The condition that $\pi_{v}\left(  A\right)  =\pi_{v}\left(  B\right)  $
is equivalent to $A\cdot v\left(  A\right)  ^{-1}=B\cdot v\left(  B\right)
^{-1}$ which is then equivalent to $B=A\cdot\left[  v\left(  A\right)
^{-1}v\left(  B\right)  \right]  ,$ i.e. $B$ and $A$ are in the same gauge orbit.

\item The inverse to the smooth map in Eq. (\ref{e.B.1}) is the smooth map,
$\mathcal{A}\ni A\rightarrow\left(  \pi_{v}\left(  A\right)  ,v\left(
A\right)  \right)  .$
\end{enumerate}
\end{proof}

\begin{example}
[Product groups]\label{ex.B.4}Let $\mathcal{G}$ be a Lie group, $N\in
\mathbb{N},$ $\mathcal{A}=\mathcal{G}^{N},$ and let $\mathcal{G}$ act on
$\mathcal{A}$ on the right by the diagonal action,
\[
\mathcal{A\times G}\ni\left(  \overrightarrow{g},k\right)  \rightarrow
\overrightarrow{g}\cdot k\in\mathcal{A}\text{ where }\left[  \overrightarrow
{g}\cdot k\right]  _{i}=g_{i}k\text{ for }1\leq i\leq N.
\]
Then $\nu:\mathcal{A}\rightarrow\mathcal{G}$ defined by $v\left(
\overrightarrow{g}\right)  =g_{1}$ is a gauge. In this case
\[
\pi_{v}\left(  \overrightarrow{g}\right)  =\overrightarrow{g}\cdot g_{1}%
^{-1}=\left(
\begin{array}
[c]{c}%
e\\
g_{2}g_{1}^{-1}\\
\vdots\\
g_{N}g_{1}^{-1}%
\end{array}
\right)
\]
and $\mathcal{A}_{v}=\left\{  e\right\}  \times\mathcal{G}^{N-1}.$
\end{example}

\begin{example}
\label{ex.B.5}Let $\mathcal{A}=\mathbb{R}^{n},$ $\mathcal{G}=\mathbb{R},$
$\xi\in sl\left(  n,\mathbb{R}\right)  $ such that $\xi_{lk}=0$ if either $l$
or $k=n$ and for $x\in\mathbb{R}^{n}$ (thought of as row vector) and
$t\in\mathbb{R}$ let
\[
x\cdot t:=xe^{t\xi}+te_{n}=\left[  x+te_{n}\right]  e^{t\xi}.
\]
Since $e_{n}\xi=0$ we have $e_{n}e^{t\xi}=e_{n}$ and by assumption $e^{t\xi}$
preserves $\operatorname*{span}\left(  e_{k}\right)  _{k<n}$ and hence%
\[
\left(  x\cdot t\right)  \cdot s=\left(  xe^{t\xi}+te_{n}\right)  e^{s\xi
}+se_{n}=xe^{t\xi}e^{s\xi}+te_{n}+se_{n}=x\cdot\left(  t+s\right)  .
\]
In this case the projection map, $v\left(  x\right)  =x_{n}$ is a gauge with
\[
\pi_{\nu}\left(  x\right)  =xe^{-x_{n}\xi}-x_{n}e_{n}=\left(
\begin{array}
[c]{c}%
x_{1}\\
\vdots\\
x_{n-1}\\
0
\end{array}
\right)  e^{-x_{n}\xi}\text{ and }\mathcal{A}_{v}=\mathbb{R}^{n-1}%
\times\left\{  0\right\}  .
\]

\end{example}

\begin{example}
\label{ex.B.6}Let us specializing Example \ref{ex.B.5} to $n=3$ and%
\[
\xi=\left[
\begin{array}
[c]{ccc}%
0 & -1 & 0\\
1 & 0 & 0\\
0 & 0 & 0
\end{array}
\right]  \quad\implies\quad\text{ }e^{t\xi}=\left[
\begin{array}
[c]{ccc}%
\cos t & -\sin t & 0\\
\sin t & \cos t & 0\\
0 & 0 & 1
\end{array}
\right]  .
\]
In this case the gauge orbits are spirals. For example, the gauge orbit of
$e_{1}=\left(  1,0,0\right)  \in\mathbb{R}^{3}$ is the spiral, $\left\{
e_{1}\cdot t=\left(  \cos t,-\sin t,t\right)  :\text{ }t\in\mathbb{R}\right\}
. $ 

\end{example}

Examples \ref{ex.B.5} and \ref{ex.B.6} were examples of \textquotedblleft
affine actions,\textquotedblright\ which we now define.

\begin{definition}
[Affine actions]\label{def.B.7}Assume $\left(  \mathcal{A},\mathcal{G}\right)
$ as above with $\mathcal{A}$ being a finite dimensional vector space and let
$SL\left(  \mathcal{A}\right)  $ denote the special linear transformations on
$\mathcal{A}.$ We say the group action of $\mathcal{G}$ on $\mathcal{A}$ is an
\textbf{affine action }if it may be written in the form;%
\begin{equation}
Ag=\rho\left(  g^{-1}\right)  A+T\left(  g\right)  \label{e.B.2}%
\end{equation}
where $\rho:\mathcal{G}\rightarrow SL\left(  \mathcal{A}\right)  $ is a
representation of $\mathcal{G}$ and $T:\mathcal{G}\rightarrow\mathcal{A}$ is a
smooth function.
\end{definition}

\begin{remark}
\label{rem.B.8}It is left to the interested reader to verify that $T\left(
e\right)  =0$ and the pair, $\left(  \rho,T\right)  ,$ must satisfy the
\textquotedblleft cocylcle\textquotedblright\ condition;%
\begin{equation}
T\left(  gh\right)  =\rho\left(  h^{-1}\right)  T\left(  g\right)  +T\left(
h\right)  \text{ }\forall~g,h\in\mathcal{G}. \label{e.B.3}%
\end{equation}

\end{remark}

The key formal example of an affine action is the right action of the
restricted gauge group acting on connection one forms as in Eq. (\ref{e.1.1}).
We will work heuristically with this formal infinite dimensional setup in
Subsection \ref{sec.B.4} below.

\subsection{Disintegration formulas\label{sec.B.2}}

\begin{proposition}
[Disintegration]\label{pro.B.9}Let $\left(  \mathcal{A},\mathcal{G}%
,m,\lambda\right)  $ be as in Notation \ref{not.B.1}. To each gauge,
$v:\mathcal{A}\rightarrow\mathcal{G},$ there exists a unique (smooth) measure
$m_{v}$ on $\mathcal{A}_{v}$ such that
\begin{equation}
\int_{\mathcal{A}}f\left(  A\right)  dm\left(  A\right)  =\int_{\mathcal{A}%
_{v}}dm_{v}\left(  B\right)  \int_{\mathcal{G}}d\lambda\left(  g\right)
f\left(  Bg\right)  \label{e.B.4}%
\end{equation}
for all $f:\mathcal{A\rightarrow}\left[  0,\infty\right]  $ measurable.
\end{proposition}

\begin{proof}
Let $\gamma$ be a fixed smooth measure on $\mathcal{A}_{v}.$ Since the map in
Eq. (\ref{e.B.1}) is a diffeomorphism and Haar measure, $\lambda,$ is a smooth
measure on $\mathcal{G},$ there exists a smooth density, $\mu:\mathcal{A}%
_{v}\times\mathcal{G}\rightarrow\left(  0,\infty\right)  ,$ such that
\begin{equation}
\int_{\mathcal{A}}f\left(  A\right)  dm\left(  A\right)  =\int_{\mathcal{A}%
_{v}}d\gamma\left(  B\right)  \int_{\mathcal{G}}d\lambda\left(  g\right)
\mu\left(  B,g\right)  f\left(  Bg\right)  \label{e.B.5}%
\end{equation}
for all $f:\mathcal{A\rightarrow}\left[  0,\infty\right]  $ measurable. Using
the invariance of $m$ and $\lambda$ under the right $\mathcal{G}$-actions on
$\mathcal{A}$ and $\mathcal{G}$ respectively, if $k\in\mathcal{G},$ then
\begin{align}
\int_{\mathcal{A}}f\left(  A\right)  dm\left(  A\right)   &  =\int
_{\mathcal{A}}f\left(  Ak\right)  dm\left(  A\right) \nonumber\\
&  =\int_{\mathcal{A}_{v}}d\gamma\left(  B\right)  \int_{\mathcal{G}}%
d\lambda\left(  g\right)  \mu\left(  B,g\right)  f\left(  Bgk\right)
\nonumber\\
&  =\int_{\mathcal{A}_{v}}d\gamma\left(  B\right)  \int_{\mathcal{G}}%
d\lambda\left(  g\right)  \mu\left(  B,gk^{-1}\right)  f\left(  Bg\right)  .
\label{e.B.6}%
\end{align}
Comparing Eqs. (\ref{e.B.5}) and (\ref{e.B.6}) implies $\mu\left(
B,gk^{-1}\right)  =\mu\left(  B,g\right)  $ for all $B\in\mathcal{A}_{v}$ and
$g,k\in\mathcal{G}.$ Taking $k=g$ shows $\mu\left(  B,g\right)  =\mu\left(
B,e\right)  $ and so Eq. (\ref{e.B.4}) holds with $dm_{v}\left(  B\right)
:=\mu\left(  B,e\right)  d\gamma\left(  B\right)  .$
\end{proof}

\begin{theorem}
[Affine Action Disintegrations]\label{thm.B.10}Assume $\left(  \mathcal{A}%
,\mathcal{G}\right)  $ as above with $\mathcal{A}$ being a finite dimensional
vector space equipped with an affine action of $\mathcal{G}$ on $\mathcal{A},$
see Definition \ref{def.B.7}. Then;

\begin{enumerate}
\item Lebesgue measure $\left(  m\right)  $ on $\mathcal{A}$ is invariant
under the $\mathcal{G}$-action.

\item If $v:\mathcal{A}\rightarrow\mathcal{G}$ is a gauge such that
$\mathcal{A}_{v}$ is a linear subspace which is invariant under the action of
$\rho,$ then the measure $\left(  m_{\nu}\right)  $ in Proposition
\ref{pro.B.9} is a Lebesgue measure on $\mathcal{A}_{v}.$
\end{enumerate}
\end{theorem}

\begin{proof}
1. The Jacobian-determinant factor for the change of variables, $B=Ag,$ is
$\left\vert \det\rho\left(  g^{-1}\right)  \right\vert =1$ and hence the
affine transformation $A\rightarrow Ag$ leaves $m$ invariant on $\mathcal{A}.$

2. Let $m_{0}$ be a Lebesgue measure on $\mathcal{A}_{v}$ (i.e. a translation
invariant Radon measure on $\mathcal{A}_{v}\mathcal{)}.$ The smooth measure
$\left(  m_{\nu}\right)  $ may be expressed as $dm_{\nu}\left(  A\right)
=\mu\left(  A\right)  dm_{0}\left(  A\right)  $ for some smooth density
$\mu:\mathcal{A}\rightarrow\left(  0,\infty\right)  .$ Our goal is to show
that $\mu$ is a constant.

According to Proposition \ref{pro.B.9}, if $f:\mathcal{A}\rightarrow\left[
0,\infty\right]  $ is measurable, then%
\begin{equation}
\int_{\mathcal{A}}f\left(  C\right)  dm\left(  C\right)  =\int_{\mathcal{A}%
_{v}}dm_{\nu}\left(  A\right)  \int_{\mathcal{G}}d\lambda\left(  g\right)
f\left(  A\cdot g\right)  =\int_{\mathcal{G}}d\lambda\left(  g\right)
\int_{\mathcal{A}_{v}}dm_{0}\left(  A\right)  \mu\left(  A\right)  f\left(
A\cdot g\right)  . \label{e.B.7}%
\end{equation}
Let $B\in\mathcal{A}_{v}$ and apply Eq. (\ref{e.B.7}) with $f$ replaced by
$f\left(  \cdot+B\right)  $ to find%
\begin{align*}
\int_{\mathcal{A}}f\left(  C+B\right)  dm\left(  C\right)   &  =\int
_{\mathcal{G}}d\lambda\left(  g\right)  \int_{\mathcal{A}_{v}}dm_{0}\left(
A\right)  \mu\left(  A\right)  f\left(  A\cdot g+B\right) \\
&  =\int_{\mathcal{G}}d\lambda\left(  g\right)  \int_{\mathcal{A}_{v}}%
dm_{0}\left(  A\right)  \mu\left(  A\right)  f\left(  \left[  A+\mathrm{Ad}%
_{g}B\right]  ^{g}\right) \\
&  =\int_{\mathcal{G}}d\lambda\left(  g\right)  \int_{\mathcal{A}_{v}}%
dm_{0}\left(  A\right)  \mu\left(  A-\mathrm{Ad}_{g}B\right)  f\left(  A\cdot
g\right)  ,
\end{align*}
wherein the last line we have used $m_{0}$ is a translation invariant measure.
On the other hand $m$ is also translation invariant and so%
\[
\int_{\mathcal{A}}f\left(  C+B\right)  dm\left(  C\right)  =\int_{\mathcal{A}%
}f\left(  C\right)  dm\left(  C\right)  =\int_{\mathcal{G}}d\lambda\left(
g\right)  \int_{\mathcal{A}_{v}}dm_{0}\left(  A\right)  \mu\left(  A\right)
f\left(  A\cdot g\right)  .
\]
Using the map in Eq. (\ref{e.B.1}) is a diffeomorphism and the last two
displayed equations are valid for all measurable functions, $f:\mathcal{A}%
\rightarrow\left[  0,\infty\right]  ,$ we conclude that $\mu\left(
A-\mathrm{Ad}_{g}B\right)  =\mu\left(  A\right)  $ for all $A,B\in
\mathcal{A}_{v}$ and $g\in\mathcal{G}.$ Taking $g\equiv I$ and $B=A$ then
implies $\mu\left(  A\right)  =\mu\left(  0\right)  $ for all $A\in
\mathcal{A}_{\nu},$ i.e. $\mu$ is constant.
\end{proof}

\subsection{Abstract gauge fixing\label{sec.B.3}}

If $\Psi:\mathcal{A}\rightarrow\lbrack0,\infty)$ is a $\mathcal{G}$-invariant
function, then from Eq. (\ref{e.B.4}) it follows that
\begin{align*}
\int_{\mathcal{A}}\Psi\left(  A\right)  dm\left(  A\right)   &  =\int
_{\mathcal{A}_{v}}dm_{v}\left(  B\right)  \int_{\mathcal{G}}d\lambda\left(
g\right)  \Psi\left(  Bg\right) \\
&  =\int_{\mathcal{A}_{v}}dm_{v}\left(  B\right)  \int_{\mathcal{G}}%
d\lambda\left(  g\right)  \Psi\left(  B\right) \\
&  =\lambda\left(  \mathcal{G}\right)  \cdot\int_{\mathcal{A}_{v}}\Psi\left(
B\right)  dm_{v}\left(  B\right)  .
\end{align*}
This suggests that we normalize $\int_{\mathcal{A}}\Psi\left(  A\right)
dm\left(  A\right)  $ by \textquotedblleft dividing\textquotedblright\ the
integral by $\lambda\left(  \mathcal{G}\right)  $ and setting
\[
\dashint_{\mathcal{A}}\Psi\left(  A\right)  dm\left(  A\right)
=\text{\textquotedblleft}\frac{1}{\lambda\left(  \mathcal{G}\right)  }%
\int_{\mathcal{A}}\Psi\left(  A\right)  dm\left(  A\right)
\text{\textquotedblright.}%
\]
The problem with this formula is that (in the interesting cases)
$\lambda\left(  \mathcal{G}\right)  =\infty.$ To avoid this division by
infinity we make the following definition.

\begin{definition}
\label{def.B.11}The $v$-\textbf{normalized integral} of a $\mathcal{G}%
$-invariant function, $\Psi:\mathcal{A}\rightarrow\lbrack0,\infty)$, is%
\[
\dashint_{\mathcal{A}}\Psi\left(  A\right)  dm\left(  A\right)  =\int
_{\mathcal{A}_{v}}\Psi\left(  B\right)  dm_{v}\left(  B\right)  .
\]

\end{definition}

\begin{notation}
\label{not.B.12}Let $\Delta:\mathcal{G\rightarrow}\left(  0,\infty\right)  $
be the \textbf{modular function} on $\mathcal{G}$ defined by requiring%
\[
\int_{\mathcal{G}}\psi\left(  kg\right)  d\lambda\left(  g\right)
=\Delta\left(  k\right)  \cdot\int_{\mathcal{G}}\psi\left(  g\right)
d\lambda\left(  g\right)
\]
for all $k\in\mathcal{G}$ and $\psi:\mathcal{G}\rightarrow\left[
0,\infty\right]  $ measurable. Recall $\mathcal{G}$ is said to be
\textbf{unimodular }if $\Delta\equiv1.$
\end{notation}

\begin{theorem}
\label{thm.B.13}If $\mathcal{G}$ is a unimodular Lie group, $\Psi
:\mathcal{A}\rightarrow\lbrack0,\infty)$ is a $\mathcal{G}$-invariant
function, and $v,w$ are two gauges, then
\begin{equation}
\int_{\mathcal{A}_{v}}\Psi\left(  B\right)  dm_{v}\left(  B\right)
=\int_{\mathcal{A}_{w}}\Psi\left(  B\right)  dm_{w}\left(  B\right)  .
\label{e.B.8}%
\end{equation}

\begin{proof}
Let $\alpha\in C\left(  \mathcal{G},[0,\infty)\right)  $ such that
$\int_{\mathcal{G}}\alpha\left(  g\right)  d\lambda\left(  g\right)  =1$ and
set $f\left(  A\right)  :=\Psi\left(  A\right)  \alpha\left(  v\left(
A\right)  \right)  .$ By Eq. (\ref{e.B.4})) we find%
\begin{align*}
\int_{\mathcal{A}}f\left(  A\right)  dm\left(  A\right)   &  =\int
_{\mathcal{A}_{w}}dm_{w}\left(  B\right)  \int_{\mathcal{G}}d\lambda\left(
g\right)  \Psi\left(  Bg\right)  \alpha\left(  v\left(  Bg\right)  \right) \\
&  =\int_{\mathcal{A}_{w}}dm_{w}\left(  B\right)  \int_{\mathcal{G}}%
d\lambda\left(  g\right)  \Psi\left(  B\right)  \alpha\left(  v\left(
B\right)  g\right) \\
&  =\int_{\mathcal{A}_{w}}dm_{w}\left(  B\right)  \int_{\mathcal{G}}%
d\lambda\left(  g\right)  \Psi\left(  B\right)  \Delta\left(  v\left(
B\right)  \right)  \alpha\left(  g\right) \\
&  =\int_{\mathcal{A}_{w}}\Psi\left(  B\right)  \Delta\left(  v\left(
B\right)  \right)  dm_{w}\left(  B\right)  .
\end{align*}
In the case $w=v,$ so that $B\in\mathcal{A}_{v},$ we find
\begin{align*}
\int_{\mathcal{A}}f\left(  A\right)  dm\left(  A\right)   &  =\int
_{\mathcal{A}_{v}}\Psi\left(  B\right)  \Delta\left(  v\left(  B\right)
\right)  dm_{v}\left(  B\right) \\
&  =\int_{\mathcal{A}_{v}}\Psi\left(  B\right)  \Delta\left(  e\right)
dm_{v}\left(  B\right)  =\int_{\mathcal{A}_{v}}\Psi\left(  B\right)
dm_{v}\left(  B\right)  .
\end{align*}
Thus we have shown in general that
\[
\int_{\mathcal{A}_{v}}\Psi\left(  B\right)  dm_{v}\left(  B\right)
=\int_{\mathcal{A}_{w}}\Psi\left(  B\right)  \Delta\left(  v\left(  B\right)
\right)  dm_{w}\left(  B\right)
\]
and in particular if $\mathcal{G}$ is unimodular, Eq. (\ref{e.B.8}) holds.
\end{proof}
\end{theorem}

\begin{theorem}
\label{thm.B.14}Let $\left(  \mathcal{A},\mathcal{G},m,\lambda\right)  $ be as
in Notation \ref{not.B.1}, $v:\mathcal{A}\rightarrow\mathcal{G}$ be a gauge,
and $\varphi:\mathcal{A\rightarrow A}$ be a diffeomorphism such that;

\begin{enumerate}
\item $\varphi$ is volume preserving, i.e. $\varphi_{\ast}m=m.$

\item $\varphi$ acts equivariantly on $\mathcal{A}$ in the sense that there
exists a Lie group isomorphism, $\gamma:\mathcal{G}\rightarrow\mathcal{G},$
such that
\begin{equation}
\varphi\left(  Ag\right)  =\varphi\left(  A\right)  \gamma\left(  g\right)
~\text{ }\forall~A\in\mathcal{A}\text{ and }g\in\mathcal{G}. \label{e.B.9}%
\end{equation}
[Note this implies $\varphi$ preserves gauge orbits.]
\end{enumerate}

Under these assumptions, if $\Psi:\mathcal{A}\rightarrow\left[  0,\infty
\right]  $ is a gauge invariant function, then%
\begin{equation}
\int_{\mathcal{A}_{v}}\Psi\left(  \varphi\left(  A\right)  \right)
dm_{v}\left(  A\right)  =\frac{1}{c_{\gamma}}\int_{\mathcal{A}_{v}}\Psi\left(
A\right)  \Delta\left(  v\left(  \varphi^{-1}\left(  A\right)  \right)
\right)  dm_{v}\left(  A\right)  \label{e.B.10}%
\end{equation}
where $c_{\gamma}$ is the constant determined by%
\begin{equation}
\gamma_{\ast}\lambda=c_{\gamma}\lambda. \label{e.B.11}%
\end{equation}

\end{theorem}

\begin{proof}
For the moment let us simply suppose that $\varphi$ preserves gauge orbits
which may be stated as saying $\varphi\left(  Ag\right)  =\varphi\left(
A\right)  \Gamma\left(  A,g\right)  $ for some function $\Gamma:\mathcal{A}%
\times\mathcal{G}\rightarrow\mathcal{G}.$ As in the proof of Theorem
\ref{thm.B.13}, let $\alpha\in C\left(  \mathcal{G},[0,\infty)\right)  $ such
that $\int_{\mathcal{G}}\alpha\left(  g\right)  d\lambda\left(  g\right)  =1$
and $\Psi:\mathcal{A}\rightarrow\left[  0,\infty\right]  $ be a gauge
invariant function in which case,%
\[
\int_{\mathcal{A}_{v}}\Psi\left(  A\right)  dm_{v}\left(  A\right)
=\int_{\mathcal{A}}\Psi\left(  A\right)  \alpha\left(  v\left(  A\right)
\right)  dm\left(  A\right)  .
\]
Applying this identity with $\Psi$ replaced by $\Psi\circ\varphi$ gives,%
\begin{align}
\int_{\mathcal{A}_{v}}\Psi\circ\varphi\left(  A\right)  dm_{v}\left(
A\right)   &  =\int_{\mathcal{A}}\Psi\circ\varphi\left(  A\right)
\alpha\left(  v\left(  A\right)  \right)  dm\left(  A\right) \nonumber\\
&  =\int_{\mathcal{A}}\Psi\left(  A\right)  \alpha\left(  v\left(
\varphi^{-1}\left(  A\right)  \right)  \right)  dm\left(  A\right)
\text{\quad}\left(  \varphi_{\ast}m=m\right) \nonumber\\
&  =\int_{\mathcal{A}_{v}}dm_{v}\left(  B\right)  \int_{\mathcal{G}}%
d\lambda\left(  g\right)  ~\Psi\left(  Bg\right)  \alpha\left(  v\left(
\varphi^{-1}\left(  B\cdot g\right)  \right)  \right) \nonumber\\
&  =\int_{\mathcal{A}_{v}}dm_{v}\left(  B\right)  \int_{\mathcal{G}}%
d\lambda\left(  g\right)  ~\Psi\left(  B\right)  \alpha\left(  v\left(
\varphi^{-1}\left(  B\cdot g\right)  \right)  \right) \nonumber\\
&  =\int_{\mathcal{A}_{v}}\Psi\left(  B\right)  \mu_{\varphi}\left(  B\right)
dm_{v}\left(  B\right)  , \label{e.B.12}%
\end{align}
where $\mu_{\varphi}:\mathcal{A}_{v}\rightarrow\lbrack0,\infty)$ is defined
by
\[
\mu_{\varphi}\left(  B\right)  :=\int_{\mathcal{G}}\alpha\left(  v\left(
\varphi^{-1}\left(  B\cdot g\right)  \right)  \right)  d\lambda\left(
g\right)  .
\]

Let us now assume that $\varphi$ satisfies Eq. (\ref{e.B.9}). Applying
$\varphi^{-1}$ to Eq. (\ref{e.B.9}) with $A$ replaced by $\varphi^{-1}\left(
A\right)  $ and $g$ by $\gamma^{-1}\left(  g\right)  $ implies,%
\[
\varphi^{-1}\left(  A\right)  \gamma^{-1}\left(  g\right)  =\varphi
^{-1}\left(  Ag\right)  .
\]
Using this fact and noting that $\gamma_{\ast}\lambda=c_{\gamma}\lambda$
implies $\lambda=c_{\gamma}\left(  \gamma^{-1}\right)  _{\ast}\lambda,$ if
follows that%
\begin{align*}
\mu_{\varphi}\left(  B\right)   &  :=\int_{\mathcal{G}}\alpha\left(  v\left(
\varphi^{-1}\left(  B\right)  \gamma^{-1}\left(  g\right)  \right)  \right)
d\lambda\left(  g\right) \\
&  =\frac{1}{c_{\gamma}}\int_{\mathcal{G}}\alpha\left(  v\left(  \varphi
^{-1}\left(  B\right)  \right)  g\right)  d\lambda\left(  g\right)
=c_{\gamma}^{-1}\Delta\left(  v\left(  \varphi^{-1}\left(  B\right)  \right)
\right)  .
\end{align*}
Combining the last equation with Eq. (\ref{e.B.12}) gives Eq. (\ref{e.B.10}).
\end{proof}

\begin{remark}
\label{rem.B.15}One might hope to relax the condition in Eq. (\ref{e.B.9}) in
the previous theorem as follows. Suppose that $\varphi:\mathcal{A\rightarrow
A}$ is a diffeomorphism which takes gauge orbits to gauge orbits. Then define
$\tilde{\varphi}\left(  Ag\right)  =\pi_{v}\left(  \varphi\left(  A\right)
\right)  \cdot g$ for all $A\in\mathcal{A}_{v}$ and $g\in\mathcal{G}.$ Then if
$\Psi:\mathcal{A}\rightarrow\left[  0,\infty\right]  $ is a gauge invariant
function we will have%
\[
\Psi\left(  \tilde{\varphi}\left(  Ag\right)  \right)  =\Psi\left(  \pi
_{v}\left(  \varphi\left(  A\right)  \right)  \cdot g\right)  =\Psi\left(
\pi_{v}\left(  \varphi\left(  A\right)  \right)  \right)  =\Psi\left(
\varphi\left(  A\right)  \right)
\]
so that
\[
\int_{\mathcal{A}_{v}}\Psi\left(  \varphi\left(  A\right)  \right)
dm_{v}\left(  A\right)  =\int_{\mathcal{A}_{v}}\Psi\left(  \tilde{\varphi
}\left(  A\right)  \right)  dm_{v}\left(  A\right)  .
\]
The point being that $\tilde{\varphi}:\mathcal{A\rightarrow A}$ is a
diffeomorphism such that $\tilde{\varphi}\left(  Agk\right)  =\pi_{v}\left(
\varphi\left(  A\right)  \right)  \cdot gk=\tilde{\varphi}\left(  Ag\right)
k$ so that Eq. (\ref{e.B.9}) holds with $\gamma\left(  g\right)  =g.$
\textbf{However}, the problem is that there is no reason that $\tilde{\varphi
}$ should still preserve $m.$
\end{remark}

\begin{corollary}
\label{cor.B.16}Let us continue the notation and assumptions of Theorem
\ref{thm.B.14}. If we further assume that $\mathcal{G}$ is unimodular and
$c_{\gamma}=1,$ then%
\begin{equation}
\dashint_{\mathcal{A}}\Psi\left(  \varphi\left(  A\right)  \right)  dm\left(
A\right)  =\dashint_{\mathcal{A}}\Psi\left(  A\right)  dm\left(  A\right)
\nonumber
\end{equation}
for all gauge invariant functions, $\Psi:\mathcal{A}\rightarrow\left[
0,\infty\right]  .$
\end{corollary}

\begin{example}
[Example \ref{ex.B.4} continued]\label{ex.B.17}Let us continue the notation in
Example \ref{ex.B.4} and further assume that $\mathcal{G}$ is a unimodular Lie
group. Further let $m=\lambda^{\otimes N}$ where $\lambda$ is a Haar measure
on $\mathcal{G}$ and set $v\left(  \overrightarrow{g}\right)  =g_{1}$ so that
$\mathcal{A}_{v}=\left\{  e\right\}  \times\mathcal{G}^{N-1}.$ To make a gauge
invariant function, let $f:\mathcal{G}^{N-1}\rightarrow\mathbb{C}$ be any
function and the set $\Psi\left(  \overrightarrow{g}\right)  =f\left(
\left\{  g_{j}g_{1}^{-1}\right\}  _{j=2}^{N}\right)  .$ In this case, $m_{v}$
is given by $m_{v}=\delta_{e}\otimes\lambda^{\otimes\left(  N-1\right)  }$
since for $f:\mathcal{A}\rightarrow\mathbb{C}$ we have
\begin{align*}
\int_{\mathcal{A}}f\left(  \overrightarrow{g}\right)  dm\left(
\overrightarrow{g}\right)   &  =\int_{\mathcal{A}}f\left(
\begin{array}
[c]{c}%
g_{1}\\
g_{2}\\
\vdots\\
g_{N}%
\end{array}
\right)  d\lambda\left(  g_{1}\right)  d\lambda^{\otimes}\left(
\overrightarrow{g_{\geq2}}\right) \\
&  =\int_{\mathcal{A}}f\left(
\begin{array}
[c]{c}%
g_{1}\\
g_{2}g_{1}\\
\vdots\\
g_{N}g_{1}%
\end{array}
\right)  d\lambda\left(  g_{1}\right)  d\lambda^{\otimes}\left(
\overrightarrow{g_{\geq2}}\right) \\
&  =\int_{\mathcal{A}}f\left(
\begin{array}
[c]{c}%
e\\
g_{2}\\
\vdots\\
g_{N}%
\end{array}
\right)  \cdot g_{1}d\lambda\left(  g_{1}\right)  d\lambda^{\otimes}\left(
\overrightarrow{g_{\geq2}}\right)  .
\end{align*}
For an example of a $\varphi:\mathcal{A\rightarrow A}$ satisfying the
assumption of Corollary \ref{cor.B.16}, fix $a,b\in\mathcal{G}$ and then
define,%
\[
\varphi\left(  \overrightarrow{g}\right)  =a\cdot\overrightarrow{g}\cdot
b:=\left(
\begin{array}
[c]{c}%
ag_{1}b\\
ag_{2}b\\
\vdots\\
ag_{N}b
\end{array}
\right)  .
\]
Then $\varphi$ an $m$-preserving diffeomorphism on $\mathcal{A}$ with
\[
\varphi\left(  \overrightarrow{g}\cdot k\right)  =a\cdot\overrightarrow
{g}\cdot k\cdot b=a\cdot\overrightarrow{g}\cdot b\cdot b^{-1}kb=\varphi\left(
\overrightarrow{g}\right)  \cdot\gamma\left(  k\right)
\]
where $\gamma\left(  k\right)  :=\mathrm{Ad}_{b^{-1}}k,$ and so $\varphi$
satisfies the hypothesis of Corollary \ref{cor.B.16}.
\end{example}

\subsection{Yang-Mills gauge fixing\label{sec.B.4}}

In this section, we suppose (as defined in Notation \ref{not.1.1} with
$M=\mathbb{R}^{2}$) that $\mathcal{A}:=\Omega^{1}\left(  \mathbb{R}%
^{2},\mathfrak{k}\right)  ,$ $\mathcal{G}$ is the\textbf{ }gauge group of
functions, $g:\mathbb{R}^{2}\rightarrow K,$ and $\mathcal{G}_{o}=\left\{
g\in\mathcal{G}:g\left(  o\right)  =I\right\}  $ is the restricted gauge group.

\begin{definition}
[Homotopies]\label{def.B.18}A continuous map, $\mathbb{R}^{d}\times\left[
0,1\right]  \ni\left(  x,t\right)  \rightarrow\sigma_{x}\left(  t\right)
\in\mathbb{R}^{d}$ is a homotopy contracting $\mathbb{R}^{d}$ to $\left\{
0\right\}  $ if $\sigma_{x}\left(  1\right)  =x$ and $\sigma_{x}\left(
0\right)  =0$ for all $x\in\mathbb{R}^{d}.$ We further say $\sigma$ is a
\textbf{follow the leader homotopy} if $\sigma_{\sigma_{x}\left(  t\right)  }$
is a reparametrization of $\sigma_{x}|_{\left[  0,t\right]  }$ for all
$x\in\mathbb{R}^{d}$ and $t\in(0,1].$ [We will further assume that
$t\rightarrow\sigma_{x}\left(  t\right)  $ is at least piecewise smooth.]
\end{definition}

\begin{example}
\label{ex.B.19}The \textbf{radial homotopy}, $\sigma,$ is define by
$\sigma_{x}\left(  t\right)  =tx$ for all $x\in\mathbb{R}^{d}$ and
$t\in\left[  0,1\right]  .$ This is a follow the leader homotopy.
\end{example}

In the main part of this paper we have secretly been using the following
\textquotedblleft complete axial homotopy\textquotedblright\ on $\mathbb{R}%
^{2}$, another follow the leader homotopy.

\begin{notation}
[Complete axial homotopy]\label{n.B.20}For any $x\in\mathbb{R}^{2},$ let
$\sigma_{x}$ be the straight line path joining $0$ to $\left(  x_{1},0\right)
$ followed by the straight line path joining $\left(  x_{1},0\right)
\rightarrow\left(  x_{1},x_{2}\right)  =x$ as in Figure \ref{fig.12}. We refer
to this homotopy as the\textbf{ complete axial homotopy.}\begin{figure}[ptbh]
\centering
\par
\psize{2.5in} %
\executeiffilenewer{\GraphicsDirectorysigmap.svg}{\GraphicsDirectorysigmap.pdf}%
{inkscape -z -D --file=\GraphicsDirectorysigmap.svg --export-pdf=\GraphicsDirectorysigmap.pdf --export-latex}%
\input{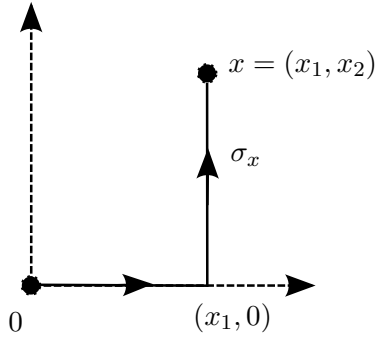}%
\caption{The taxi-cab path, $\varphi_{x},$
joining $0$ to $x\in\mathbb{R}^{2}.$}%
\label{fig.12}%
\end{figure}

\end{notation}

\begin{example}
\label{ex.B.21}If $\left\{  \sigma_{x}:x\in\mathbb{R}^{d}\right\}  $ is a
homotopy contracting $\mathbb{R}^{d}$ to $\left\{  0\right\}  ,$ then
$v_{\sigma}\left(  A\right)  \in\mathcal{G}$ defined by%
\begin{equation}
v_{\sigma}\left(  A\right)  \left(  x\right)  :=\left[  \pt_{1}^{A}\left(
\sigma_{x}\right)  \right]  ^{-1}\text{ for }x\in\mathbb{R}^{d} \label{e.B.13}%
\end{equation}
is a gauge on $\mathcal{A}$ because,
\begin{align*}
v_{\sigma}\left(  A^{g}\right)  \left(  x\right)   &  =\left[  \pt_{1}^{A^{g}%
}\left(  \sigma_{x}\right)  \right]  ^{-1}=\left[  g^{-1}\left(  x\right)
\pt_{1}^{A}\left(  \sigma_{x}\right)  g\left(  0\right)  \right]  ^{-1}\\
&  =\left[  \pt_{1}^{A}\left(  \sigma_{x}\right)  \right]  ^{-1}g\left(
x\right)  =v_{\sigma}\left(  A\right)  \left(  x\right)  g\left(  x\right)
=\left(  v_{\sigma}\left(  A\right)  g\right)  \left(  x\right)  .
\end{align*}
In this case we write $\mathcal{A}_{\sigma}$ for $\mathcal{A}_{v_{\sigma}}$ so
that%
\begin{equation}
\mathcal{A}_{\sigma}:=v_{\sigma}^{-1}\left(  \left\{  I\right\}  \right)
=\left\{  A\in\mathcal{A}:\pt_{1}^{A}\left(  \sigma_{x}\right)  =I\text{ for
all }x\in\mathbb{R}^{d}\right\}  \subset\mathcal{A}. \label{e.B.14}%
\end{equation}

\end{example}

\begin{definition}
[Homotopy gauges]If $\sigma$ is a homotopy contracting $\mathbb{R}^{d}$ to
$\left\{  0\right\}  $ we refer to $v_{\sigma}$ in Eq. (\ref{e.B.13}) as
\textbf{homotopy gauge }and $\mathcal{A}_{\sigma}$ in Eq. (\ref{e.B.14}) as a
\textbf{homotopy slice.}
\end{definition}

\begin{proposition}
[Follow the leader gauges]\label{pro.B.23}If $\sigma$ is a follow the leader
homotopy and $A$ is a connection one form then the following are equivalent;

\begin{enumerate}
\item $A$ is in the $\sigma$-gauge (i.e. $\pt_{1}^{A}\left(  \sigma
_{x}\right)  =I$ for all $x\in\mathbb{R}^{d}),$

\item $\pt_{t}^{A}\left(  \sigma_{x}\right)  =I$ for all $x\in\mathbb{R}^{d}$
and $t\in\left[  0,1\right]  ,$ and

\item $A\left\langle \dot{\sigma}_{x}\left(  t\right)  \right\rangle =0$ for
all $x\in\mathbb{R}^{d}$ and a.e. $t\in\left[  0,1\right]  .$
\end{enumerate}

Consequently, by item 3. above,
\[
\mathcal{A}_{\sigma}=\left\{  A\in\mathcal{A}:A\left\langle \dot{\sigma}%
_{x}\left(  t\right)  \right\rangle =0\text{ for all }x\in\mathbb{R}^{d}\text{
and a.e. }t\in\left[  0,1\right]  \right\}
\]
is a linear slice for any follow the leader homotopy.
\end{proposition}

\begin{proof}
Since \thinspace$\pt^{A}\left(  \sigma\right)  $ invariant under
reparametrizations of $\sigma$ it follows that for a follow the leader gauge,
\[
\pt_{t}^{A}\left(  \sigma_{x}\right)  =\pt_{1}^{A}\left(  \sigma_{\sigma
_{x}\left(  t\right)  }\right)  \text{ }\forall~x\in\mathbb{R}^{d}\text{ and
}t\in\left[  0,1\right]
\]
and this shows $1.\implies2.$ To prove $2.$ implies 3. simply notice that
\[
0=\frac{d}{dt}I=\frac{d}{dt}\pt_{t}^{A}\left(  \sigma_{x}\right)
=-A\left\langle \dot{\sigma}_{x}\left(  t\right)  \right\rangle \pt_{t}%
^{A}\left(  \sigma_{x}\right)  =-A\left\langle \dot{\sigma}_{x}\left(
t\right)  \right\rangle .
\]
The assertion $3.\implies1.$ is obvious since $\pt_{t}^{A}\left(  \sigma
_{x}\right)  $ satisfies
\[
0=\frac{d}{dt}\pt_{t}^{A}\left(  \sigma_{x}\right)  +A\left\langle \dot
{\sigma}_{x}\left(  t\right)  \right\rangle \pt_{t}^{A}\left(  \sigma
_{x}\right)  =\frac{d}{dt}\pt_{t}^{A}\left(  \sigma_{x}\right)  .
\]

\end{proof}

In what follows, if $x,v\in\mathbb{R}^{2},$ we let $v_{x}\in T\mathbb{R}^{d}$
be the tangent vector defined by
\[
v_{x}f:=\left(  \partial_{v}f\right)  \left(  x\right)  =\frac{d}{ds}%
|_{0}f\left(  x+sv\right)
\]
where $f$ is any differentiable function on $\mathbb{R}^{2}.$

\begin{corollary}
\label{cor.B.24}If $\sigma$ is a follow the leader homotopy, $A\in
\mathcal{A}_{\sigma},$ and $F=F^{A}$ is the curvature of $A,$ then we can
recover $A$ from $F$ using
\begin{equation}
A\left\langle v_{x}\right\rangle =\int_{0}^{1}F\left\langle \dot{\sigma}%
_{x}\left(  t\right)  ,v_{x}\sigma_{\left(  \cdot\right)  }\left(
\tau\right)  \right\rangle dt\text{ }\forall~v_{x}\in T\mathbb{R}^{d},
\label{e.B.15}%
\end{equation}
where explicitly,
\[
v_{x}\sigma_{\left(  \cdot\right)  }\left(  \tau\right)  =\partial_{v}%
\sigma_{x}\left(  \tau\right)  =\frac{d}{ds}|_{0}\sigma_{\left(  x+sv\right)
}\left(  \tau\right)  .
\]

\end{corollary}

\begin{proof}
Using Cartan's formula while repeatedly using item 3. of Proposition
\ref{pro.B.23} shows%
\[
dA\left\langle \dot{\sigma}_{x}\left(  t\right)  ,\partial_{v}\sigma
_{x}\left(  t\right)  \right\rangle =\frac{d}{dt}A\left\langle \partial
_{v}\sigma_{x}\left(  t\right)  \right\rangle -\partial_{v}A\left\langle
\dot{\sigma}_{x}\left(  t\right)  \right\rangle =\frac{d}{dt}A\left\langle
\partial_{v}\sigma_{x}\left(  t\right)  \right\rangle
\]
and
\[
\left(  A\wedge A\right)  \left\langle \dot{\sigma}_{x}\left(  t\right)
,\partial_{v}\sigma_{x}\left(  t\right)  \right\rangle =\left[  A\left\langle
\dot{\sigma}_{x}\left(  t\right)  \right\rangle ,A\left\langle \partial
_{v}\sigma_{x}\left(  t\right)  \right\rangle \right]  =0.
\]
Therefore we may conclude that%
\[
F\left\langle \dot{\sigma}_{x}\left(  t\right)  ,\partial_{v}\sigma_{x}\left(
t\right)  \right\rangle =\frac{d}{dt}A\left\langle \partial_{v}\sigma
_{x}\left(  t\right)  \right\rangle .
\]
Integrating this expression on $t$ while using $\partial_{v}\sigma_{x}\left(
0\right)  =\partial_{v}0=0$ and $\partial_{v}\sigma_{x}\left(  1\right)
=\partial_{v}x=v_{x}$ gives Eq. (\ref{e.B.15}).
\end{proof}

Let us generalize the previous result in order to compute $\pi_{\sigma}\left(
A\right)  $ for an arbitrary $A\in\mathcal{A}.$

\begin{theorem}
\label{thm.B.25}If $\sigma$ is a follow the leader homotopy and $A\in
\mathcal{A},$ then%
\begin{equation}
\pi_{\sigma}\left(  A\right)  \left\langle v_{x}\right\rangle =\mathrm{Ad}%
_{g_{A}\left(  x\right)  ^{-1}}\int_{0}^{1}\mathrm{Ad}_{\pt_{1\leftarrow\tau
}^{A}\left(  \sigma_{x}\right)  }F^{A}\left(  \dot{\sigma}_{x}\left(
\tau\right)  ,v_{x}\sigma_{\left(  \cdot\right)  }\left(  \tau\right)
\right)  d\tau\label{e.B.16}%
\end{equation}
where $\pi_{\sigma}\left(  A\right)  :=A^{g_{A}}$ with $g_{A}\left(  x\right)
:=\pt_{1}^{A}\left(  \sigma_{x}\left(  \cdot\right)  \right)  $ and
\[
\pt_{1\leftarrow\tau}^{A}\left(  \sigma_{x}\right)  :=\pt_{1}^{A}\left(
\sigma_{x}\right)  \pt_{\tau}^{A}\left(  \sigma_{x}\right)  ^{-1}.
\]

\end{theorem}

\begin{proof}
Let $v_{x}\in T_{x}\mathbb{R}^{d}$ and $x\left(  s\right)  \in\mathbb{R}^{d}$
such that $x^{\prime}\left(  0\right)  =v_{x}$ and in particular $x\left(
0\right)  =x.$ We are now going to apply Proposition \ref{pro.A.7} with
$\ell_{s}\left(  t\right)  =\sigma_{x\left(  s\right)  }\left(  t\right)  .$
First observe that%
\[
dg_{A}\left(  v_{x}\right)  +A\left\langle v_{x}\right\rangle g_{A}\left(
x\right)  =\frac{\nabla}{ds}|_{0}g_{A}\left(  x\left(  s\right)  \right)
=\frac{\nabla}{ds}|_{0}\pt_{1}^{A}\left(  \sigma_{x\left(  s\right)  }\right)
=\frac{\nabla}{ds}|_{0}\pt_{1}^{A}\left(  \ell_{s}\right)
\]
and, by Eq. (\ref{e.A.10}) of Proposition \ref{pro.A.7},
\begin{align*}
\frac{\nabla}{ds}|_{0}\pt_{1}^{A}\left(  \ell_{s}\right)   &  =\pt_{1}%
^{A}\left(  \ell_{0}\right)  \int_{0}^{1}\mathrm{Ad}_{\pt_{\tau}^{A}\left(
\ell_{0}\right)  ^{-1}}F^{A}\left(  \dot{\ell}_{0}\left(  \tau\right)
,\ell_{0}^{\prime}\left(  \tau\right)  \right)  d\tau\\
&  =g_{A}\left(  x\right)  \int_{0}^{1}\mathrm{Ad}_{\pt_{\tau}^{A}\left(
\sigma_{x}\right)  ^{-1}}F^{A}\left(  \dot{\sigma}_{x}\left(  \tau\right)
,v_{x}\sigma_{\left(  \cdot\right)  }\left(  \tau\right)  \right)  d\tau\\
&  =g_{A}\left(  x\right)  \int_{0}^{1}\mathrm{Ad}_{g_{A}\left(  \sigma
_{x}\left(  \tau\right)  \right)  ^{-1}}F^{A}\left(  \dot{\sigma}_{x}\left(
\tau\right)  ,v_{x}\sigma_{\left(  \cdot\right)  }\left(  \tau\right)
\right)  d\tau.
\end{align*}
Combining these two identities and multiplying the result on the left by
$g_{A}\left(  x\right)  ^{-1}$ gives,%
\begin{align*}
B\left\langle v_{x}\right\rangle  &  =A^{g_{A}}\left\langle v_{x}\right\rangle
=\int_{0}^{1}\mathrm{Ad}_{g_{A}\left(  \sigma_{x}\left(  \tau\right)  \right)
^{-1}}F^{A}\left(  \dot{\sigma}_{x}\left(  \tau\right)  ,v_{x}\sigma_{\left(
\cdot\right)  }\left(  \tau\right)  \right)  d\tau\\
&  =\mathrm{Ad}_{g_{A}\left(  x\right)  ^{-1}}\int_{0}^{1}\mathrm{Ad}%
_{g_{A}\left(  x\right)  }\mathrm{Ad}_{g_{A}\left(  \sigma_{x}\left(
\tau\right)  \right)  ^{-1}}F^{A}\left(  \dot{\sigma}_{x}\left(  \tau\right)
,v_{x}\sigma_{\left(  \cdot\right)  }\left(  \tau\right)  \right)  d\tau.
\end{align*}
Finally we have $g_{A}\left(  x\right)  g_{A}\left(  \sigma_{x}\left(
\tau\right)  \right)  ^{-1}=\pt_{1\leftarrow\tau}^{A}\left(  \sigma
_{x}\right)  $ so that
\[
A^{g_{A}}\left\langle v_{x}\right\rangle =\mathrm{Ad}_{g_{A}\left(  x\right)
^{-1}}\int_{0}^{1}\mathrm{Ad}_{\pt_{1\leftarrow\tau}^{A}\left(  \sigma
_{x}\right)  }F^{A}\left(  \dot{\sigma}_{x}\left(  \tau\right)  ,v_{x}%
\sigma_{\left(  \cdot\right)  }\left(  \tau\right)  \right)  d\tau.
\]

\end{proof}

\begin{remark}
\label{rem.B.26}If we take $v_{x}=\frac{d}{ds}\sigma_{y}\left(  s\right)  $
for some $y\in\mathbb{R}^{d}$ and $s\in\left[  0,1\right]  $ (so that
$x=\sigma_{y}\left(  s\right)  ),$ then
\[
\dot{\sigma}_{x}\left(  \tau\right)  =\frac{d}{d\tau}\sigma_{\sigma_{y}\left(
s\right)  }\left(  \tau\right)  \text{ and }v_{x}\sigma_{\left(  \cdot\right)
}\left(  \tau\right)  =\frac{d}{ds}\sigma_{\sigma_{y}\left(  s\right)
}\left(  \tau\right)
\]
are parallel by the follow the leader property so that $F^{A}\left(
\dot{\sigma}_{x}\left(  \tau\right)  ,v_{x}\sigma_{\left(  \cdot\right)
}\left(  \tau\right)  \right)  =0$ for a.e. $\tau$ in this case. This shows
explicitly that right side of Eq. (\ref{e.B.16}) is indeed in $\mathcal{A}%
_{\sigma}.$
\end{remark}

\begin{corollary}
\label{cor.B.27}If $A\in\mathcal{A}_{\sigma}$ and $\eta\in\mathcal{A},$ then
\begin{align}
\pi_{\sigma}\left(  A+\eta\right)  \left\langle v_{x}\right\rangle  &
=\left(  A+\eta\right)  ^{g_{\eta}}\left\langle v_{x}\right\rangle =\left[
\mathrm{Ad}_{g_{\eta}^{-1}}A+\pi_{\sigma}\left(  \eta\right)  \right]
\left\langle v_{x}\right\rangle \label{e.B.17}\\
&  =\mathrm{Ad}_{g_{\eta}^{-1}\left(  x\right)  }\left[  A\left\langle
v_{x}\right\rangle +\int_{0}^{1}\mathrm{Ad}_{\pt_{1\leftarrow\tau}^{\eta
}\left(  \sigma_{x}\right)  }F^{\eta}\left(  \dot{\sigma}_{x}\left(
\tau\right)  ,v_{x}\sigma_{\left(  \cdot\right)  }\left(  \tau\right)
\right)  d\tau\right]  . \label{e.B.18}%
\end{align}

\end{corollary}

\begin{proof}
As $A\in\mathcal{A}_{\sigma}$ we know that $A\left\langle \dot{\sigma}%
_{x}\left(  t\right)  \right\rangle =0$ for a.e. $t$ and therefore,
\[
g_{A+\eta}\left(  x\right)  =\pt_{1}^{A+\eta}\left(  \sigma_{x}\right)
=\pt_{1}^{\eta}\left(  \sigma_{x}\right)  =g_{\eta}\left(  x\right)
\]
and hence
\begin{align*}
\pi_{\sigma}\left(  A+\eta\right)   &  =\left(  A+\eta\right)  ^{g_{\eta}%
}=\mathrm{Ad}_{g_{\eta}^{-1}}\left(  A+\eta\right)  +g_{\eta}^{-1}dg_{\eta}\\
&  =\mathrm{Ad}_{g_{\eta}^{-1}}A+\mathrm{Ad}_{g_{\eta}^{-1}}\eta+g_{\eta}%
^{-1}dg_{\eta}=\mathrm{Ad}_{g_{\eta}^{-1}}A+\pi_{\sigma}\left(  \eta\right)
\end{align*}
which gives Eq. (\ref{e.B.17}). Making use of Theorem \ref{thm.B.25} with $A$
replaced by $\eta$ to evaluate $\pi_{\sigma}\left(  \eta\right)  $ in Eq.
(\ref{e.B.17}) then gives Eq. (\ref{e.B.18}).
\end{proof}

\begin{metapropositions}
\label{mpro.B.28}If $\sigma$ is a follow the leader homotopy, $\eta
\in\mathcal{A},$ and $\Psi:\mathcal{A}\rightarrow\left[  0,\infty\right]  $ is
a function such that $\Psi\left(  A^{g_{\eta}}\right)  =\Psi\left(  A\right)
$ for all $A\in\mathcal{A},$ then
\begin{equation}
\int_{\mathcal{A}_{\sigma}}\Psi\left(  A+\eta\right)  dm_{\sigma}\left(
A\right)  =\int_{\mathcal{A}_{\sigma}}\Psi\left(  A\right)  dm_{\sigma}\left(
A\right)  . \label{e.B.19}%
\end{equation}
[Note that $\eta$ is not assumed to be in $\mathcal{A}_{\sigma}$ and so we can
not directly prove Eq. (\ref{e.B.19}) by invoking translation invariance of
$m_{\sigma}.]$
\end{metapropositions}

\begin{proof}
If $A\in\mathcal{A}_{\sigma}$ and $\eta\in\mathcal{A},$ then $A+\eta
\in\mathcal{A}$ and so by assumption and
\[
\Psi\left(  A+\eta\right)  =\Psi\left(  \left[  A+\eta\right]  ^{g_{\eta}%
}\right)  =\Psi\left(  \mathrm{Ad}_{g_{\eta}^{-1}}A+\pi_{\sigma}\left(
\eta\right)  \right)  .
\]
As we have already explained, $A\rightarrow\mathrm{Ad}_{g_{\eta}^{-1}}%
A+\pi_{\sigma}\left(  \eta\right)  $ is a rotation followed by a translation
which preserves Lebesgue measure and $m_{\sigma}$ is a Lebesgue measure. Thus,
it follows that
\[
\int_{\mathcal{A}_{\sigma}}\Psi\left(  A+\eta\right)  dm_{\sigma}\left(
A\right)  =\int_{\mathcal{A}_{\sigma}}\Psi\left(  \mathrm{Ad}_{g_{\eta}^{-1}%
}A+\pi_{\sigma}\left(  \eta\right)  \right)  dm_{\sigma}\left(  A\right)
=\int_{\mathcal{A}_{\sigma}}\Psi\left(  A\right)  dm_{\sigma}\left(  A\right)
.
\]

\end{proof}

\begin{metacorollary}
\label{mcor.B.29}If $\sigma$ is a follow the leader homotopy, $\eta
\in\mathcal{A},$ and $\Psi:\mathcal{A}\rightarrow\left[  0,\infty\right]  $ is
a function such that $\Psi\left(  A^{g_{s\eta}}\right)  =\Psi\left(  A\right)
$ for all $A\in\mathcal{A}$ and $s\in\left(  -\varepsilon,\varepsilon\right)
$ for some $\varepsilon>0,$ then
\begin{equation}
\int_{\mathcal{A}_{\sigma}}\left(  \partial_{\eta}\Psi\right)  \left(
A\right)  dm_{\sigma}\left(  A\right)  =0. \label{e.B.20}%
\end{equation}

\end{metacorollary}

\begin{proof}
By Proposition \ref{mpro.B.28},
\[
\int_{\mathcal{A}_{\sigma}}\Psi\left(  A\right)  dm_{\sigma}\left(  A\right)
=\int_{\mathcal{A}_{\sigma}}\Psi\left(  A+s\eta\right)  dm_{\sigma}\left(
A\right)  \text{ }\forall~s\in\left(  -\varepsilon,\varepsilon\right)  .
\]
Differentiating this equation at $s=0$ then gives Eq. (\ref{e.B.20}).
\end{proof}

\begin{remark}
\label{rem.B.30}\textbf{Warning: }even if $\Psi$ is gauge invariant, it is
quite unlikely that $\partial_{\eta}\Psi$ will still be gauge invariant since
in general we have,
\begin{align*}
\left(  \partial_{\eta}\Psi\right)  \left(  A^{g}\right)   &  =\frac{d}%
{dt}|_{0}\Psi\left(  A^{g}+t\eta\right)  =\frac{d}{dt}|_{0}\Psi\left(  \left(
A^{g}+t\eta\right)  ^{g^{-1}}\right) \\
&  =\frac{d}{dt}|_{0}\Psi\left(  A+t\mathrm{Ad}_{g}\eta\right)  =\left(
\partial_{\mathrm{Ad}_{g}\eta}\Psi\right)  \left(  A\right)  .
\end{align*}
So in order for $\partial_{\eta}\Psi$ to be gauge invariant we would typically
need $\mathrm{Ad}_{g}\eta=\eta$ for all $g\in\mathcal{G}$ which would force
$\eta$ to take values in the center of $\mathfrak{k}.$ On the other hand, for
any $g\in\mathcal{G}$ such that $\mathrm{Ad}_{g}\eta=\eta$ we will have
\[
\left(  \partial_{\eta}\Psi\right)  \left(  A^{g}\right)  =\left(
\partial_{\eta}\Psi\right)  \left(  A\right)  \text{ for all }A\in
\mathcal{A}.
\]

\end{remark}

\begin{notation}
[$\sigma$-fixed $YM$ \textquotedblleft measures\textquotedblright%
]\label{not.B.31}To each follow the leader homotopy, $\sigma,$ let
$\mu_{\sigma}$ be the formal probability measure on $\mathcal{A}_{\sigma}$
given by,%
\[
d\mu_{\sigma}\left(  A\right)  =\frac{1}{Z_{\sigma}}e^{-\frac{1}{2}\left\Vert
F^{A}\right\Vert ^{2}}dm_{\sigma}\left(  A\right)  .
\]

\end{notation}

\begin{metacorollary}
\label{mcor.B.32}If $\sigma$ is a follow the leader homotopy, $\eta
\in\mathcal{A},$ and $\Psi:\mathcal{A}\rightarrow\left[  0,\infty\right]  $ is
a function such that $\Psi\left(  A^{g_{s\eta}}\right)  =\Psi\left(  A\right)
$ for all $A\in\mathcal{A}$ and $s\in\left(  -\varepsilon,\varepsilon\right)
$ for some $\varepsilon>0,$ then%
\begin{equation}
\int_{\mathcal{A}_{\sigma}}\left(  \partial_{\eta}\Psi\right)  \left(
A\right)  d\mu_{\sigma}\left(  A\right)  =\int_{\mathcal{A}_{\sigma}}%
\Psi\left(  A\right)  \cdot\left\langle d^{A}\eta,F^{A}\right\rangle
d\mu_{\sigma}\left(  A\right)  , \label{e.B.21}%
\end{equation}
where
\[
\left(  d^{A}\eta\right)  _{ij}=\nabla_{i}^{A}\eta_{j}-\nabla_{j}^{A}\eta
_{i}\text{ and }\nabla_{i}^{A}\eta_{j}:=\partial_{i}\eta_{j}+\mathrm{ad}%
_{A_{i}}\eta_{j}.
\]
\textbf{Warning: }gauge invariance has been broken in Eq. (\ref{e.B.21}) which
holds for all follow the leader homotopies, $\sigma,$ but both sides of this
equation may very well depend on the choice of $\sigma.$
\end{metacorollary}

\begin{proof}
Since $A\rightarrow e^{-\frac{1}{2}\left\Vert F^{A}\right\Vert ^{2}}$ is gauge
invariant we may apply Meta-Corollary \ref{mcor.B.29} with $\Psi$ replaced by
$A\rightarrow\Psi\left(  A\right)  e^{-\frac{1}{2}\left\Vert F^{A}\right\Vert
^{2}}$ in order to find,%
\begin{align*}
0  &  =\frac{1}{Z_{\sigma}}\int_{\mathcal{A}_{\sigma}}\partial_{\eta}\left[
A\rightarrow\Psi\left(  A\right)  e^{-\frac{1}{2}\left\Vert F^{A}\right\Vert
^{2}}\right]  dm_{\sigma}\left(  A\right) \\
&  =\frac{1}{Z_{\sigma}}\int_{\mathcal{A}_{\sigma}}\left[  \left(
\partial_{\eta}\Psi\right)  \left(  A\right)  -\frac{1}{2}\partial_{\eta
}\left\Vert F^{A}\right\Vert ^{2}\right]  e^{-\frac{1}{2}\left\Vert
F^{A}\right\Vert ^{2}}dm_{\sigma}\left(  A\right)
\end{align*}
wherein we have used the product and the chain rule for the second equality.
This completes the proof since,%
\[
\frac{1}{2}\partial_{\eta}\left\Vert F^{A}\right\Vert ^{2}=\left\langle
\partial_{\eta}F^{A},F^{A}\right\rangle
\]
and
\[
\partial_{\eta}F_{ij}^{A}=\partial_{\eta}\left(  \partial_{i}A_{j}%
-\partial_{j}A_{i}+\left[  A_{i},A_{j}\right]  \right)  =\partial_{i}\eta
_{j}-\partial_{j}\eta_{i}+\left[  \eta_{i},A_{j}\right]  +\left[  A_{i}%
,\eta_{j}\right]  =\nabla_{i}^{A}\eta_{j}-\nabla_{j}^{A}\eta_{i}.
\]

\end{proof}

\begin{notation}
\label{not.B.33}To each follow the leader homotopy, $\sigma,$ and $\eta
\in\mathcal{A}$ let $u_{\eta}^{\sigma}:\mathbb{R}^{d}\rightarrow\mathfrak{k}$
be defined by
\[
u_{\eta}^{\sigma}\left(  x\right)  :=\int_{0}^{1}\eta\left\langle \dot{\sigma
}_{x}\left(  \tau\right)  \right\rangle d\tau.
\]

\end{notation}

\begin{remark}
\label{rem.B.34} Since $\sigma$ is a follow the leader homotopy we have,%
\[
u_{\eta}^{\sigma}\left(  \sigma_{x}\left(  t\right)  \right)  =\int_{0}%
^{1}\eta\left\langle \dot{\sigma}_{\sigma_{x}\left(  t\right)  }\left(
\tau\right)  \right\rangle d\tau=\int_{\sigma_{x}|_{\left[  0,t\right]  }}%
\eta=\int_{0}^{t}\eta\left\langle \dot{\sigma}_{x}\left(  \tau\right)
\right\rangle d\tau
\]
and therefore,
\begin{equation}
du_{\eta}^{\sigma}\left(  \dot{\sigma}_{x}\left(  t\right)  \right)  =\frac
{d}{dt}u_{\eta}^{\sigma}\left(  \sigma_{x}\left(  t\right)  \right)  =\frac
{d}{dt}\int_{0}^{t}\eta\left\langle \dot{\sigma}_{x}\left(  \tau\right)
\right\rangle d\tau=\eta\left\langle \dot{\sigma}_{x}\left(  t\right)
\right\rangle . \label{e.B.22}%
\end{equation}

\end{remark}

\begin{proposition}
[Projected vector fields]\label{pro.B.35}If $A\in\mathcal{A}_{\sigma}$ and
$A,\eta\in\mathcal{A},$ then
\begin{equation}
d\pi_{\sigma}\left\langle \eta_{A}\right\rangle :=\frac{d}{ds}|_{0}\pi
_{\sigma}\left(  A+s\eta\right)  =-\mathrm{ad}_{u_{\eta}^{\sigma}}%
A+\eta-du_{\eta}^{\sigma}. \label{e.B.23}%
\end{equation}
[As is seen directly from Eq. (\ref{e.B.22}), $\eta-du_{\eta}^{\sigma}%
\in\mathcal{A}_{\sigma}$ for all $\eta\in\mathcal{A}.]$
\end{proposition}

\begin{proof}
Let $v_{x}\in T_{x}\mathbb{R}^{d}.$ Replace $\eta$ by $s\eta$ in Eq.
(\ref{e.B.18}) and then differentiate the result with respect to $s$ to find,
\begin{align*}
\left(  d\pi_{\sigma}\left\langle \eta_{A}\right\rangle \right)  \left(
\left\langle v_{x}\right\rangle \right)   &  =\frac{d}{ds}|_{0}\pi_{\sigma
}\left(  A+s\eta\right)  \left\langle v_{x}\right\rangle \\
&  =\frac{d}{ds}|_{0}\left(  \mathrm{Ad}_{g_{s\eta}^{-1}\left(  x\right)
}\left[  A\left\langle v_{x}\right\rangle +\int_{0}^{1}\mathrm{Ad}%
_{\pt_{1\leftarrow\tau}^{s\eta}\left(  \sigma_{x}\right)  }F^{s\eta}\left(
\dot{\sigma}_{x}\left(  \tau\right)  ,v_{x}\sigma_{\left(  \cdot\right)
}\left(  \tau\right)  \right)  d\tau\right]  \right) \\
&  =\left(  \frac{d}{ds}|_{0}\mathrm{Ad}_{g_{s\eta}^{-1}\left(  x\right)
}\right)  A\left\langle v_{x}\right\rangle +\int_{0}^{1}\frac{d}{ds}%
|_{0}F^{s\eta}\left(  \dot{\sigma}_{x}\left(  \tau\right)  ,v_{x}%
\sigma_{\left(  \cdot\right)  }\left(  \tau\right)  \right)  d\tau\\
&  =\left(  \frac{d}{ds}|_{0}\mathrm{Ad}_{g_{s\eta}^{-1}\left(  x\right)
}\right)  A\left\langle v_{x}\right\rangle +\int_{0}^{1}d\eta\left(
\dot{\sigma}_{x}\left(  \tau\right)  ,v_{x}\sigma_{\left(  \cdot\right)
}\left(  \tau\right)  \right)  d\tau.
\end{align*}
Choosing $x\left(  s\right)  \in\mathbb{R}^{d}$ so that $x^{\prime}\left(
0\right)  =v_{x}$ and using%
\begin{align}
\int_{0}^{1}d\eta\left(  \dot{\sigma}_{x}\left(  \tau\right)  ,\sigma
_{x\left(  s\right)  }^{\prime}\left(  \tau\right)  \right)  d\tau|_{s=0}  &
=\int_{0}^{1}\left[  \frac{d}{d\tau}\eta\left(  \sigma_{x\left(  s\right)
}^{\prime}\left(  \tau\right)  \right)  -\frac{d}{ds}\eta\left(  \dot{\sigma
}_{x\left(  s\right)  }\left(  \tau\right)  \right)  \right]  d\tau
|_{s=0}\nonumber\\
&  =\eta\left\langle v_{x}\right\rangle -\int_{0}^{1}\frac{d}{ds}|_{0}%
\eta\left(  \dot{\sigma}_{x\left(  s\right)  }\left(  \tau\right)  \right)
d\tau\nonumber\\
&  =\eta\left\langle v_{x}\right\rangle -\frac{d}{ds}|_{0}u_{\eta}^{\sigma
}\left(  x\left(  s\right)  \right)  =\left(  \eta-du_{\eta}^{\sigma}\right)
\left\langle v_{x}\right\rangle \label{e.B.24}%
\end{align}
and so%
\begin{equation}
d\pi_{\sigma}\left\langle \eta_{A}\right\rangle =\left(  \frac{d}{ds}%
|_{0}\mathrm{Ad}_{g_{s\eta}^{-1}\left(  x\right)  }\right)  A+\eta-du_{\eta
}^{\sigma}. \label{e.B.25}%
\end{equation}
Moreover, since
\[
\frac{d}{dt}\pt_{t}^{s\eta}\left(  \sigma_{x}\right)  =-s\eta\left\langle
\dot{\sigma}_{x}\left(  t\right)  \right\rangle \pt_{t}^{s\eta}\left(
\sigma_{x}\right)  \text{ with }\pt_{0}^{s\eta}\left(  \sigma_{x}\right)  =I,
\]
we conclude that%
\[
\frac{d}{dt}\frac{d}{ds}|_{0}\pt_{t}^{s\eta}\left(  \sigma_{x}\right)
=\frac{d}{ds}|_{0}\frac{d}{dt}\pt_{t}^{s\eta}\left(  \sigma_{x}\right)
=\frac{d}{ds}|_{0}\left[  -s\eta\left\langle \dot{\sigma}_{x}\left(  t\right)
\right\rangle \pt_{t}^{s\eta}\left(  \sigma_{x}\right)  \right]
=\eta\left\langle \dot{\sigma}_{x}\left(  t\right)  \right\rangle .
\]
Integrating this equation in $t$ shows
\[
\frac{d}{ds}|_{0}g_{s\eta}\left(  x\right)  =\frac{d}{ds}|_{0}\pt_{1}^{s\eta
}\left(  \sigma_{x}\right)  =\int_{0}^{1}\eta\left\langle \dot{\sigma}%
_{x}\left(  \tau\right)  \right\rangle d\tau=u_{\eta}^{\sigma}\left(
x\right)
\]
and hence $\frac{d}{ds}|_{0}\mathrm{Ad}_{g_{s\eta}^{-1}\left(  x\right)
}=-\mathrm{ad}_{u_{\eta}^{\sigma}\left(  x\right)  }$ which combined with Eq.
(\ref{e.B.25}) gives Eq. (\ref{e.B.23}).
\end{proof}

\begin{example}
\label{ex.B.36}Let us work out $u_{\eta}^{\sigma}$ in the special case where
$d=2,$ $\sigma$ is the complete axial homotopy, and $\eta=\eta_{1}dx.$ In this
case,%
\[
u_{\eta}^{\sigma}\left(  x,y\right)  =u_{\eta}^{\sigma}\left(  x,0\right)
=\int_{0}^{1}\eta\left\langle \dot{\sigma}_{\left(  x,0\right)  }\left(
\tau\right)  \right\rangle d\tau=\int_{0}^{x}\eta_{1}\left(  s,0\right)  ds
\]
and therefore%
\[
\eta-du_{\eta}^{\sigma}=\left[  \eta_{1}\left(  x,y\right)  -\eta_{1}\left(
x,0\right)  \right]  dx=\bar{\eta}_{1}\left(  x,0\right)  dx
\]
which agrees with formulas we have used in the body of this paper.
\end{example}

\begin{corollary}
\label{cor.B.37}If $\sigma$ is a follow the leader homotopy, $\Psi$ is a
smooth gauge invariant function on $\mathcal{A},$ and $\eta\in\mathcal{A},$
then%
\[
\left(  \partial_{\eta}\Psi\right)  \left(  A\right)  =\left(  \partial
_{\left[  -\mathrm{ad}_{u_{\eta}^{\sigma}}A+\left(  \eta-du_{\eta}^{\sigma
}\right)  \right]  }\Psi\right)  \left(  A\right)  \text{ }\forall
~A\in\mathcal{A}_{\sigma}.
\]

\end{corollary}

\begin{proof}
By gauge invariance of $\Psi,$ $\Psi\left(  A+s\eta\right)  =\Psi\left(
\pi_{\sigma}\left(  A+s\eta\right)  \right)  $ and therefore using Proposition
\ref{pro.B.35},%
\begin{align*}
\left(  \partial_{\eta}\Psi\right)  \left(  A\right)   &  =\frac{d}{ds}%
|_{0}\Psi\left(  A+s\eta\right)  =\frac{d}{ds}|_{0}\Psi\left(  \pi_{\sigma
}\left(  A+s\eta\right)  \right) \\
&  =\left(  \partial_{\left[  -\mathrm{ad}_{u_{\eta}^{\sigma}}A+\left(
\eta-du_{\eta}^{\sigma}\right)  \right]  }\Psi\right)  \left(  A\right)  .
\end{align*}

\end{proof}

\begin{lemma}
\label{lem.B.38}If $d=2,$ $\sigma$ is a follow the leader homotopy, and
$A\in\mathcal{A}_{\sigma},$ then $F^{A}=dA.$
\end{lemma}

\begin{proof}
The point is that $A\wedge A$ is determined by its value on any two linearly
independent vectors, $\left\{  u_{1},u_{2}\right\}  .$ We may always take
$u_{1}=\dot{\sigma}_{x}\left(  1\right)  $ in which case
\[
A\wedge A\left\langle u_{1},u_{2}\right\rangle =\left[  A\left\langle
u_{1}\right\rangle ,A\left\langle u_{2}\right\rangle \right]  =\left[
A\left\langle \dot{\sigma}_{x}\left(  1\right)  \right\rangle ,A\left\langle
u_{2}\right\rangle \right]  =0.
\]

\end{proof}

\begin{remark}
\label{rem.B.39}If $\sigma$ is a follow the leader homotopy and $g\in
\mathcal{G},$ then $\mathrm{Ad}_{g^{-1}}$ preserves $\mathcal{A}_{\sigma}.$
Indeed if $A\in\mathcal{A}_{\sigma},$ then $\mathrm{Ad}_{g}A\in\mathcal{A}%
_{\sigma}$ since
\[
\left(  \mathrm{Ad}_{g}A\right)  \left\langle \dot{\sigma}_{x}\left(
t\right)  \right\rangle :=\mathrm{Ad}_{g\left(  \sigma_{x}\left(  t\right)
\right)  }\left[  A\left\langle \dot{\sigma}_{x}\left(  t\right)
\right\rangle \right]  =0~\forall x\in\mathbb{R}^{d}\text{ \&~a.e. }%
t\in\left[  0,1\right]  .
\]

\end{remark}

\begin{metapropositions}
\label{mpro.B.40}Let $m$ denote formal Lebesgue measure on $\mathcal{A}$ and
$\sigma$ be a follow the leader homotopy. Then the formal measure, $m_{\sigma
}=m_{v_{\sigma}},$ given by Proposition \ref{pro.B.9} is a Lebesgue measure on
$\mathcal{A}_{\sigma}.$
\end{metapropositions}

\begin{proof}
[Meta-Proof]Since, for $g\in\mathcal{G},$ $\mathrm{Ad}_{g^{-1}}$ acts
orthogonally on $\mathcal{A}$ equipped with the $L^{2}$-norm and hence we
(heuristically) have $\mathrm{Det}\left(  \mathrm{Ad}_{g^{-1}}\right)  =1$ and
so $A\rightarrow A^{g}=\mathrm{Ad}_{g^{-1}}A+g^{-1}dg$ is (formally) an affine
action. Combining this observation with Remark \ref{rem.B.39} allows us to
formally apply Theorem \ref{thm.B.10} in this infinite dimensional context.
\end{proof}

\begin{metacorollary}
\label{mcor.B.41}Let $m$ denote formal Lebesgue measure on $\mathcal{A}$ and
$\sigma$ be a follow the leader homotopy then (recall Definition
\ref{def.B.11})
\[
\dashint_{\mathcal{A}}\Psi\left(  A\right)  dm\left(  A\right)  =\int
_{\mathcal{A}_{\sigma}}\Psi\left(  A\right)  dm_{\sigma}\left(  A\right)
\]
where $m_{\sigma}$ is a Lebesgue measure on $\mathcal{A}_{\sigma}.$
\end{metacorollary}

To apply this last result to the formal $YM$-measures we need the following
simple lemma.

\begin{lemma}
\label{lem.B.42}The function, $\mathcal{A}\ni A\rightarrow\left\Vert
F^{A}\right\Vert $ as described in Eq. (\ref{e.1.4}) is invariant under the
full gauge group.
\end{lemma}

\begin{proof}
From Theorem \ref{thm.A.1}, we know that $F^{A^{g}}=\mathrm{Ad}_{g^{-1}}F^{A}$
and since $\left\vert \cdot\right\vert _{\mathfrak{k}}$ is assumed to be
$\mathrm{Ad}_{K}$-invariant we find, for any $g\in C^{1}\left(  \mathbb{R}%
^{d},K\right)  ,$ then%
\[
\sum_{i<j}\left\vert F^{A^{g}}\left\langle e_{i},e_{j}\right\rangle \left(
x\right)  \right\vert _{\mathfrak{k}}^{2}=\sum_{i<j}\left\vert \mathrm{Ad}%
_{g\left(  x\right)  ^{-1}}F^{A}\left\langle e_{i},e_{j}\right\rangle \left(
x\right)  \right\vert _{\mathfrak{k}}^{2}=\sum_{i<j}\left\vert F^{A}%
\left\langle e_{i},e_{j}\right\rangle \left(  x\right)  \right\vert
_{\mathfrak{k}}^{2}.
\]
Integrating this equation over $\mathbb{R}^{d}$ immediately gives $\left\Vert
F^{A^{g}}\right\Vert ^{2}=\left\Vert F^{A}\right\Vert ^{2}.$
\end{proof}

\begin{definition}
[Formal Yang-Mills Expectations]\label{def.B.43}If $\Psi:\mathcal{A}%
\rightarrow\mathbb{C}$ is a restricted gauge invariant function, we define
\begin{equation}
\left\langle \Psi\right\rangle _{YM}:=\frac{1}{Z_{\sigma}}\int_{\mathcal{A}%
_{\sigma}}\Psi\left(  A\right)  e^{-\frac{1}{2}\left\Vert F^{A}\right\Vert
^{2}}d\tilde{m}_{\sigma}\left(  A\right)  , \label{e.B.26}%
\end{equation}
where $\sigma$ is any follow the leader homotopy, $\tilde{m}_{\sigma}$ is a
formal Lebesgue measure on $\mathcal{A}_{\sigma},$ and (formally)%
\[
Z_{\sigma}:=\int_{\mathcal{A}_{\sigma}}e^{-\frac{1}{2}\left\Vert
F^{A}\right\Vert ^{2}}d\tilde{m}_{\sigma}\left(  A\right)  .
\]

\end{definition}

A few remarks are in order;

\begin{enumerate}
\item The expression in Eq. (\ref{e.B.26}) is formally independent of the
choice of Lebesgue measure on $\mathcal{A}_{\sigma}$ since they all differ by
a multiplicative constant and any such multiplicative constant will also occur
in the normalization constant, $Z_{\sigma}.$

\item The expression in Eq. (\ref{e.B.26}) is formally independent of the
choice of the follow the leader homotopy, $\sigma,$ used in the definition
since by the first remark we may choose $\tilde{m}_{\sigma}=m_{\sigma}$ in
which case
\begin{equation}
\left\langle \Psi\right\rangle _{YM}:=\frac{1}{Z}\dashint_{\mathcal{A}}%
\Psi\left(  A\right)  e^{-\frac{1}{2}\left\Vert F^{A}\right\Vert ^{2}%
}dm\left(  A\right)  \label{e.B.27}%
\end{equation}
with
\[
Z=\dashint_{\mathcal{A}}e^{-\frac{1}{2}\left\Vert F^{A}\right\Vert ^{2}%
}dm\left(  A\right)  .
\]

\end{enumerate}

Our final goal is to show (formally) that $\left\langle \Psi\right\rangle
_{YM_{2}}$ is invariant under area preserving diffeomorphisms.

\subsection{Area preserving diffeomorphisms\label{sec.B.5}}

Let $\sigma$ be a homotopy contracting $\mathbb{R}^{d}$ to $\left\{
0\right\}  .$

\begin{notation}
[Diffeomorphism action on $\mathcal{A}_{\sigma}$]\label{not.B.44}If
$\varphi:\mathbb{R}^{d}\rightarrow\mathbb{R}^{d}$ is a diffeomorphism, let
$\hat{\varphi}_{\sigma}:\mathcal{A}_{\sigma}\rightarrow\mathcal{A}_{\sigma}$
be defined by
\[
\hat{\varphi}_{\sigma}\left(  A\right)  :=\pi_{\sigma}\left(  \varphi^{\ast
}A\right)  =\left(  \varphi^{\ast}A\right)  ^{g_{A}}\text{ for all }%
A\in\mathcal{A}_{\sigma},
\]
where%
\begin{equation}
g_{A}\left(  p\right)  =\pt_{1}^{A}\left(  \varphi\circ\sigma_{p}\right)
\text{ for all }p\in\mathbb{R}^{d}. \label{e.B.28}%
\end{equation}

\end{notation}

\begin{proposition}
[The diffeomorphism action parallel translation]\label{pro.B.45}If
$\varphi:\mathbb{R}^{d}\rightarrow\mathbb{R}^{d}$ is a diffeomorphism,
$A\in\mathcal{A}_{\sigma},$ and $\alpha\in C^{1}\left(  \left[  a,b\right]
,\mathbb{R}^{2}\right)  $ is a path, then
\begin{equation}
\pt^{\hat{\varphi}_{\sigma}\left(  A\right)  }\left(  \alpha\right)
=g_{A}\left(  \alpha\left(  b\right)  \right)  ^{-1}\pt^{A}\left(
\varphi\circ\alpha\right)  g_{A}\left(  \alpha\left(  a\right)  \right)
\label{e.B.29}%
\end{equation}
where $g_{A}$ is as in Eq. (\ref{e.B.28}).
\end{proposition}

\begin{proof}
Using Theorem \ref{thm.A.1} and Proposition \ref{pro.A.3} we have%
\begin{align*}
\pt^{\hat{\varphi}_{\sigma}\left(  A\right)  }\left(  \alpha\right)
=\pt^{\left[  \varphi^{\ast}A\right]  ^{g_{A}}}\left(  \alpha\right)   &
=g_{A}\left(  \alpha\left(  b\right)  \right)  ^{-1}\pt^{\left[  \varphi
^{\ast}A\right]  }\left(  \alpha\right)  g_{A}\left(  \alpha\left(  a\right)
\right) \\
&  =g_{A}\left(  \alpha\left(  b\right)  \right)  ^{-1}\pt^{A}\left(
\varphi\circ\alpha\right)  g_{A}\left(  \alpha\left(  a\right)  \right)  .
\end{align*}

\end{proof}

For the rest of this appendix we now exclusively assume that $d=2$ and further
assume that $\varphi:\mathbb{R}^{2}\rightarrow\mathbb{R}^{2}$ is an
\textbf{area preserving diffeomorphism}.

\begin{definition}
\label{def.B.46}A diffeomorphism, $\varphi:\mathbb{R}^{2}\rightarrow
\mathbb{R}^{2}$ is \textbf{area preserving }provided $\left\vert \det
\varphi^{\prime}\left(  p\right)  \right\vert =1$ for all $p\in\mathbb{R}%
^{2}.$ We further let $\varepsilon\left(  \varphi\right)  =\mathrm{sgn}\left(
\det\varphi^{\prime}\right)  \in\left\{  \pm1\right\}  $ so that $\varphi$ is
orientation preserving if $\varepsilon\left(  \varphi\right)  =1$ and
orientation reversing if $\varepsilon\left(  \varphi\right)  =-1.$
Alternatively stated, a diffeomorphism. $\varphi:\mathbb{R}^{2}\rightarrow
\mathbb{R}^{2},$ is area preserving iff $\varphi^{\ast}\left(  dx\wedge
dy\right)  =\varepsilon\left(  \varphi\right)  dx\wedge dy$ where
$\varepsilon\left(  \varphi\right)  $ is either $1$ or $-1.$ i.e.
\end{definition}

Our final goal of this appendix is to \textquotedblleft
prove\textquotedblright\ the following Meta-Theorem.

\begin{metatheorem}
\label{mtm.B.47}Let $\varphi:\mathbb{R}^{2}\rightarrow\mathbb{R}^{2}$ be an
area preserving diffeomorphism and $\Psi:\mathcal{A}\rightarrow\left[
0,\infty\right]  $ be a function. If either;

\begin{enumerate}
\item $\varphi\left(  0\right)  =0$ and $\Psi$ is a restricted gauge
invariant, or

\item $\Psi$ is invariant under the full gauge group,
\end{enumerate}

then%
\begin{equation}
\left\langle \Psi\circ\varphi^{\ast}\right\rangle _{YM_{2}}=\left\langle
\Psi\right\rangle _{YM_{2}}. \label{e.B.30}%
\end{equation}

\end{metatheorem}

\begin{proof}
[Meta-Proof]This result follows from using either Meta-Theorem \ref{mthm.B.55}
or Meta-Theorem \ref{mthm.B.56} below along with the observation in the next
lemma that $\mathcal{A}\ni A\rightarrow\left\Vert F^{A}\right\Vert $ is
invariant under $\varphi^{\ast}.$
\end{proof}

\begin{lemma}
\label{lem.B.48}If $\varphi:\mathbb{R}^{2}\rightarrow\mathbb{R}^{2}$ is an
area preserving diffeomorphism, then $\left\Vert F^{\varphi^{\ast}%
A}\right\Vert =\left\Vert F^{A}\right\Vert $ for all $A\in\mathcal{A}.$
\end{lemma}

\begin{proof}
Writing $F^{A}=f^{A}dx_{1}\wedge dx^{2}$ we have
\begin{align*}
f^{\varphi^{\ast}A}dx_{1}\wedge dx^{2}  &  =F^{\varphi^{\ast}A}=\varphi^{\ast
}F^{A}=f^{A}\circ\varphi\cdot\varphi^{\ast}\left(  dx_{1}\wedge dx_{2}\right)
\\
&  =\varepsilon\left(  \varphi\right)  f^{A}\circ\varphi\cdot dx_{1}\wedge
dx_{2}%
\end{align*}
where $\varepsilon\left(  \varphi\right)  \in\left\{  \pm1\right\}  $ since
$\varphi$ is area preserving and consequently,
\begin{align*}
\left\Vert F^{\varphi^{\ast}A}\right\Vert ^{2}  &  =\int_{\mathbb{R}^{2}%
}\left\vert f^{\varphi^{\ast}A}\left(  x\right)  \right\vert _{\mathfrak{k}%
}^{2}dx=\int_{\mathbb{R}^{2}}\left\vert f^{A}\left(  \varphi\left(  x\right)
\right)  \right\vert _{\mathfrak{k}}^{2}dx\\
&  =\int_{\mathbb{R}^{2}}\left\vert f^{A}\left(  x\right)  \right\vert
_{\mathfrak{k}}^{2}dx=\left\Vert F^{A}\right\Vert ^{2}.
\end{align*}

\end{proof}

So it now remains to \textquotedblleft prove\textquotedblright\ Meta-Theorem
\ref{mthm.B.55} and Meta-Theorem \ref{mthm.B.56} below. In brief these
theorems assert; if $\varphi:\mathbb{R}^{2}\rightarrow\mathbb{R}^{2}$ is an
area preserving diffeomorphism and $\Psi:\mathcal{A}\rightarrow\left[
0,\infty\right]  $ is a function, then
\begin{equation}
\dashint_{\mathcal{A}}\Psi\left(  \varphi^{\ast}A\right)  dm\left(  A\right)
=\dashint_{\mathcal{A}}\Psi\left(  A\right)  dm\left(  A\right)  ,
\label{e.B.31}%
\end{equation}
provided $\Psi$ is invariant under the full gauge group (Meta-Theorem
\ref{mthm.B.56}) or $\varphi\left(  0\right)  =0$ and $\Psi$ is invariant
under the restricted gauge group (Meta-Theorem \ref{mthm.B.55}). Let us note
that when $\varphi\left(  0\right)  \neq0,$ $\varphi^{\ast}g=g\circ\varphi$
will not be in $\mathcal{G}$ for $g\in\mathcal{G}.$ Nevertheless, if $\Psi$ is
invariant under the full gauge group, then
\[
\Psi\left(  \varphi^{\ast}A^{g}\right)  =\Psi\left(  \left(  \varphi^{\ast
}A\right)  ^{\varphi^{\ast}g}\right)  =\Psi\left(  \varphi^{\ast}A\right)
\text{ }\forall~g\in\mathcal{G}%
\]
and therefore $A\rightarrow\Psi\left(  \varphi^{\ast}A\right)  $ is still a
restricted gauge invariant function on $\mathcal{A}.$

Our \textquotedblleft proof\textquotedblright\ of Eq. (\ref{e.B.31}) will boil
down to formally verifying the hypothesis of Theorem \ref{thm.B.14} in this
infinite dimensional setting. It is worth noting that the results to follow
hold for any $d\in\mathbb{N}$ with $d\geq2$ in the special case where
$\varphi\left(  x\right)  =Rx+b$ with $R$ be a rotation on $\mathbb{R}^{d}$
and $b\in\mathbb{R}^{d}.$ We now begin \textquotedblleft
verifying\textquotedblright\ the hypothesis of Theorem \ref{thm.B.14} in this
infinite dimensional gauge theory context.

\begin{metalemma}
\label{mlem.B.49}The restricted gauge group, $\mathcal{G},$ is formally unimodular.
\end{metalemma}

\begin{proof}
[Meta-Proof]The Lie algebra $\left(  \operatorname*{Lie}\left(  \mathcal{G}%
\right)  \right)  $ of $\mathcal{G}$ consists of functions, $\xi
:\mathbb{R}^{2}\rightarrow\mathfrak{k}$ with $\xi\left(  0\right)  =0.$ Let
\[
\left\langle \xi,\eta\right\rangle _{\operatorname*{Lie}\left(  \mathcal{G}%
\right)  }:=\int_{\mathbb{R}^{2}}\left\langle \xi\left(  x\right)
,\eta\left(  x\right)  \right\rangle _{\mathfrak{k}}dx
\]
where $\left\langle \cdot,\cdot\right\rangle _{\mathfrak{k}}$ is an
$\mathrm{Ad}_{K}$ -- invariant inner product on $\mathfrak{k.}$ If
$g\in\mathcal{G},$ then%
\[
\left\Vert \mathrm{Ad}_{g}\xi\right\Vert _{\operatorname*{Lie}\left(
\mathcal{G}\right)  }^{2}=\int_{\mathbb{R}^{2}}\left\vert \mathrm{Ad}%
_{g\left(  x\right)  }\xi\left(  x\right)  \right\vert _{\mathfrak{k}}%
^{2}dx=\int_{\mathbb{R}^{2}}\left\vert \xi\left(  x\right)  \right\vert
_{\mathfrak{k}}^{2}dx=\left\Vert \xi\right\Vert _{\operatorname*{Lie}\left(
\mathcal{G}\right)  }^{2}%
\]
so that $\mathrm{Ad}_{g}$ acts as orthogonal transformation and therefore
$\Delta_{\mathcal{G}}\left(  g\right)  =\left\vert \det\mathrm{Ad}_{g^{-1}%
}\right\vert =1.$ Alternatively, extended $\left\langle \cdot,\cdot
\right\rangle _{\operatorname*{Lie}\left(  \mathcal{G}\right)  }$ to a left
invariant Riemannian metric on $T\mathcal{G}$ and note that the fact that
$\mathrm{Ad}_{g}$ is an isometry for $\left\langle \cdot,\cdot\right\rangle
_{\operatorname*{Lie}\left(  \mathcal{G}\right)  }$ implies this Riemannian
metric is also right invariant. Thus the Riemannian volume measure associated
to this Riemannian metric is both right and left invariant and so this measure
is both a left and a right invariant Haar measure.
\end{proof}

\begin{lemma}
\label{lem.B.50}If $\varphi:\mathbb{R}^{2}\rightarrow\mathbb{R}^{2}$ is an
area preserving diffeomorphism such that $\varphi\left(  0\right)  =0,$ then
$\gamma:\mathcal{G}\rightarrow\mathcal{G}$ defined by $\gamma\left(  g\right)
=\varphi^{\ast}g$ is a group isomorphism\footnote{The assumption
$\varphi\left(  0\right)  =0$ is needed to guarantee that $\varphi^{\ast}%
g\in\mathcal{G}$ for every $g\in\mathcal{G}.$} such that%
\[
\varphi^{\ast}\left(  A\cdot g\right)  =\varphi^{\ast}A\cdot\gamma\left(
g\right)  \text{ }\forall~A\in\mathcal{A}\text{ and }g\in\mathcal{G}.
\]

\end{lemma}

\begin{proof}
This result is a consequence of the following elementary computation;%
\begin{align*}
\varphi^{\ast}\left(  A\cdot g\right)   &  =\varphi^{\ast}\left(
A^{g}\right)  =\varphi^{\ast}\left[  \mathrm{Ad}_{g^{-1}}A+g^{-1}dg\right] \\
&  =\left[  \mathrm{Ad}_{\left(  g\circ\varphi\right)  ^{-1}}\varphi^{\ast
}A+\left(  g\circ\varphi\right)  ^{-1}d\left(  g\circ\varphi\right)  \right]
\\
&  =\varphi^{\ast}A\cdot\varphi^{\ast}g.
\end{align*}

\end{proof}

\begin{metalemma}
\label{mlem.B.51}Let $\varphi$ and $\gamma$ be as in Lemma \ref{lem.B.50}.
Then $\gamma$ preserves Haar measure on $\mathcal{G}$ and hence $c_{\gamma}=1$
where $c_{\gamma}$ was defined in Eq. (\ref{e.B.11}).
\end{metalemma}

\begin{proof}
[Meta-Proof]For $g\in\mathcal{G}$ and $\xi\in\operatorname*{Lie}\left(
\mathcal{G}\right)  $ let $\tilde{\xi}\left(  g\right)  =L_{g\ast}\xi$ be the
left invariant extension of $\xi.$ Then%
\[
\gamma_{\ast}\tilde{\xi}\left(  g\right)  =\frac{d}{dt}|_{0}\gamma\left(
ge^{t\xi}\right)  =\frac{d}{dt}|_{0}\gamma\left(  g\right)  \gamma\left(
e^{t\xi}\right)  =L_{\gamma\left(  g\right)  \ast}\gamma_{\ast}\xi
=L_{\gamma\left(  g\right)  \ast}\left[  \xi\circ\varphi\right]  .
\]
By construction of the Riemannian metric on $\mathcal{G},$ $L_{\gamma\left(
g\right)  \ast}$ is an isometry and therefore (using $\varphi$ is area
preserving) we find%
\begin{align*}
\left\langle \gamma_{\ast}\tilde{\xi}\left(  g\right)  ,\gamma_{\ast}%
\tilde{\xi}\left(  g\right)  \right\rangle _{T_{\gamma\left(  g\right)
}\mathcal{G}}  &  =\left\langle \xi\circ\varphi,\xi\circ\varphi\right\rangle
_{\operatorname*{Lie}\left(  \mathcal{G}\right)  }\\
&  =\int_{\mathbb{R}^{2}}\left\vert \xi\left(  \varphi\left(  x\right)
\right)  \right\vert _{\mathfrak{k}}^{2}dx=\int_{\mathbb{R}^{2}}\left\vert
\xi\left(  x\right)  \right\vert _{\mathfrak{k}}^{2}dx=\left\langle \tilde
{\xi}\left(  g\right)  ,\tilde{\xi}\left(  g\right)  \right\rangle
_{T_{g}\mathcal{G}}.
\end{align*}
This shows $\gamma:\mathcal{G}\rightarrow\mathcal{G}$ is a Riemannian isometry
and hence preserves the (fictitious) Riemannian volume measure on
$\mathcal{G}.$ As this volume measure is precisely Haar measure the
\textquotedblleft proof\textquotedblright\ is complete.
\end{proof}

The last item we need to verify is that $\varphi^{\ast}:\mathcal{A}%
\rightarrow\mathcal{A}$ preserves Lebesgue measure when $\varphi$ is an area
preserving diffeomorphisms on $\mathbb{R}^{2}.$ To do this we will make use of
the following meta-lemma.

\begin{metalemma}
\label{lem.B.52}Let $V$ be a finite dimensional inner product space and
$U:\mathbb{R}^{2}\rightarrow\operatorname*{End}\left(  V\right)  $ be a
function such that $\det U\left(  x\right)  =1$ for all $x\in\mathbb{R}^{2}.$
If $M_{U}:L^{2}\left(  \mathbb{R}^{2};V\right)  \rightarrow L^{2}\left(
\mathbb{R}^{2};V\right)  $ is the operation of multiplication by $U,$ then
$\,\det M_{U}=1$ or more usefully stated; the map, $f\rightarrow Uf,$ leaves
Lebesgue measure invariant.
\end{metalemma}

\begin{proof}
[Meta-Proof]Here we suppose that $U\left(  x\right)  =U_{1}\left(  x\right)  $
where $\left\{  U_{t}\left(  x\right)  \right\}  _{t\in\left[  0,1\right]  }$
is a one parameter ($C^{1}$ in $t)$ family of functions in $SL\left(
V\right)  $ with $U_{0}\left(  x\right)  =I$ for all $x.$ Further let
$\alpha_{t}\left(  x\right)  :=U_{t}\left(  x\right)  ^{-1}\dot{U}_{t}\left(
x\right)  $ so that $\dot{U}_{t}\left(  x\right)  =\alpha_{t}\left(  x\right)
U_{t}\left(  x\right)  $ with $U_{0}\left(  x\right)  =I$ and
$\operatorname{tr}\left(  \alpha_{t}\left(  x\right)  \right)  =0.$ We then
formally should have,%
\begin{equation}
\frac{d}{dt}\det\left[  M_{U_{t}}\right]  =\det\left[  M_{U_{t}}\right]
\operatorname{Tr}\left[  M_{U_{t}^{-1}}M_{\dot{U}_{t}}\right]  =\det\left[
M_{U_{t}}\right]  \operatorname{Tr}\left[  M_{\alpha_{t}}\right]
\label{e.B.32}%
\end{equation}
where $\operatorname{Tr}$ is the infinite dimensional trace on $L^{2}\left(
\mathbb{R}^{d},V\right)  .$ To evaluate the trace, let $\left\{
u_{m}\right\}  _{m=1}^{\infty}$ be an orthonormal basis for $L^{2}\left(
\mathbb{R}^{d},\mathbb{R}\right)  $ and $\left\{  e_{i}\right\}  _{i=1}^{\dim
V}$ be an orthonormal basis for $V$ relative to some fixed inner product on
$V.$ Then $\left\{  u_{m}\cdot e_{i}:m\in\mathbb{N}\text{ \&~}1\leq i\leq\dim
V\right\}  $ is an orthonormal basis for $L^{2}\left(  \mathbb{R}%
^{d},V\right)  $ and hence it is reasonable to compute $\operatorname{Tr}%
\left(  M_{\alpha_{t}}\right)  $ as,
\begin{align*}
\operatorname{Tr}\left(  M_{\alpha_{t}}\right)   &  =\sum_{m=1}^{\infty}%
\sum_{i=1}^{d}\left\langle M_{\alpha_{t}}u_{m}\cdot e_{i},u_{m}\cdot
e_{i}\right\rangle _{L^{2}\left(  \mathbb{R}^{d},V\right)  }\\
&  =\sum_{m=1}^{\infty}\sum_{i=1}^{d}\int_{\mathbb{R}^{2}}\left\langle
\alpha_{t}\left(  x\right)  e_{i},e_{i}\right\rangle u_{m}^{2}\left(
x\right)  dx\\
&  =\sum_{m=1}^{\infty}\int_{\mathbb{R}^{2}}\operatorname{tr}\left[
\alpha_{t}\left(  x\right)  \right]  u_{m}^{2}\left(  x\right)  dx=\sum
_{m=1}^{\infty}\int_{\mathbb{R}^{2}}0\cdot u_{m}^{2}\left(  x\right)  dx=0.
\end{align*}
Thus we have shown $\det\left[  M_{U_{t}}\right]  $ is constant in $t$ and so%
\[
\det M_{U}=\det\left[  M_{U_{1}}\right]  =\det\left[  M_{U_{0}}\right]  =\det
I_{L^{2}\left(  \mathbb{R}^{2};V\right)  }=1.
\]

\end{proof}

\begin{remark}
\label{rem.B.53}The computation of the trace of $M_{\alpha_{t}}$ above is
certainly not rigorous as $M_{\alpha_{t}}$ is \textbf{not }a trace class operator.

\end{remark}

\begin{metapropositions}
\label{mpro.B.54}If $\varphi:\mathbb{R}^{2}\rightarrow\mathbb{R}^{2}$ is an
area preserving diffeomorphism, then the induced map, $\mathcal{A\ni
}A\rightarrow\varphi^{\ast}A\in\mathcal{A}$ formally preserves Lebesgue
measure on $\mathcal{A}.$
\end{metapropositions}

\begin{proof}
[Meta-Proof]If $A=A_{1}dx_{1}+A_{2}dx_{2},$ then%
\begin{align*}
\varphi^{\ast}A  &  =A_{1}\circ\varphi d\left[  x_{1}\circ\varphi\right]
+A_{2}\circ\varphi d\left[  x_{2}\circ\varphi\right] \\
&  =A_{1}\circ\varphi\left[  \partial_{1}\varphi_{1}dx_{1}+\partial_{2}%
\varphi_{1}dx_{2}\right]  +A_{2}\circ\varphi\left[  \partial_{1}\varphi
_{2}dx_{1}+\partial_{2}\varphi_{2}dx_{2}\right] \\
&  =\left(  \left[  A_{1}\circ\varphi\right]  \partial_{1}\varphi_{1}+\left[
A_{2}\circ\varphi\right]  \partial_{1}\varphi_{2}\right)  dx_{1}+\left(
\left[  A_{1}\circ\varphi\right]  \partial_{2}\varphi_{1}+\left[  A_{2}%
\circ\varphi\right]  \partial_{2}\varphi_{2}\right)  dx_{2}.
\end{align*}
Thus identifying $A$ with $\left[
\begin{array}
[c]{cc}%
A_{1} & A_{2}%
\end{array}
\right]  ^{\operatorname{tr}},$ the transformation $\mathcal{A\ni}%
A\rightarrow\varphi^{\ast}A\in\mathcal{A}$ is the composition of the linear
transformation%
\begin{equation}
\left[
\begin{array}
[c]{c}%
A_{1}\\
A_{2}%
\end{array}
\right]  \rightarrow\left[
\begin{array}
[c]{c}%
A_{1}\\
A_{2}%
\end{array}
\right]  \circ\varphi\label{e.B.33}%
\end{equation}
followed by applying the linear transformation, $M_{U},$ where
\[
U\left(  x,y\right)  :=\left[
\begin{array}
[c]{cc}%
\left(  \partial_{1}\varphi_{1}\right)  \left(  x,y\right)  & \left(
\partial_{1}\varphi_{2}\right)  \left(  \left(  x,y\right)  \right) \\
\left(  \partial_{2}\varphi_{1}\right)  \left(  x,y\right)  & \left(
\partial_{2}\varphi_{2}\right)  \left(  \left(  x,y\right)  \right)
\end{array}
\right]  .
\]
The assumption that $\varphi$ is area preserving is equivalent to $\det
U\left(  x,y\right)  =1$ and hence by Meta-Lemma \ref{lem.B.52},
$\operatorname{Det}\left[  M_{U}\right]  =1.$ The assumption that $\varphi$ is
area preserving also implies that transformation in Eq. (\ref{e.B.33}) is an
isometry on $L^{2}\left(  \mathbb{R}^{2};\mathfrak{k}^{2}\right)  $ and so
again (formally) preserves Lebesgue measure. As $A\rightarrow\varphi^{\ast}A$
is a composition of two Lebesgue measure preserving maps it also (formally)
preserves Lebesgue measure on $\mathcal{A}.$
\end{proof}

\begin{metatheorem}
\label{mthm.B.55}Let $d=2.$ If $\varphi$ is an area preserving diffeomorphism
such that $\varphi\left(  0\right)  =0$ and $\Psi$ is a restricted gauge
invariant function, then formally Eq. (\ref{e.B.31}) holds.
\end{metatheorem}

\begin{proof}
This result heuristically follows from Theorem \ref{thm.B.14} whose hypothesis
have been heuristically verified in Meta-Lemmas \ref{mlem.B.49} and
\ref{mlem.B.51} and Meta-Proposition \ref{mpro.B.54}.
\end{proof}

\begin{metatheorem}
\label{mthm.B.56}Let $d=2.$ If $\varphi$ is an area preserving diffeomorphism
and $\Psi$ is invariant under the full gauge group, then (formally) Eq.
(\ref{e.B.31}) still holds.
\end{metatheorem}

\begin{proof}
If $\varphi\left(  0\right)  =0$ the result follows from Meta-Theorem
\ref{mthm.B.55}. When $\varphi\left(  0\right)  \neq0$ we let $\eta\left(
\cdot\right)  :=\varphi\left(  \cdot\right)  -\varphi\left(  0\right)  .$ Then
$\varphi=\eta+\varphi\left(  0\right)  $ which shows $\varphi$ is the
composition of an area preserving diffeomorphism $\left(  \eta\right)  $
fixing $0\in\mathbb{R}^{2}$ followed by translation by $\varphi\left(
0\right)  .$ Thus to finish the proof it suffices to consider the special case
where $\varphi\left(  x\right)  =x+b$ for some vector $b\in\mathbb{R}^{2}.$ We
can further reduce the problem to the case where $b\in\mathbb{R}e_{1}$ where
$e_{1}=\left(  0,1\right)  .$ To verify this, let $R$ be the $2\times2$
rotation matrix such that $R^{-1}b=v\in\mathbb{R}e_{1}$ and then write
$\varphi$ as $\varphi\left(  x\right)  =R\left[  R^{-1}x+v\right]  .$ This
shows that $\varphi$ is a composition of two area preserving diffeomorphism,
$R$ and $R^{-1},$ which fix $0\in\mathbb{R}^{2}$ along with a translation by
$v\in\mathbb{R}e_{1}.$

Owing to the above reductions we now assume that $\varphi\left(  x\right)
=x+\lambda e_{1}$ for some $\lambda\in\mathbb{R}.$ Let $\sigma$ be the
complete axial homotopy in Example \ref{n.B.20} in which case $A\in
\mathcal{A}_{\sigma}$ iff $A=A_{1}dx_{1}\in\mathcal{A}_{\sigma}$ with
$A_{1}\left(  x_{1},0\right)  =0$ for all $x_{1}\in\mathbb{R}.$ Since, for
$A\in\mathcal{A}_{\sigma}$
\[
\varphi^{\ast}\left(  A_{1}dx_{1}\right)  =A_{1}\circ\varphi dx_{1}%
=A_{1}\left(  \cdot+\lambda,\cdot\right)  dx_{1}%
\]
it follows that $\varphi^{\ast}$ preserves $\mathcal{A}_{\sigma}.$ Moreover,
as $\varphi^{\ast}$ acts orthogonally on $\mathcal{A}_{\sigma}$ equipped with
the $L^{2}\left(  \mathbb{R}^{2},\mathfrak{k}\right)  $-inner product it is
reasonable to (formally) assert that $\varphi^{\ast}$ leave \textquotedblleft
Lebesgue measure\textquotedblright\ on $\mathcal{A}_{\sigma}$ invariant. As
$m_{\sigma}$ is formally a Lebesgue measure on $\mathcal{A}_{\sigma}$ by
Meta-Proposition \ref{mpro.B.40}, we conclude that%
\begin{align*}
\dashint_{\mathcal{A}}\Psi\left(  \varphi^{\ast}A\right)  dm\left(  A\right)
&  =\int_{\mathcal{A}_{\sigma}}\Psi\left(  \varphi^{\ast}A\right)  dm_{\sigma
}\left(  A\right) \\
&  =\int_{\mathcal{A}_{\sigma}}\Psi\left(  A\right)  dm_{\sigma}\left(
A\right)  =\dashint_{\mathcal{A}}\Psi\left(  A\right)  dm\left(  A\right)  .
\end{align*}

\end{proof}

\providecommand{\bysame}{\leavevmode\hbox to3em{\hrulefill}\thinspace}
\providecommand{\MR}{\relax\ifhmode\unskip\space\fi MR }
\providecommand{\MRhref}[2]{%
  \href{http://www.ams.org/mathscinet-getitem?mr=#1}{#2}
}
\providecommand{\href}[2]{#2}


\end{document}